\documentclass[a4paper,11pt]{article}

\input{doct_preamble.tex}

\begin{document}

\title{Parameterized Complexity and Approximability of Directed Odd Cycle Transversal
\thanks{ Supported by {\em Pareto-Optimal Parameterized Algorithms}, ERC Starting Grant 715744 and {\em Parameterized Approximation}, ERC Starting Grant 306992.
%
%
%
          M. S. Ramanujan also acknowledges support from {\em BeHard},  Bergen Research Foundation and {\em X-Tract}, Austrian Science Fund (FWF, project P26696).
        }
        }

\author{
 Daniel Lokshtanov\thanks{University of Bergen, Bergen, Norway. \texttt{daniello@ii.uib.no}}
 \and M. S. Ramanujan\thanks{Algorithms and Complexity Group, TU Wien, Vienna, Austria.     \texttt{ramanujan@ac.tuwien.ac.at}} 
 \and  Saket Saurabh\addtocounter{footnote}{-2}\footnotemark\thanks{The Institute of Mathematical Sciences, HBNI, Chennai, India. \texttt{saket@imsc.res.in}}
 \and Meirav Zehavi\thanks{University of Bergen, Bergen, Norway. \texttt{meirav.zehavi@ii.uib.no}} 
}

\maketitle

\thispagestyle{empty}

\begin{abstract} 

A {\em directed odd cycle transversal} of a directed graph (digraph) $D$ is a vertex set $S$ that intersects every \emph{odd directed cycle} of $D$. In the {\sc Directed Odd Cycle Transversal (DOCT)} problem, the input consists of a digraph $D$ and an integer $k$. The objective is to determine whether there exists a directed odd cycle transversal of $D$ of size at most $k$.
%
%
%
%
%
%
%
In this paper, we settle the parameterized complexity of \doct when parameterized by the solution size $k$
by showing that 
\doct does not admit an algorithm with running time $f(k)n^{\Oh(1)}$ unless $\FPT = \W[1]$. On the positive side, we give a factor $2$ fixed parameter tractable (\FPT) approximation algorithm for the problem. More precisely, our algorithm takes as input $D$ and $k$, runs in time $2^{\Oh(k^2)}n^{\Oh(1)}$, and either concludes that $D$ does not have a directed odd cycle transversal of size at most $k$, or produces a solution of size at most $2k$.
Finally, we provide evidence that there exists $\epsilon > 0$ such that \doct does not admit a factor $(1+\epsilon)$ \FPT-approximation algorithm.









\end{abstract}

\newpage
\pagestyle{plain}
\setcounter{page}{1}

\section{Introduction}



A {\em directed odd cycle transversal} of a digraph $D$ is a set $S$ of vertices of $D$ such that deleting $S$ from $D$ results in a graph without any directed odd cycles. 
In the NP-complete~(\cite{Karp72}, see Footnote~\ref{ftn:second})  \doctfull (\doct) problem, the input consists of a digraph $D$ on $n$ vertices and an integer $k$, and the task is to determine whether $D$ has a directed odd cycle transversal of size at most $k$.
%
%
%
%
%
DOCT generalizes several well studied problems such as {\sc Odd Cycle Transversal} (OCT) on undirected graphs~\cite{AgarwalCMM05,ChoiNR89,KhotV15,ReedSV04}, {\sc Directed Feedback Vertex Set} (DFVS)~\cite{ChenLLOR08,EvenNSS98,GuruswamiL16,Karp72}, and {\sc Directed Subset Feedback Vertex Set}~\cite{Chitnis:2012DSFVS,EvenNSS98}. In OCT, the input consists of an undirected graph $G$ and integer $k$, and the task is to determine whether there exists a subset $S$ of vertices such that $G - S$ is bipartite.\footnote{OCT reduces to \doct by replacing every edge by two arcs, one in each direction.} In DFVS, the input consists of a digraph $D$ and integer $k$, and the task is to determine whether there exists a subset $S$ of vertices such that $D - S$ is a directed acyclic graph.\footnote{\label{ftn:second}DFVS reduces to \doct by adding for every arc $uv$ of $D$ a new vertex $x$ as well as the arcs $ux$ and~$xv$.}


The existence of fixed parameter tractable (\FPT) algorithms for OCT and DFVS were considered to be major open problems in parameterized complexity, until \FPT algorithms were found for OCT in 2003 by Reed et al.~\cite{ReedSV04}, and for DFVS in 2007 by Chen et al.~\cite{ChenLLOR08}. The algorithms for these two problems have had significant influence on the development of the field, resulting in proliferation of techniques such as {\em iterative compression} and {\em important separators}~\cite{CyganFKLMPPS15,DBLP:series/txcs/DowneyF13}.
Once both OCT and DFVS were shown to be \FPT, \doct immediately became the next natural target. The parameterized complexity of \doct was explicitly stated as an open problem~\cite{DemGMS07} for the first time in 2007, immediately after the announcement of an \FPT algorithm for DFVS. Since then the problem has been re-stated several times~\cite{chitnis2014directed,ChitnisH16,Marx12wHAT,Marx17Talk}. In this paper, we settle the parameterized complexity of \doct, by showing that the problem is \W[1]-hard. Our hardness proof also gives a near-tight running time lower bound for \doct assuming the Exponential Time Hypothesis (\ETH). In particular, we prove the following.

\begin{restatable}{theorem}{wHardTheorem}\label{thm:maindoct}\label{thm:doctHard1}  
\doct is \W[1]-hard. Furthermore, assuming the \ETH there is no algorithm for \doct with running time $f(k)n^{o(k/\log k)}$.
\end{restatable}

On the one hand, 
Theorem~\ref{thm:maindoct} shows that \doct is intractable from the perspective of parameterized complexity. On the other hand, the problem is known not to admit a constant factor approximation algorithm running in polynomial time, assuming the Unique Games Conjecture~\cite{KhotV15}. Hence, 
the next natural question is whether one could get a constant factor approximation algorithm in \FPT time. Our second result is an affirmative answer to this question.

\begin{restatable}{theorem}{doctappx}\label{thm:doctappx} 
\doct admits a  $2^{\Oh(k^2)}n^{\Oh(1)}$ time \FPT-approximation algorithm with approximation ratio $2$.
\end{restatable}

In fact, Theorem~\ref{thm:doctappx} follows as a corollary from a stronger result for a ``labeled digraph problem'' which we introduce.
We show that this problem subsumes {\doct} as well as the {\sc Node Unique Label Cover} problem~\cite{ChitnisCHPP16,IwataWY16,LokshtanovRSULC16}, and design an {\FPT} approximation algorithm that works even for this more general problem.  

In light of Theorem~\ref{thm:doctappx} the next natural question is whether the approximation factor can be made arbitrarily close to $1$. Our final contribution is to provide evidence that there exists an $\epsilon > 0$ such that \doct does not admit a $(1+\epsilon)$ \FPT-approximation algorithm. In particular, the proof of Theorem~\ref{thm:maindoct} can be thought of as a parameterized reduction from the {\sc Binary Constraint Satisfaction} (BCSP) problem, informally defined as follows. The input consists of two integers $n$ and $k$ specifying that there are $k$ variables, $x_1, \ldots, x_k$, each variable $x_i$ taking a value from $\{1,  \ldots, n\}$, together with a list of {\em constraints}.  Each constraint specifies two variables, $x_i$ and $x_j$, together with a list $L$ of all legal pairs of values that $x_i$ and $x_j$ may take simultaneously. An assignment of values to the variables satisfies the constraint if $(x_i, x_j) \in L$. The task is to find an assignment that satisfies all constraints. It is well known (see e.g.~\cite{Marx07}) that BCSP parameterized by the number of variables $k$ is \W[1]-complete. 
We conjecture that not only is it \W[1]-hard to find a satisfying assignment to a BCSP instance if there is one, but it is also  
\W[1]-hard to distinguish between instances that have a satisfying assignment from instances where every assignment violates at least an $\epsilon$ fraction of the constraints. Formally, for every $\epsilon > 0$, we define the promise problem $\epsilon$-{\sc gap-BCSP}, as BCSP where the input instance is promised to either be satisfiable, or have the property that every assignment violates at least an $\epsilon$ fraction of the constraints. The task is to determine whether the input instance is satisfiable or not. 

\begin{hypothesis}[{\bf Parameterized Inapproximability Hypothesis (\PIH)}]\label{PIH}\label{conjecture}
There exists an $\epsilon>0$ such that $\epsilon$-{\sc gap-BCSP} is \W[1]-hard.
\end{hypothesis}

We remark that for purposes of showing hardness of approximation, we could just as well have conjectured that there exists an $\epsilon>0$ such that there is no $f(k)n^{\Oh(1)}$ time algorithm for $\epsilon$-{\sc gap-BCSP}. However, we strongly believe that the \PIH is true as stated---indeed, we should hardly claim this conjecture as our own, as quite a few researchers in parameterized complexity have stated this conjecture as a natural formulation of a PCP-theorem in the context of parameterized inapproximability. Our final result is that assuming the \PIH, there exists $\epsilon > 0$ such that \doct does not admit an \FPT-approximation algorithm with ratio $1+\epsilon$.

\begin{restatable}{theorem}{inapproxTheorem}\label{thm:doctinappx}\label{thm:doctHard2} 
Assuming the \PIH and \FPT $\neq$ \W[1], there exists $\epsilon>0$ such that \doct does not admit an \FPT-approximation algorithm with approximation ratio $1+\epsilon$.
\end{restatable}

\noindent
{\bf Arc-Directed Odd Cycle Transversal.}
We remark that easy reductions transfer all of our results to {\sc Arc}-\doct, the ``arc'' version of \doct where the goal is to remove at most $k$ arcs such that the resulting graph does not have any directed odd cycles.  To transfer the hardness results we need to reduce \doct to {\sc Arc}-\doct. For this purpose, it is sufficient to subdivide every arc, and then split every original vertex $u$ of the input digraph into two vertices, $u_{in}$ and $u_{out}$, such that all arcs leading into $u$ lead into $u_{in}$ instead, all arcs leading out of $u$ lead out of $u_{out}$ instead, and adding the arc $u_{in}u_{out}$.  To transfer the algorithmic results from \doct to {\sc Arc}-\doct, we need to reduce {\sc Arc}-\doct to \doct. This is achieved by subdividing every arc twice, and then making each original vertex undeletable by adding $k+1$ copies of it.

\subsection*{Our Methods}

{\bf \W[1]-hardness.}
The starting point for both our hardness results as well as our approximation algorithm is a failed attempt at obtaining an \FPT algorithm. The root of this attempt was the \FPT algorithm for DFVS by Chen et al.~\cite{ChenLLOR08}. The key concept in this algorithm is the notion of {\em important separators}, defined by Marx~\cite{Marx06}. Given a digraph $D$ and two vertices $u$ and $v$, a $u$-$v$-{\em separator} is a vertex set $S \subseteq V(D) \setminus \{u, v\}$ such that there is no directed path from $u$ to $v$ in $D - S$. A $u$-$v$-separator $S$ is called a {\em minimal} $u$-$v$-{\em separator} if no proper subset of $S$ is also a $u$-$v$-separator.

Given a vertex set $S$ such that $u$ is not in $S$, we define the {\em reach of $u$ in $D - S$} as the set $R_D(u,S)$ of vertices reachable from $u$ by a directed path in $D - S$. We can now define a partial order on the set of minimal $u$-$v$ separators as follows. Given two minimal $u$-$v$ separators $S_1$ and $S_2$, we say that $S_1$ is ``at least as good as''  $S_2$ if $|S_1| \leq |S_2|$ and $R_D(u, S_2) \subseteq R_D(u, S_1)$. In plain words, $S_1$ ``costs less'' than  $S_2$ in terms of the number of vertices deleted, and $S_1$ ``is pushed further towards $v$'' than $S_2$ is. A minimal $u$-$v$ separator $S$ is an {\em important} $u$-$v$-{\em separator} if no minimal $u$-$v$-separators other than $S$ is at least as good as $S$. The key insight behind the algorithm for DFVS by Chen et al.~\cite{ChenLLOR08}, as well as algorithms for several other parameterized problems~\cite{ChitnisHM13,Chitnis:2012DSFVS,CyganPPW13,Kratsch:2012MCDAG,LokshtanovM13,LokshtanovR12,LokshtanovRS15,LokshtanovRSULC16,MarxR14}, is that for every $k$, the number of important $u$-$v$-separators of size at most $k$ is at most $4^k$~\cite{ChenLL09}. We refer the reader to the textbook by Cygan et al.~\cite{CyganFKLMPPS15} for a more thorough exposition of important separators.

Applying the initial steps of the DFVS algorithm to \doct (i.e. the methods of iterative compression, and guessing an order on an undeletable solution), one naturally arrives at an extension of the notion of important separators. Let us define the {\em cleaning cost} of a minimal $u$-$v$ separator $S$ as $\doctfunc{D[R_D(u, S)]}$, where $\doctfunc{D}$ is the minimum size of a directed odd cycle transversal of $D$. Then, we define a new partial order on minimal $u$-$v$ separators. Here, given two minimal $u$-$v$ separators $S_1$ and $S_2$, we say that $S_1$ is ``at least as good as''  $S_2$ if $|S_1| \leq |S_2|$, $R_D(u, S_2) \subseteq R_D(u, S_1)$, and the cleaning cost of $S_1$ is at most the cleaning cost of $S_2$.
In other words, $S_1$ costs less than $S_2$, $S_1$ is pushed further towards $v$ than $S_2$, and ``cleaning up'' the reach of $u$ in $G - S_1$ does not cost more than cleaning up the reach of $u$ in $G - S_2$. We say that a minimal $u$-$v$ separator $S$ is a {\em \doct-important} $u$-$v$-{\em separator} if no minimal $u$-$v$-separators other than $S$ are at least as good as $S$ with respect to this new partial order. 

For every digraph $D$, vertices $u$ and $v$ and integer $k$, we know that there are at most $4^k$ important $u$-$v$ separators of size at most $k$. For the purposes of an \FPT algorithm for \doct, the pivotal question becomes whether the number of \doct-important $u$-$v$-separators of size at most $k_1$ and cleaning cost at most $k_2$ can be upper bounded by a function of $k_1$ and $k_2$ only, or if there exist families of graphs where the number of \doct-important $u$-$v$-separators of size at most $k_1$ and cleaning cost at most $k_2$ grows with the size of the graphs. Indeed, a constructive upper bound on $f(k_1, k_2)$, the number of \doct-important $u$-$v$-separators of size at most $k_1$ and cleaning cost at most $k_2$, would have implied an \FPT algorithm for \doct. 

We managed to prove that there exists a function $f$ such that the number of \doct-important $u$-$v$-separators of size at most $k$ and cleaning cost $0$ is at most $f(k)$. Emboldened by this proof, we attempted to similarly upper bound the number of \doct-important $u$-$v$-separators of size at most $k$ and cleaning cost $1$. At this point, we discovered the {\em clock gadgets} (see Section~\ref{sec:clockBasic}), which are graphs where the number of \doct-important $u$-$v$-separators of size at most $2$ and cleaning cost $1$ is $\Omega(n)$.

A clock gadget essentially lets us encode (in the language of \doct) the choice of one element out of a domain of size $n$, without it being clear a priori which element(s) should be the best one(s) to select. For many problems, once one has such a selection gadget it is easy to prove \W[1]-hardness by reducing from BCSP (or, equivalently, from {\sc Multicolored Clique}). However, we were able to show that on graphs consisting only of clock gadgets glued together in the most natural way, \doct is in fact \FPT{}\footnote{Because this is such a specialized graph class, we did not include a proof of this fact in the paper.}. In particular, clocks do not provide a general way of synchronizing the choices of different elements, making it difficult to encode the constraints of BCSP using \doct. We were able to engineer such a synchronization gadget by a non-trivial modification of the ``grid gadget'' used by Pilipczuk and Wahlstr\"{o}m~\cite{DBLP:conf/soda/PilipczukW16} to show \W[1]-hardness of {\sc Directed Multicut} with four terminal pairs. At this point one can complete a reduction from BCSP using clocks to encode the selection of a value for each variable and using synchronization gadgets to encode the constraints of the BCSP instance.

\smallskip
\noindent
{\bf \FPT-Approximation.}
The hardness of \doct comes from the fact that \doct-important $u$-$v$-separators have to do two jobs at the same time. First, they need to disconnect $v$ from $u$, and second they need to clean the reach of $u$ from directed odd cycles. Our approximation algorithm works by delegating the two jobs to different solutions, and solving each of the jobs separately and optimally. 

Just like our \W[1]-hardness proof, our \FPT-approximation for \doct{} builds on the algorithm of Chen et al.~\cite{ChenLLOR08} for DFVS. The method of iterative compression (see~\cite{CyganFKLMPPS15,DBLP:series/txcs/DowneyF13}) allow us to reduce the original problem to the setting where we are given a digraph $D$, an integer $k$, and a directed odd cycle transversal $\hat{S}$ of size $2k+1$. The task is to either determine that $D$ does not have a directed odd cycle transversal of size at most $k$, or output a directed odd cycle transversal of size at most $2k$. We now proceed with a sketch of how to solve this task in \FPT time.

In order to witness that a  digraph $D$ has no directed odd cycles it is sufficient to partition the vertex set of $D$ into sets $Z_1, Z_2, \ldots, Z_\ell$ such that (a) no arc goes from $Z_i$ to $Z_j$ with $j < i$ and (b) for every $i \leq \ell$ the underlying undirected graph of $D[Z_i]$ is bipartite. To certify (b) it is sufficient to provide a coloring of all vertices in $D$ with black or white, such that every arc with both endpoints in $Z_i$ for some $i$ has different colored endpoints. The sets $Z_1, Z_2, \ldots, Z_\ell$ can always be chosen to be the strongly connected components of $D$, and in this case the ordering $Z_1, Z_2, \ldots, Z_\ell$ can be any topological ordering of the directed acyclic graph obtained from $D$ by collapsing every strongly connected component to a vertex.

Suppose now that $D$ has a directed odd cycle transversal $S$ of size at most $k$. Let $Z_1, Z_2, \ldots, Z_\ell$ be a partitioning of $V(D - S)$ and $\phi : V(D - S) \rightarrow \{\mbox{black, white}\}$ be a coloring that certifies that $D - S$ does not have directed odd cycles.
At the cost of a $3^k$ overhead in the running time we can guess for each vertex $v \in \hat{S}$ whether it is deleted (i.e put in the directed odd cycle transversal), colored black or colored white. At the cost of an additional $k!$ overhead in the running time we can guess for every pair of vertices $u$, $v$ in $\hat{S}$ whether they occur in the same strong component $Z_i$, and if not, which of the two strong components containing $u$ and $v$ respectively occurs first in the ordering  $Z_1, Z_2, \ldots, Z_\ell$. Applying these guesses together with some simple reduction rules, we end up in the following setting. The input is a digraph $D$, an integer $k$ and a set $\hat{S}$ such that $D - S$ contains no directed odd cycles, and $D[S]$ is an acyclic tournament (that is, there is an arc between every pair of vertices in S). The task is to either find a set $S \subseteq V(D) \setminus \hat{S}$ of size at most $2k$ such that (a) $S$ is a directed odd cycle transversal, and (b) no strong component of $D - S$ contains more than one vertex of $\hat{S}$, or to conclude that no such set of size at most $k$ exists.

A set $S$ that only satisfies (b) is called a {\em skew separator} for $\hat{S}$, and the main subroutine in the algorithm of Chen et al.~\cite{ChenLLOR08} for DFVS is an algorithm that given $D$, $\hat{S}$ and $k$, runs in time $\Oh(4^kk^{\Oh(1)}(n+m))$, and finds a skew separator $S$ for $\hat{S}$ of size at most $k$ if such a skew separator exists. Our approximation algorithm runs this subroutine and either finds a skew separator $S$ of size at most $k$, or concludes that no set of size at most $k$ can satisfy both (a) and (b) (in particular, just (b)). It then determines in time $2^{\Oh(k^2)}n^{\Oh(1)}$ whether $D - S$ has a directed odd cycle transversal of size at most $k$ disjoint from $\hat{S}$. If such a set $S^*$ exists, the algorithm outputs $S \cup S^*$ as a solution of size at most $2k$ that satisfies (a) and (b). If no such directed odd cycle transversal $S^*$ exists, the approximation algorithm concludes that no set of size at most $k$  can satisfy both (a) and (b) (in particular, just (a)). All that remains is to describe the algorithm for finding  in time $2^{\Oh(k^2)}n^{\Oh(1)}$ a directed odd cycle transversal $S^*$ of size at most $k$ disjoint from $\hat{S}$ in $D - S$, or determining that such a set does not exist.

At this point we observe that the problem breaks up into independent sub-problems for each strongly connected component of $D-S$. For each such component $C$ we have that $|C \cap \hat{S}| \leq 1$, because $S$ is a skew separator for $\hat{S}$. Since $\hat{S}$ is a directed odd cycle transversal for $D$, if  $C \cap \hat{S} = \emptyset$ then there can be no directed odd cycles in $D[C]$. Hence we concentrate on the case when $C \cap \hat{S} = \{w\}$ for a vertex $w$. In other words, we are down to the case where the input is a digraph $D$, integer $k$ and a vertex $w$ such that $\{w\}$ is a directed odd cycle transversal for $D$. The task is to find a directed odd cycle transversal $S^*$ of $D$ of size at most $k$ with $w \notin S^*$.

Define the {\em shadow} of $S^*$ to be the set of all vertices of $D - S^*$ that are not in the strongly connected component of $G - S^*$ containing $w$. Using the technique of {\em shadow removal}, introduced by Marx and Razgon~\cite{MarxR14} (see also~\cite{ChitnisHM13,Chitnis:2012DSFVS}) in their \FPT algorithm for {\sc Multicut}, we can reduce the problem to the special case where the shadow of $S^*$ is empty, at the cost of a $2^{\Oh(k^2)}n^{\Oh(1)}$ overhead in the running time. In this special case $D-S^*$ is strongly connected, and therefore the underlying undirected graph of $D-S^*$ is bipartite. Thus, $S^*$ is an {\em undirected} odd cycle transversal for the underlying undirected graph of $D$. Here we can apply any one of the numerous \FPT algorithms~\cite{IwataOY14,LokshtanovNRRS14,RamanujanS14,ReedSV04} for OCT. 
Thus we can find optimal directed odd cycle transversals in \FPT time for the case when a single undeletable vertex is a directed odd cycle transversal, and as discussed above, this is sufficient to complete the factor $2$ \FPT-approximation.

As a subroutine of our \FPT-approximation we gave an \FPT algorithm for DOCT for the special case where an undeletable directed odd cycle transversal of size $1$ is given as input. Our hardness result also holds for the case where an undeletable directed odd cycle transversal of size $3$ is given as input (the vertices $\{x,y,z\}$ in the construction). Therefore, the parameterized complexity of the case when one also has an undeletable directed odd cycle transversal of size $2$ in the input, is an interesting open problem. It is conceivable that an \FPT algorithm for this case could help in obtaining an {\FPT}-approximation for {\doct} with a factor better than 2.

We remark that this high-level  approach extends to a more general problem that subsumes \doct as well as {\sc Node Unique Label Cover}. Therefore, we design our {\FPT} approximation for the general problem and derive the algorithm for {\doct} as a corollary.



\smallskip
\noindent
{\bf \FPT-Inapproximability.}
For every $\epsilon > 0$ there exists a $\delta > 0$ such that our first reduction, which proves the \W[1]-hardness of \doct, also translates the hardness of $\epsilon$-{\sc gap-BCSP} into hardness of distinguishing between digraphs $D$ such that $\doctfunc{D} \leq k$ from digraphs $D$ such that $\doctfunc{D} > k(1 + \delta)$. However, the reduction only works if in the instance of $\epsilon$-{\sc gap-BCSP} every variable occurs in at most three constraints. To complete the proof of the parameterized inapproximability of \doct, we need to reduce $\epsilon$-{\sc gap-BCSP} to this special case. We achieve this by replacing every ``high degree'' variable by a group of independent low degree variables, while ensuring that the low degree variables all get the same value by introducing a constant degree expander of equality constraints between them.

\section{Preliminaries}\label{sec:prelims}
We use the notations $[t]$ and $[t]_0$ as shorthands of $\{1,2,\ldots,t\}$ and $\{0,1,\ldots,t\}$, respectively. Given a function $f:A\rightarrow\mathbb{R}$ and a subset $A'\subseteq A$, denote $f(A')=\sum_{a\in A'}f(a)$.

\smallskip 
\myparagraph{Parameterized Complexity.} 
Formally, a {\em parameterization} of a problem is the assignment of an integer $k$ to each input instance.  Here, the goal is to confine the combinatorial explosion in the running time of an algorithm for $\Pi$ to depend only on $k$. We say that a parameterized problem $\Pi$ is {\em fixed-parameter tractable} ({\FPT}) if there exists an algorithm that solves $\Pi$ in time $f(k)\cdot |I|^{\bigoh(1)}$, where $|I|$ is the size of the input instance and $f$ is an
arbitrary computable function depending only on the parameter $k$. 

On the negative side, parameterized complexity also provides methods to show that a problem is unlikely to be \FPT. The main technique is the one of parameterized reductions analogous to those employed in classical complexity. Here, the concept of \WO-hardness replaces the one of \NP-hardness, and we need not only construct an equivalent instance in \FPT\ time, but also ensure that the size of the parameter in the new instance depends only on the size of the parameter in the original instance. For our purposes, it is sufficient to note that if there exists such a reduction transforming a problem known to be \WOH\ to another problem $\Pi$, then the problem $\Pi$ is \WO-hard as well. Central \WOH-problems include, for example, the problem of deciding whether a nondeterministic single-tape Turing machine accepts within $k$ steps, the {\sc Clique} problem parameterized be solution size, and the {\sc Independent Set} problem parameterized by solution size

In the context of a parameterized minimization problem $\Pi$, we say that an algorithm for $\Pi$ is an {\em $\alpha$-approximation algorithm} if it always outputs a solution of size at most $\alpha k$ when there exists a solution of size at most $k$ (in other words, the input instance is a \yesinstance), and it always outputs \No\ when there does not not exist a solution of size at most $\alpha k$.
Additional details can be found in the monographs \cite{FG06,Nie06,DBLP:series/txcs/DowneyF13,CyganFKLMPPS15}.

\smallskip 
\myparagraph{Digraphs.}
We refer to standard terminology from the book of Diestel~\cite{Diestel10} for those graph-related terms that are not explicitly defined here.
Given a digraph $D$ and a vertex set $X\subseteq V(D)$, we say that $X$ is a {\em directed odd cycle transversal} of $D$ if $X$ intersects every directed odd cycle of $D$. We further say that $X$ is a {\em minimal} directed odd cycle transversal of $D$ if no proper subset of $D$ is also a directed odd cycle transversal of $D$. Finally, we call $X$ a {\em minimum} directed odd cycle transversal of $D$ if there is no directed odd cycle transversal of $D$ whose size is strictly smaller than the size of $X$. In the context of {\sc DOCT}, we use the terms {\em solution} and {\em $\alpha$-approximate solution} to refer to directed odd cycle transversals of sizes at most $k$ and at most $\alpha k$, respectively.

Given a vertex set $X\subseteq V(D)$, we let $D[X]$ denote the subgraph of $D$ induced by $X$, and we define $D\setminus X=D[V(D)\setminus X]$. Given an arc $(u,v)\in A(D)$, we refer to $u$ as the {\em tail} of the arc and to $v$ as the {\em head} of the arc.
Given a vertex set $X\subseteq V(G)$, we use $N^+(X)$ to denote the set of out-neighbors of $X$ and $N^-(X)$ to denote the set of in-neighbors of $X$. We use $N^i[X]$ to denote the set $X\cup N^i(X)$ where $i\in \{+,-\}$. We denote by $A[X]$ the subset of edges in $A(D)$ with both endpoints in $X$. A {\em strongly connected component} of $D$ is a maximal subgraph in which every vertex has a directed path to every other vertex. We say that a strongly connected component is {\em non-trivial} if it consists of at least two vertices and {\em trivial} otherwise. For disjoint vertex sets $X$ and $Y$, the set $Y$ is said to be {\em reachable from $X$} if for \emph{every} vertex  $y\in Y$, there exists a vertex $x\in X$ such that the $D$ contains a directed path from $x$ to $y$. 
 For a vertex $v\in V(D)$ and walk $W=v_1,\dots, v_r$, we say that $W$ is a $v$-{\em walk} if there is an $i\in [r]$ such that $v=v_i$. 
We say that $W$ is a {\em closed} $v$-walk if $v_1=v_r=v$. We say that $W$ is an $x$-$y$ walk if $v_1=x$ and $v_r=y$. Sometimes we say that a $v$-walk is a $v$-$v$ walk. This is simply so that we can refer to $x$-$y$ walks in general without having to resort to a separate proof (when it is not necessary) for the case when $x$=$y$.
For $1\leq i<j\leq r$, we denote by $W[v_i,v_j]$ the subwalk of $W$ from $v_i$ to $v_j$. We call the vertices $v_2,\dots, v_{r-1}$, the {\em internal} vertices of the walk $W$. For two walks $W_1=v_1,\dots v_t$ and $W_2=w_1,\dots, w_q$ such that $v_t=w_1$, we denote by $W_1+W_2$ the concatenated walk $v_1,\dots, v_{t-1},v_t,w_2,\dots, w_q$.
For disjoint subsets $X,Y,Z\subseteq V(D)$, we call $Z$ an $X$-$Y$ {\em separator} if there is no path from a vertex of $X$ to a vertex of $Y$ in $D-Z$.

Our proofs rely on the following well-known proposition (see, e.g., \cite{DBLP:books/daglib/0006487}).

\begin{proposition}[Folklore]\label{prop:underlyingBipartite}
Let $D$ be a strongly connected directed graph that does not contain a  directed odd cycle. Then, the underlying undirected graph of $D$ is a bipartite graph.
\end{proposition}

\section{W[1]-Hardness}
In this section, we resolve the question of the parameterized complexity of {\sc DOCT}. More precisely, we prove Theorem \ref{thm:maindoct}. For convenience, let us restate the theorem below.

\wHardTheorem*

The source of our reduction is the {\sc Partitioned Subgraph Isomorphism (PSI)} problem. The definition of this problem relies on the notion of a {\em colorful mapping}. Given undirected graphs $H$ and $G$, and a coloring function $col:V(H)\rightarrow V(G)$, we say that an injective function $\varphi: V(G')\rightarrow V(H)$ is a {\em colorful mapping of $G'$ into $H$}, where $G'$ is a subgraph of $G$, if for every $v\in V(G')$, $col(\varphi(v))=v$, and for every $\{u,v\}\in E(G')$, $\{\varphi(u),\varphi(v)\}\in E(H)$. Formally, the {\sc PSI} problem is defined as follows.

\begin{center}
\begin{boxedminipage}{.8\textwidth}
\decnamedefn{{\sc Partitioned Subgraph Isomorphism (PSI)}}{Undirected graphs $H$ and $G$, and a coloring function $col:V(H)\rightarrow V(G)$. The maximum degree of a vertex of $G$ is 3.}
{Does there exist a colorful mapping of $G$ into $H$?}
\end{boxedminipage}
\end{center}

While the {\sc PSI} problem requires us to map the entire graph $G$, to prove our inapproximability result we would also be interested in colorful mappings of {\em subgraphs} of $G$. In the context of the {\sc PSI} problem, we rely on a well-known proposition due to Marx~\cite{Marx07}.

\begin{proposition}[Corollary 6.3, \cite{Marx07}]\label{prop:psiHard}
The {\sc PSI} problem is \WOH. Moreover, unless {\sf ETH} fails, {\sc PSI} cannot be solved in time $f(k)n^{o(\frac{k} {\log k})}$ for any function $f$ where $k=\vert E(G)\vert$. Here, $n=|V(H)|$.
\end{proposition}

The components introduced by our proof may play key roles in other reductions that aim to establish the \W[1]-hardness of problems involving parities and/or cuts. Hence, we have structured our proof as follows. First, for the sake of clarity of the proof, we integrate arc and vertex annotations into the definition of {\sc DOCT}. Then, we introduce the concept of a {\em clock}, which is a gadget that lies at the heart of our reduction. This gadget captures the power of parities in a compact, easy-to-use manner. In particular, it elegantly encodes the selection of two (not necessarily distinct) indices from a set $[n]$ whose sum is upper bounded by $n+1$ (in the case of a {\em forward clock}) or lower bounded by $n+1$ (in the case of a {\em reverse clock}). We remark that the selection is orchestrated by a variable that we call {\em time}. Next, we ``glue'' the tips of the {\em hands} of a forward clock and a reverse clock together as well as attach arcs that connect carefully chosen vertices on these hands to obtain a {\em double clock}. The double clock is a gadget that both ensures that two clocks show the exact same time and that this time corresponds to the selection of two indices whose sum is {\em exactly} $n+1$. Roughly speaking, it is mentally convenient to associate each double clock with a different {\em time zone} that encodes the selection of {\em one} element. Here, since our source problem is a graph problem, the natural choice of an element is a vertex. Having established a time zone for each selection of one element, we turn to {\em synchronize} hands of {\em different} double clocks. For this purpose, we introduce the {\em synchronizer}, which is a gadget that resembles a {\em folded grid}. We remark that this specific gadget is different yet inspired by a folded grid gadget that is the core of the paper \cite{DBLP:conf/soda/PilipczukW16}. Having double clocks and synchronizers at hand, we are finally able to present the entire reduction in an intuitive (yet precise) manner. Lastly, we prove that our reduction is correct. At this point, having already established key properties of our gadgets, the reverse direction (``solution to {\sc DOCT} $\rightarrow$ solution to {\sc PSI}'') is simple. For the forward direction (``solution to {\sc PSI} $\rightarrow$ solution to {\sc DOCT}'') we exhibit a partition of the vertex set of the output digraph into pairwise-disjoint sets on which we can define a topological order, such that the graph induced by each set can be shown to exclude directed odd cycles.\footnote{For the sake of clarity, we integrate the lemmata necessary to exhibit this partition into the sections presenting individual gadgets.}

\subsection{Annotations}

Let us begin our proof by integrating arc and vertex annotations into the definition of {\sc DOCT}. More precisely, we generalize {\sc DOCT} as follows.

\begin{center}
\begin{boxedminipage}{.8\textwidth}
\decnamedefn{{\sc Annotated DOCT (A-DOCT)}}{A digraph $D$, a non-negative integer $k$, a labeling function $\ell: A(D)\rightarrow \{0,1\}$, and a weight function $w: V(D)\rightarrow [2k+1]$.}
{Does there exist a subset $X\subseteq V(D)$ such that $w(X)\leq k$ and $X$ intersects every directed cycle $C$ of $D$ where $\ell(E(C))$ is odd?}
\end{boxedminipage}
\end{center}

Henceforth, in the context of {\sc A-DOCT}, the term {\em directed odd cycle} would refer to a directed cycle such that $\ell(E(C))$ is odd. As we show in this section, it is easy to see that in order to prove Theorems~\ref{thm:doctHard1} and~\ref{thm:doctHard2}, we can focus on the {\sc A-DOCT} problem.

Let us now present our reduction from {\sc A-DOCT} to {\sc DOCT}. For this purpose, let $(D,k,\ell,w)$ be an instance of {\sc A-DOCT}. Then, we construct an instance ${\bf red}(D,k,\ell,w)=(D',k')$ of {\sc DOCT} as follows. First, set $k'=k$. Let $A_0=\{a\in A(D): \ell(a)=0\}$ and $A_1=\{a\in A(D): \ell(a)=1\}$.
Next, define $V(D')=P\cup Q$, where $P=\{p^i_a: i\in [\alpha k+1], a\in A_0\}$ and $Q=\{q^i_v: i\in[w(v)],v\in V(D)\}$.
Finally, we define $A(D')=S\cup T\cup R$, where $S=\{(q^i_v,p^j_a): q^i_v\in Q, p^j_a\in P, v$ is the tail of $a\}$, $T=\{(p^i_a,q^j_v): p^i_a\in P, q^j_v\in Q, v$ is the head of $a\}$ and $R=\{(q^i_u,q^j_v): (u,v)\in A_1\}$. Clearly, $(D',k')$ can be outputted in polynomial time.

\begin{lemma}\label{lem:annotate}
Let $(D,k,\ell,w)$ be an instance of {\sc A-DOCT}. If there exists a solution for $(D,k,\ell,w)$, then there exists a solution for ${\bf red}(D,k,\ell,w)=(D',k')$. Moreover, if there exists an $\alpha$-approximate solution for $(D',k')$, then there exists an $\alpha$-approximate solution for $(D,k,\ell,w)$.
\end{lemma}

\begin{proof}
Fix $\alpha\geq 1$. In the first direction, let $X$ be a solution for $(D,k,\ell,w)$. 
We claim that $X'=\{q^i_v: v\in X, i\in [w(v)]\}$ is a solution to $(D',k')$. Suppose, by way of contradiction, that this claim is false. Then, since $|X'|=w(X)\leq k=k'$, there exists a directed odd cycle $C'$ of minimum size of $D'\setminus X$. If there exist $q^i_v,q^j_v\in V(C')\cap Q$ such that $i\neq j$, then we obtain a contradiction to the choice of $C'$. Indeed, if we replace $q^i_v$ by $q^j_v$ in $C'$, then the result is a directed odd closed walk. Since a directed odd closed walk contains a directed odd cycle, we obtain a directed odd cycle that is shorter than $C'$.
Hence, by the definitions of $P$, $S$ and $T$, we have that the graph $C$ on $\{v: q^i_v\in V(C')\}$, where $(u,v)\in A(C)$ if and only if there exist indices $i,j,t$ such that either $(q^i_u,q^j_v)\in A(C')$ or $(q^i_u,p^t_{(q^i_u,q^j_v)})\in A(C')$, is a directed odd cycle of $D\setminus X$. Thus, we have reached a contradiction to the supposition that $X$ is a solution to $(D,k,\ell,w)$.

Second, let $X'$ be an $\alpha$-approximate solution for $(D',k')$. Without loss of generality, assume that $X'$ is a minimal solution. We first claim that $X'\cap P=\emptyset$. For all $a\in D(A)$, the vertices $p^i_a$ have the same set of outgoing neighbors and the same set of incoming neighbors. Hence, if there exist $i\neq j$ such that $p^i_a\in X'$ but $p^j_a\notin X'$, then $X'\setminus \{p^i_a\}$ is also a solution to $(D',k')$. Indeed, if $D'\setminus (X'\setminus \{p^i_a\})$ contains a directed odd cycle, then this cycle must contain $p^i_a$. By replacing $p^i_a$ by $p^j_a$, we obtain a directed odd walk of $D'$ (which contains a directed odd cycle of $D'$). Hence, we reach a contradiction to the minimality of $X'$.
Since there are $\alpha k+1$ vertices $p^i_a$ while $|X'|\leq \alpha k$, we conclude that $X'\cap P=\emptyset$. Moreover, for all $v\in V(D)$, the vertices $q^i_v$ have the same set of outgoing neighbors and the same set of incoming neighbors. Hence, we again deduce that there cannot exist $i\neq j$ such that $q^i_v\in X'$ but $q^i_v\notin X'$. Let us denote $X=\{v: q^1_v\in X'\}$. Then, $w(X)=|X'|\leq \alpha k'=\alpha k$.
 We claim that $X$ is a solution to $(D,k,\ell,w)$. Suppose, by way of contradiction, that this claim is false. Then, let $C$ be a directed odd cycle of $D$. Let $C'$ be obtained from $C$ by replacing each arc $a=(u,v)\in A(C)\cap A_0$ by the two arcs $(q^1_u,p^1_a)$ and $(p^1_a,q^1_v)$.
By the definitions of $P$, $S$ and $T$, and since we have argued that $X'\cap P=\emptyset$, we have that $C$ is a directed odd cycle of $D'\setminus X'$. Hence, we have reached a contradiction to the supposition that $X'$ is an $\alpha$-approximate solution to $(D',k')$.
\end{proof}

As a corollary to Lemma \ref{lem:annotate}, we derive the following result.

\begin{corollary}\label{cor:annotate}
For all $\alpha\geq 1$, if there exists an $\alpha$-approximation algorithm for {\sc DOCT} that runs in time $\tau$, then there exists an $\alpha$-approximation algorithm for {\sc A-DOCT} that runs in time $\bigoh(\tau+n^{\bigoh(1)})$.
\end{corollary}

\begin{proof}
If there exists an $\alpha$-approximation algorithm $\cal A$ for {\sc DOCT} that runs in time $\tau$, then we define $\cal B$ as the algorithm that given an instance $(D,k,\ell,w)$ of {\sc A-DOCT}, constructs the instance ${\bf red}(D,k,\ell,w)=(D',k')$ of {\sc DOCT}, and calls algorithm $\cal A$ with $(D',k')$ as input. If $\cal A$ returns an $\alpha$-approximate solution for $(D',k')$, we have shown (in Lemma \ref{lem:annotate}) how to translate it to an $\alpha$-approximate solution for $(D,k,\ell,w)$. Moreover, if $(D,k,\ell,w)$ is a \yesinstance, then we have shown that $(D',k')$ is a \yesinstance. Hence, $\cal A$ would return an $\alpha$-approximate solution for $(D',k')$. Thus, we obtain an $\alpha$-approximation algorithm for {\sc A-DOCT} that runs in time $\bigoh(\tau+n^{\bigoh(1)})$.
\end{proof}

\begin{figure}[t!]\centering
\fbox{\includegraphics[scale=0.7]{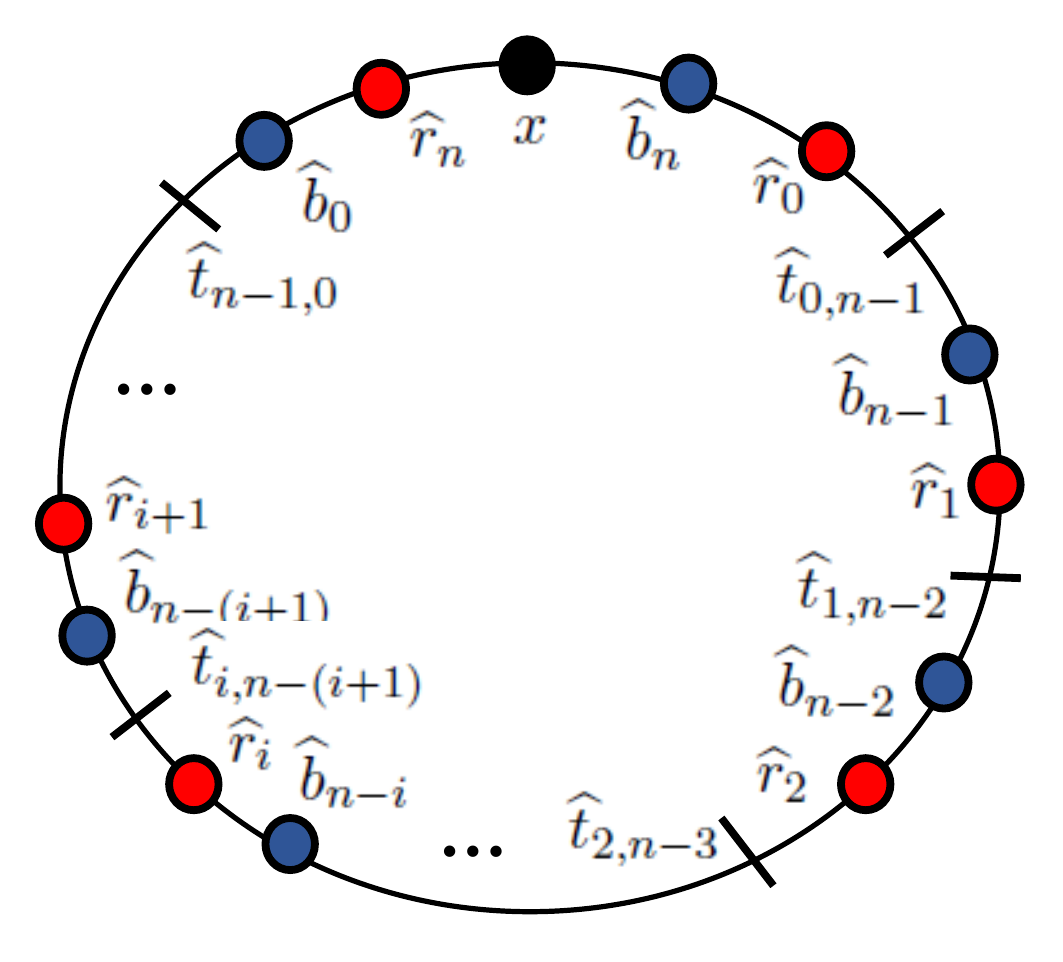}}
\caption{The face of a forward clock.}\label{fig:ForwardClock1}
\end{figure}

\subsection{The Basic Clock Gadget}\label{sec:clockBasic}

Let $n,k\in\mathbb{N}$ such that $k\geq 100$. Here, we define an {\em $(n,k)$-forward clock} and an {\em $(n,k)$-reverse clock}. Since $n$ and $k$ would be clear from context, we simply write {\em forward clock} and {\em reverse clock} rather than $(n,k)$-forward clock and $(n,k)$-reverse clock, respectively.

\subsubsection{Forward Clock}

\myparagraph{Structure.} We first define a forward clock $C$. The {\em face} of $C$ is an ``undirected'' cycle whose vertex set is the union of four pairwise-disjoint sets $\widehat{R}$ (red), $\widehat{B}$ (blue), $\widehat{T}$ (time) and $\{x\}$. We refer the reader to Fig.~\ref{fig:ForwardClock1}. We set $\widehat{R}=\{\widehat{r}_i: i\in[n]_0\}$, $\widehat{B}=\{\widehat{b}_i: i\in[n]_0\}$ and $\widehat{T}=\{\widehat{t}_{i,n-i-1}: i\in[n-1]_0\}$. The arc set of the face is the union of the following three pairwise-disjoint sets.
\begin{itemize}
\item $\{(x,\widehat{r}_n),(x,\widehat{b}_n),(\widehat{r}_n,x),(\widehat{b}_n,x)\}$.
\item $\{(\widehat{b}_i,\widehat{r}_{n-i}): i\in [n]_0\}\cup \{(\widehat{r}_{n-i},\widehat{b}_{i}): i\in [n]_0\}$.
\item $\{(\widehat{r}_i,\widehat{t}_{i,n-i-1}): i\in [n-1]_0\}\cup \{(\widehat{b}_{n-i-1},\widehat{t}_{i,n-i-1}): i\in [n-1]_0\}\cup \{(\widehat{t}_{i,n-i-1},\widehat{r}_i): i\in [n-1]_0\}\cup \{(\widehat{t}_{i,n-i-1},\widehat{b}_{n-i-1}): i\in [n-1]_0\}$.
\end{itemize}

\begin{figure}[t!]\centering
\fbox{\includegraphics[scale=0.7]{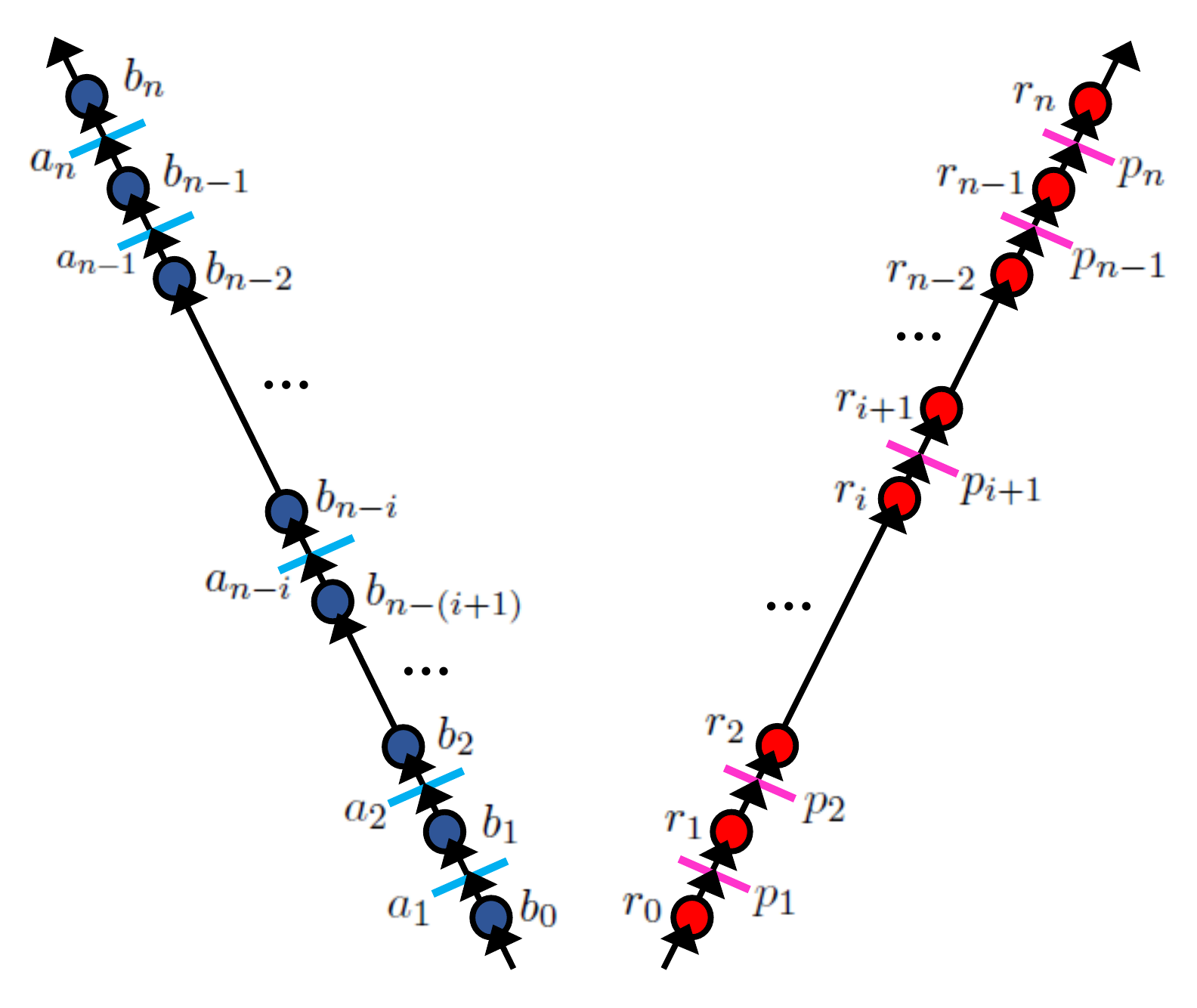}}
\caption{The hands of a forward clock. For all $i\in[n]$, $\pre(p_i)=r_{i-1}$ and $\post(p_i)=r_i$, and $\pre(a_i)=b_{i-1}$ and $\post(a_i)=b_i$.}\label{fig:ForwardClock2}
\end{figure}

The {\em hands} of $C$ are two directed paths, red and blue (see Fig.~\ref{fig:ForwardClock2}). The vertex set of the red path is the union of two pairwise disjoint sets, $R=\{r_i: i\in[n]_0\}$ (red) and $P=\{p_i: i\in[n]\}$ (pink). For all $i\in[n]$, we denote $\pre(p_i)=r_{i-1}$ and $\post(p_i)=r_i$. The arc set of the red path is $\{(\pre(p_i),p_i): i\in[n]\}\cup\{(p_i,\post(p_i)): i\in [n]\}$. Symmetrically, the blue path is the union of two pairwise disjoint sets, $B=\{b_i: i\in[n]_0\}$ (red) and $A=\{a_i: i\in[n]\}$ (azure). For all $i\in[n]$, we denote $\pre(a_i)=b_{i-1}$ and $\post(a_i)=b_i$. The arc set of the blue path is $\{(\pre(a_i),a_i): i\in[n]\}\cup\{(a_i,\post(a_i)): i\in [n]\}$.

The hands are attached to the face as follows (see Fig.~\ref{fig:ForwardClock3}). First, we add the arcs $(x,r_0)$ and $(x,b_0)$. Second, for all $i\in[n]_0$, we add the arcs $(r_i,\widehat{r}_i)$ and $(b_i,\widehat{b}_i)$. Then, we ``glue'' the hands by adding a new vertex, $y$, and the arcs $(r_n,y)$, $(b_n,y)$ and $(y,x)$.

Finally, let us annotate $C$ (see Fig.~\ref{fig:ForwardClock3}). The labels of the arcs in the set $\{(x,\widehat{r}_n),(\widehat{r}_n,x),(y,x)\}\cup\{(b_i,\widehat{b}_i): i\in[n]_0\}$ are equal to $1$, and the labels of all other arcs are equal to $0$. Moreover, the weight of the vertices in the set $\widehat{T}\cup P\cup A$ are equal to $10$, and the weights of all other vertices are equal to $2k+1$. This completes the description of $C$. When the clock $C$ is not clear from context, we add the notation $(C)$ to an element (vertex set or vertex) of the clock. For example, we may write $R(C)$ and $x(C)$.

\begin{figure}[t!]\centering
\fbox{\includegraphics[scale=0.7]{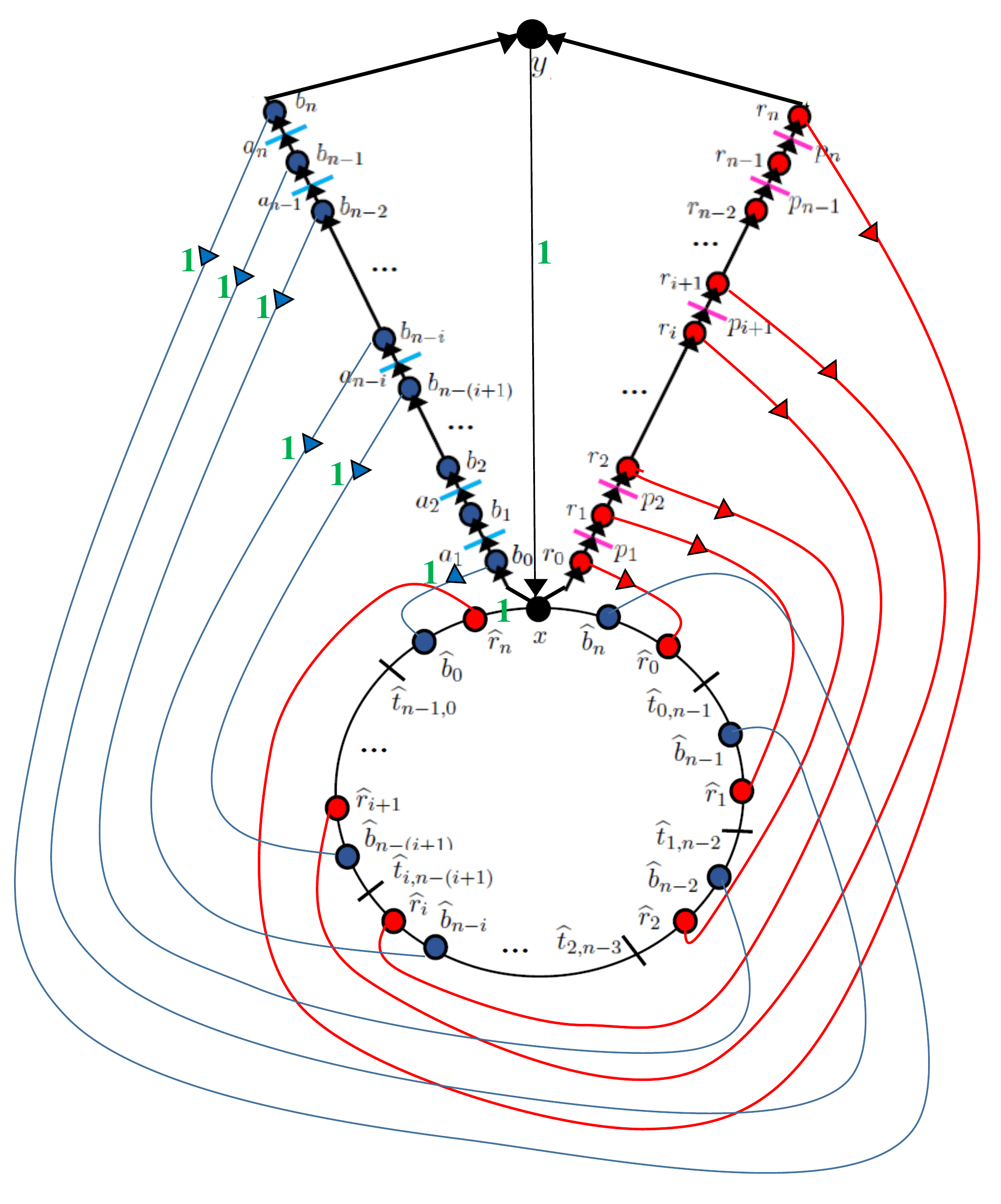}}
\caption{A forward clock. The arcs labeled 1 are marked by a green `1'. The weight of vertices marked by circles is $2k+1$, and the weight of vertices marked by lines is 10.}\label{fig:ForwardClock3}
\end{figure}

\smallskip
\myparagraph{Properties.} By the definition of a forward clock, we directly identify which directed odd cycles are present in such a clock.

\begin{observation}\label{obs:forwardCTypes}
Let $C$ be a forward clock. The set of directed odd cycles of $C$ is the union of the following sets.
\begin{itemize}
\item {\bf Type 1:} The set whose only directed odd cycle is the one consisting of the entire red hand and the arc $(y,x)$.
\item {\bf Type 2:} The set whose only directed odd cycle is the one consisting of the entire blue hand and the arc $(y,x)$.
\item {\bf Type 3:} For all $i\in [n]_0$, this set contains the directed odd cycle consisting of the directed path from $x$ to $r_i$ on the red hand, the arc $(r_i,\widehat{r}_i)$, and the directed path from $\widehat{r}_i$ to $x$ on the face of the clock that contains the arc $(\widehat{r}_n,x)$.
\item {\bf Type 4:} For all $i\in [n]_0$, this set contains the directed odd cycle consisting of the directed path from $x$ to $b_i$ on the blue hand, the arc $(b_i,\widehat{b}_i)$, and the directed path from $\widehat{b}_i$ to $x$ on the face of the clock that contains the arc $(\widehat{b}_n,x)$.
\item {\bf Type 5:} The set whose only directed odd cycle is the face of the clock.
\end{itemize}
\end{observation}

We proceed to derive properties of ``cuts'' of a forward clock. To this end, we first need to define the kind of sets using which we would like to ``cut'' forward clocks.

\begin{definition}\label{def:clockNiceCut}
Let $C$ be a forward clock. We say that a set $X\subseteq V(C)$ {\em cuts $C$ precisely} if there exist $i,j,s\in[n]$ such that $X=\{p_i,a_j,\widehat{t}_{s,n-s-1}\}$, $i-1\leq s$ and $j\leq n-s$.
\end{definition}

Definition \ref{def:clockNiceCut} directly implies the following observation.

\begin{observation}\label{obs:forClockMaxN1}
Let $C$ be a forward clock. If $X=\{p_i,a_j,\widehat{t}_{s,n-s-1}\}$ is a set that cuts $C$ precisely, then $i+j\leq n+1$.
\end{observation}

We are now ready to present the desired properties of ``cuts'' of a forward clock.

\begin{lemma}\label{lem:forClockPrecise}
Let $C$ be a forward clock. A set $X\subseteq V(C)$ is a directed odd cycle transversal of $C$ of weight exactly 30 if and only if $X$ cuts $C$ precisely.
\end{lemma}

\begin{proof}
In the forward direction, let $X\subseteq V(C)$ be a directed odd cycle transversal of $C$ of weight exactly 30. Recall that $2k+1>40$. Hence, to intersect the directed odd cycle of Type 1 (see Observation \ref{obs:forwardCTypes}), the set $X$ must contain a vertex of the form $p_i$ for some $i\in[n]$. Symmetrically, to intersect the directed odd cycle of Type 2, the set $X$ must contain a vertex of the form $a_j$ for some $j\in[n]$. Moreover, to intersect the directed odd cycle of Type 5, the set $X$ must contain a vertex of the form $\widehat{t}_{s,n-s-1}$ for some $s\in[n-1]_0$. Since $w(X)=30$, we deduce that $X=\{p_i,a_j,\widehat{t}_{s,n-s-1}\}$. To prove that $X$ cuts $C$ precisely, it remains to show that $i-1\leq s$ and $j\leq n-s$. For this purpose, consider the directed odd cycle of Type 3 that consists of the directed path from $x$ to $r_{i-1}$ on the red hand, the arc $(r_{i-1},\widehat{r}_{i-1})$, and the directed path from $\widehat{r}_{i-1}$ to $x$ on the face of the clock that contains the arc $(\widehat{r}_n,x)$. Since $X$ must intersect this directed odd cycle, it must hold that $i-1\leq s$. Now, consider the directed odd cycle of Type 4 that consists of the directed path from $x$ to $b_{j-1}$ on the blue hand, the arc $(b_{j-1},\widehat{b}_{j-1})$, and the directed path from $\widehat{b}_{j-1}$ to $x$ on the face of the clock that contains the arc $(\widehat{b}_n,x)$. Since $X$ must intersect this directed odd cycle, it must also hold that $j-1\leq n-s-1$, and therefore $j\leq n-s$.

In the reverse direction, let $X\subseteq V(C)$ be a set that cuts $C$ precisely. Then, there exist $i,j,s\in[n]$ such that $X=\{p_i,a_j,\widehat{t}_{s,n-s-1}\}$, $i-1\leq s$ and $j\leq n-s$. Clearly, $w(X)=30$. Since $p_i\in X$, it holds that $X$ intersects the directed odd cycle of Type 1 as well as every directed odd cycle of Type 3 that consists of a directed path from $x$ to $r_{i'}$ on the red hand for some $i'\geq i$, the arc $(r_{i'},\widehat{r}_{i'})$, and the directed path from $\widehat{r}_{i'}$ to $x$ on the face of the clock that contains the arc $(\widehat{r}_n,x)$. Symmetrically, since $a_j\in X$, it holds that $X$ intersects the directed odd cycle of Type 2 as well as every directed odd cycle of Type 4 that consists of a directed path from $x$ to $b_{j'}$ on the blue hand for some $j'\geq j$, the arc $(b_{j'},\widehat{b}_{j'})$, and the directed path from $\widehat{b}_{j'}$ to $x$ on the face of the clock that contains the arc $(\widehat{b}_n,x)$. Since $\widehat{t}_{s,n-s-1}\in X$, it holds that $X$ intersects that directed odd cycle of Type 5. Moreover, since $i-1\leq s$, it holds that $X$ intersects every directed odd cycle of Type 3 that consists of a directed path from $x$ to $r_{i'}$ on the red hand for some $i'<i$, the arc $(r_{i'},\widehat{r}_{i'})$, and the directed path from $\widehat{r}_{i'}$ to $x$ on the face of the clock that contains the arc $(\widehat{r}_n,x)$. Symmetrically, since $j-1\leq n-s-1$, it holds that $X$ intersects every directed odd cycle of Type 4 that consists of a directed path from $x$ to $b_{j'}$ on the blue hand for some $j'< j$, the arc $(b_{j'},\widehat{b}_{j'})$, and the directed path from $\widehat{b}_{j'}$ to $x$ on the face of the clock that contains the arc $(\widehat{b}_n,x)$. We have thus verified that $X$ is a directed odd cycle transversal of $C$.
\end{proof}

\begin{lemma}\label{lem:forClockNotPrecise}
Let $C$ be a forward clock. The weight of a set $X\subseteq V(C)$ that is a directed odd cycle transversal of $C$ but does not cut $C$ precisely is at least 40.
\end{lemma}

\begin{proof}
Consider the directed odd cycles of $C$ of Types 1, 2 and 5 (see Observation \ref{obs:forwardCTypes}). The only vertices that are present in at least two of these cycles are of weight $2k+1>40$. Hence, if $w(X)\leq 40$, then $X$ must contain at least three vertices of $C$, each of weight 10. The set $X$ cannot contain exactly three vertices of $C$ of weight 10, since then $w(X)=30$, in which Lemma~\ref{lem:forClockPrecise} implies $X$ should have cut $C$ precisely. Hence,  if $w(X)\leq 40$, then $X$ contains at least four vertices of $C$, and therefore $w(X)\geq 40$.
\end{proof}

\begin{figure}[t!]\centering
\fbox{\includegraphics[scale=0.7]{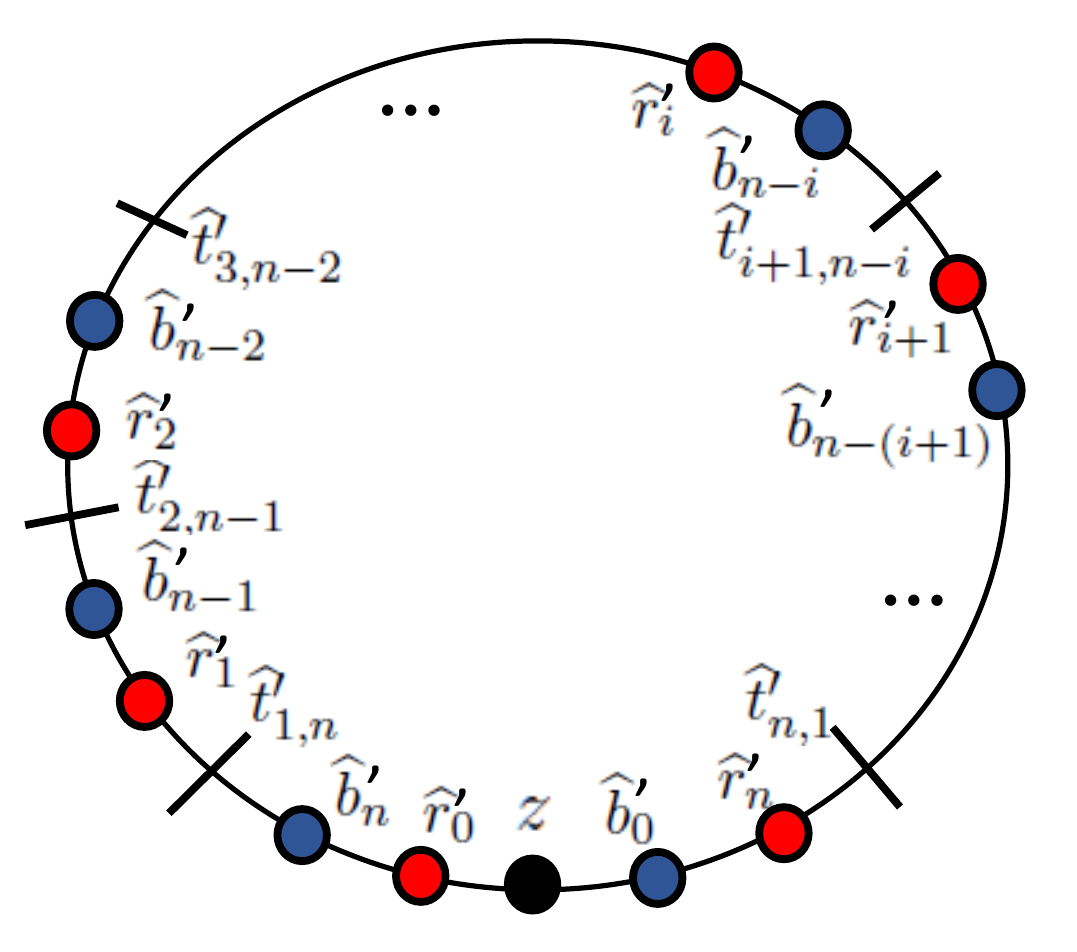}}
\caption{The face of a reverse clock.}\label{fig:ReverseClock1}
\end{figure}

\subsubsection{Reverse Clock}

A reverse clock is simply a forward clock where the directions of {\em all} arcs have been reversed. However, to ensure the readability of the rest of the paper, it would be convenient to rename the vertices of a reverse clock as well as draw the illustrations differently. In particular, the vertices would be indexed differently. Hence, for the sake of clarity, we provide a detailed description of a reverse clock.

\myparagraph{Structure.} The {\em face} of a reverse clock $C$ is an ``undirected'' cycle whose vertex set is the union of four pairwise-disjoint sets $\widehat{R}'$ (red), $\widehat{B}'$ (blue), $\widehat{T}'$ (time) and $\{z\}$. We refer the reader to Fig.~\ref{fig:ReverseClock1}. We set $\hat{R}'=\{\widehat{r}'_i: i\in[n]_0\}$, $\hat{B}'=\{\widehat{b}'_i: i\in[n]_0\}$ and $\widehat{T}=\{\widehat{t}'_{i,n-i+1}: i\in[n]\}$. The arc set of the face is the union of the following three pairwise-disjoint sets.
\begin{itemize}
\item $\{(z,\widehat{r}'_0),(z,\widehat{b}'_0),(\widehat{r}'_0,z),(\widehat{b}'_0,z)\}$.
\item $\{(\widehat{b}'_i,\widehat{r}'_{n-i}): i\in [n]_0\}\cup \{(\widehat{r}'_{n-i},\widehat{b}'_{i}): i\in [n]_0\}$.
\item $\{(\widehat{r}'_i,\widehat{t}'_{i,n-i+1}): i\in [n]\}\cup \{(\widehat{b}'_{n-i+1},\widehat{t}'_{i,n-i+1}): i\in [n]\}\cup \{(\widehat{t}'_{i,n-i+1},\widehat{r}'_i): i\in [n]\}\cup \{(\widehat{t}'_{i,n-i+1},\widehat{b}'_{n-i+1}): i\in [n]\}$.
\end{itemize}

The {\em hands} of $C$ are two directed paths, red and blue. These hands are defined exactly as the hands of a forward clock, except that tags are added to the names of all of their vertices (see Fig.~\ref{fig:ReverseClock2}). The hands are attached to the face as follows (see Fig.~\ref{fig:ReverseClock3}). First, we add the arcs $(r'_n,z)$ and $(b'_n,z)$. Second, for all $i\in[n]_0$, we add the arcs $(\widehat{r}'_i,r'_i)$ and $(\widehat{b}'_i,b'_i)$. Then, we ``glue'' the hands by adding a new vertex, $y$, and the arcs $(y,r'_0)$, $(y,b'_0)$ and $(z,y)$.

\begin{figure}[t!]\centering
\fbox{\includegraphics[scale=0.7]{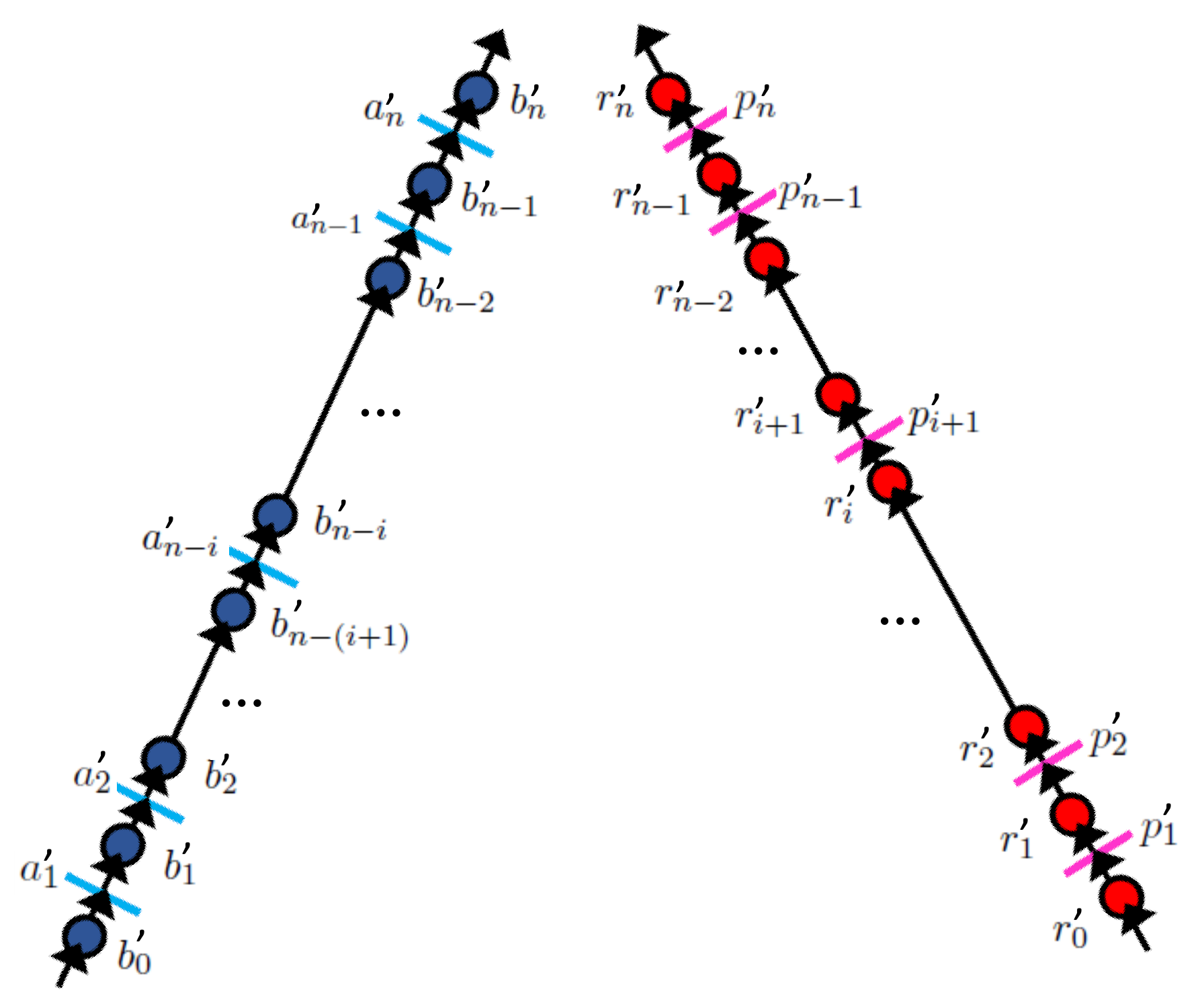}}
\caption{The hands of a reverse clock. For all $i\in[n]$, $\pre(p_i')=r_{i-1}'$ and $\post(p_i')=r_i'$, and $\pre(a_i')=b_{i-1}'$ and $\post(a_i')=b_i'$.}\label{fig:ReverseClock2}
\end{figure}

\begin{figure}[t!]\centering
\fbox{\includegraphics[scale=0.6]{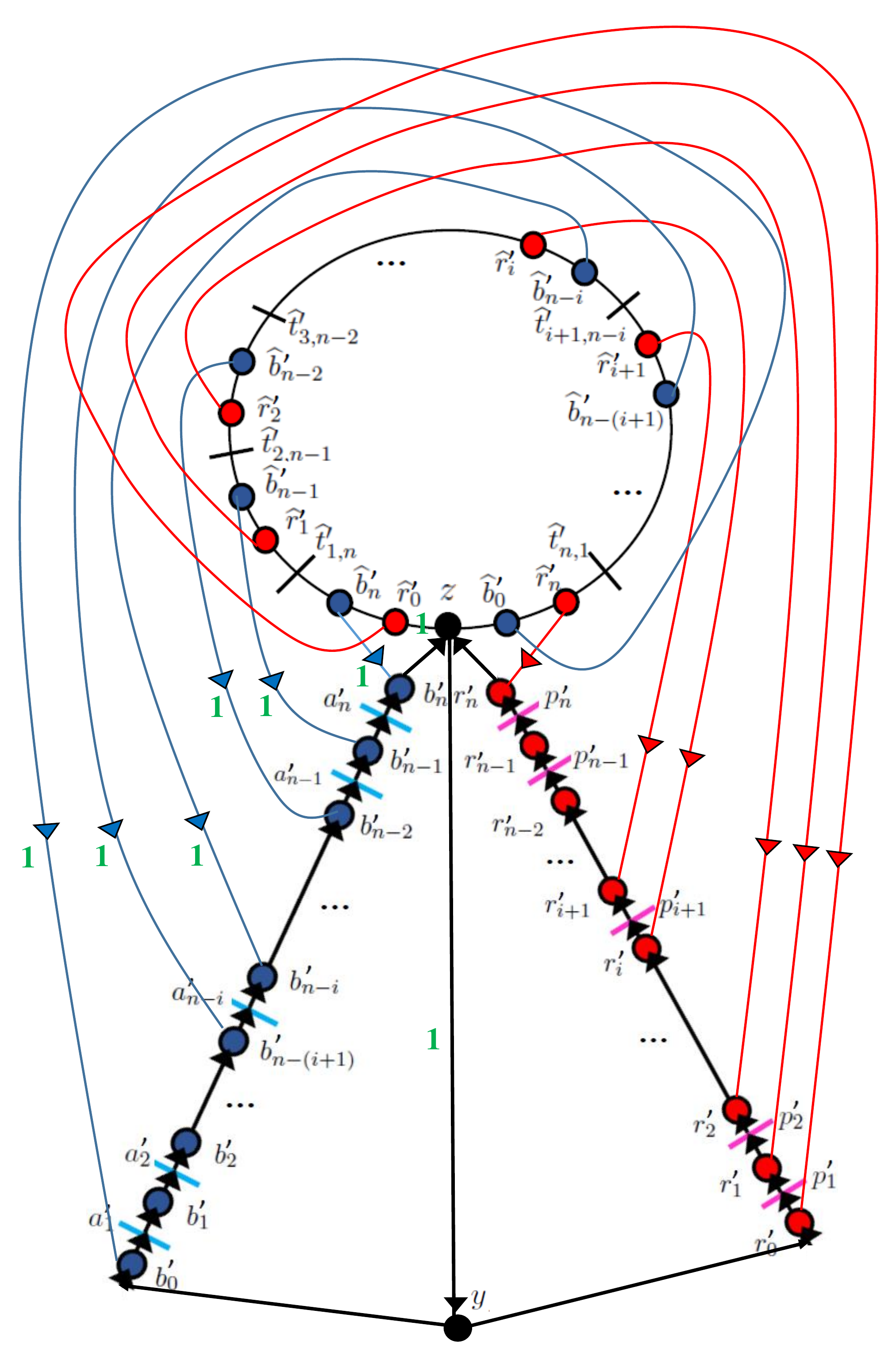}}
\caption{A reverse clock. The arcs labeled 1 are marked by a green `1'. The weight of vertices marked by circles is $2k+1$, and the weight of vertices marked by lines is 10.}\label{fig:ReverseClock3}
\end{figure}

Finally, let us annotate $C$ (see Fig.~\ref{fig:ReverseClock3}). The labels of the arcs in the set $\{(z,\widehat{r}'_0),(\widehat{r}'_0,z),(z,y)\}\cup\{(\widehat{b}'_i,b'_i): i\in[n]_0\}$ are equal to $1$, and the labels of all other arcs are equal to $0$. Moreover, the weight of the vertices in the set $\widehat{T}'\cup P'\cup A'$ are equal to $10$, and the weights of all other vertices are equal to $2k+1$. This completes the description of $C$.

\smallskip
\myparagraph{Properties.} By the definition of a reverse clock, we directly identify which directed odd cycles are present in such a clock.

\begin{observation}\label{obs:reverseCTypes}
Let $C$ be a reverse clock. The set of directed odd cycles of $C$ is the union of the following sets.
\begin{itemize}
\item {\bf Type 1:} The set whose only directed odd cycle is the one consisting of the entire red hand and the arc $(z,y)$.
\item {\bf Type 2:} The set whose only directed odd cycle is the one consisting of the entire blue hand and the arc $(z,y)$.
\item {\bf Type 3:} For all $i\in [n]_0$, this set contains the directed odd cycle consisting of the arc $(\widehat{r}'_i,r'_i)$, the directed path from $r'_i$ to $z$ on the red hand, and the directed path from $z$ to $\widehat{r}'_i$ on the face of the clock that contains the arc $(z,\widehat{r}'_0)$.
\item {\bf Type 4:} For all $i\in [n]_0$, this set contains the directed odd cycle consisting of the arc $(\widehat{b}'_i,b'_i)$, the directed path from $b'_i$ to $z$ on the blue hand, and the directed path from $z$ to $\widehat{b}'_i$ on the face of the clock that contains the arc $(z,\widehat{b}'_0)$.
\item {\bf Type 5:} The set whose only directed odd cycle is the face of the clock.
\end{itemize}
\end{observation}

As in the case of a forward clock, we proceed to derive properties of ``cuts'' of a reverse clock. To this end, we first need to define the kind of sets using which we would like to ``cut'' reverse clocks.

\begin{definition}\label{def:clockNiceCutReverse}
Let $C$ be a reverse clock. We say that a set $X\subseteq V(C)$ {\em cuts $C$ precisely} if there exist $i,j,s\in[n]$ such that $X=\{p'_i,a'_j,\widehat{t}'_{s,n-s+1}\}$, $s\leq i$ and $n-s+1\leq j$.
\end{definition}

Definition \ref{def:clockNiceCutReverse} directly implies the following observation. Note that due to our placement of indices, the inequality is complementary to the one in Observation \ref{obs:forClockMaxN1}.

\begin{observation}\label{obs:revClockMaxN1}
Let $C$ be a reverse clock. If $X=\{p'_i,a'_j,\widehat{t}'_{s,n-s-1}\}$ is a set that cuts $C$ precisely, then $i+j\geq n+1$.
\end{observation}

We are now ready to present the desired properties of ``cuts'' of a reverse clock.

\begin{lemma}\label{lem:revClockPrecise}
Let $C$ be a reverse clock. A set $X\subseteq V(C)$ is a directed odd cycle transversal of $C$ of weight exactly 30 if and only if $X$ cuts $C$ precisely.
\end{lemma}

\begin{proof}
In the forward direction, let $X\subseteq V(C)$ be a directed odd cycle transversal of $C$ of weight exactly 30. Recall that $2k+1>40$. Hence, to intersect the directed odd cycle of Type 1 (see Observation \ref{obs:reverseCTypes}), the set $X$ must contain a vertex of the form $p'_i$ for some $i\in[n]$. Symmetrically, to intersect the directed odd cycle of Type 2, the set $X$ must contain a vertex of the form $a'_j$ for some $j\in[n]$. Moreover, to intersect the directed odd cycle of Type 5, the set $X$ must contain a vertex of the form $\widehat{t}'_{s,n-s+1}$ for some $s\in[n]$. Since $w(X)=30$, we deduce that $X=\{p'_i,a'_j,\widehat{t}'_{s,n-s+1}\}$. To prove that $X$ cuts $C$ precisely, it remains to show that $s\leq i$ and $n-s+1\leq j$. For this purpose, consider the directed odd cycle of Type 3 that consists of the arc $(\widehat{r}'_i,r'_i)$, the directed path from $r'_i$ to $z$ on the red hand, and the directed path from $z$ to $\widehat{r}'_i$ on the face of the clock that contains the arc $(z,\widehat{r}'_0)$. Since $X$ must intersect this directed odd cycle, it must hold that $s\leq i$. Now, consider the directed odd cycle of Type 4 that consists of the arc $(\widehat{b}'_j,b'_j)$, the directed path from $b'_j$ to $z$ on the blue hand, and the directed path from $z$ to $\widehat{b}'_j$ on the face of the clock that contains the arc $(z,\widehat{b}'_0)$. Since $X$ must intersect this directed odd cycle, it must also hold that $n-s+1\leq j$.

In the reverse direction, let $X\subseteq V(C)$ be a set that cuts $C$ precisely. Then, there exist $i,j,s\in[n]$ such that $X=\{p'_i,a'_j,\widehat{t}'_{s,n-s+1}\}$, $s\leq i$ and $n-s+1\leq j$. Clearly, $w(X)=30$. Since $p'_i\in X$, it holds that $X$ intersects the directed odd cycle of Type 1 as well as every directed odd cycle of Type 3 that consists of the arc $(\widehat{r}'_{'},r'_{i'})$ for some $i'<i$, the directed path from $r'_{i'}$ to $z$ on the red hand, and the directed path from $z$ to $\widehat{r}'_{i'}$ on the face of the clock that contains the arc $(z,\widehat{r}'_0)$. Symmetrically, since $a'_j\in X$, it holds that $X$ intersects the directed odd cycle of Type 2 as well as every directed odd cycle of Type 4 that consists of the arc $(\widehat{b}'_{j'},b'_{j'})$ for some $j'<j$, the directed path from $b'_{j'}$ to $z$ on the blue hand, and the directed path from $z$ to $\widehat{b}'_{j'}$ on the face of the clock that contains the arc $(z,\widehat{b}'_0)$. Since $\widehat{t}'_{s,n-s+1}\in X$, it holds that $X$ intersects that directed odd cycle of Type 5. Moreover, since $s\leq i$, it holds that $X$ intersects every directed odd cycle of Type 3 that consists of the arc $(\widehat{r}'_{'},r'_{i'})$ for some $i'\geq i$, the directed path from $r'_{i'}$ to $z$ on the red hand, and the directed path from $z$ to $\widehat{r}'_{i'}$ on the face of the clock that contains the arc $(z,\widehat{r}'_0)$. Symmetrically, since $n-s+1\leq j$, it holds that $X$ intersects every directed odd cycle of Type 4 that consists of the arc $(\widehat{b}'_{j'},b'_{j'})$ for some $j'\geq j$, the directed path from $b'_{j'}$ to $z$ on the blue hand, and the directed path from $z$ to $\widehat{b}'_{j'}$ on the face of the clock that contains the arc $(z,\widehat{b}'_0)$. We have thus verified that $X$ is a directed odd cycle transversal of $C$.
\end{proof}

\begin{lemma}\label{lem:revClockNotPrecise}
Let $C$ be a reverse clock. The weight of a set $X\subseteq V(C)$ that is a directed odd cycle transversal of $C$ but does not cut $C$ precisely is at least 40.
\end{lemma}

\begin{proof}
Consider the directed odd cycles of $C$ of Types 1, 2 and 5 (see Observation \ref{obs:reverseCTypes}). The only vertices that are present in at least two of these cycles are of weight $2k+1>40$. Hence, if $w(X)\leq 40$, then $X$ must contain at least three vertices of $C$, each of weight 10. The set $X$ cannot contain exactly three vertices of weight 10 of $C$, since then $w(X)=30$, in which Lemma~\ref{lem:revClockPrecise} implies $X$ should have cut $C$ precisely. Hence, if $w(X)\leq 40$, then $X$ contains at least four vertices of $C$, and therefore $w(X)\geq 40$.
\end{proof}

\begin{figure}[t!]\centering
\fbox{\includegraphics[scale=0.425]{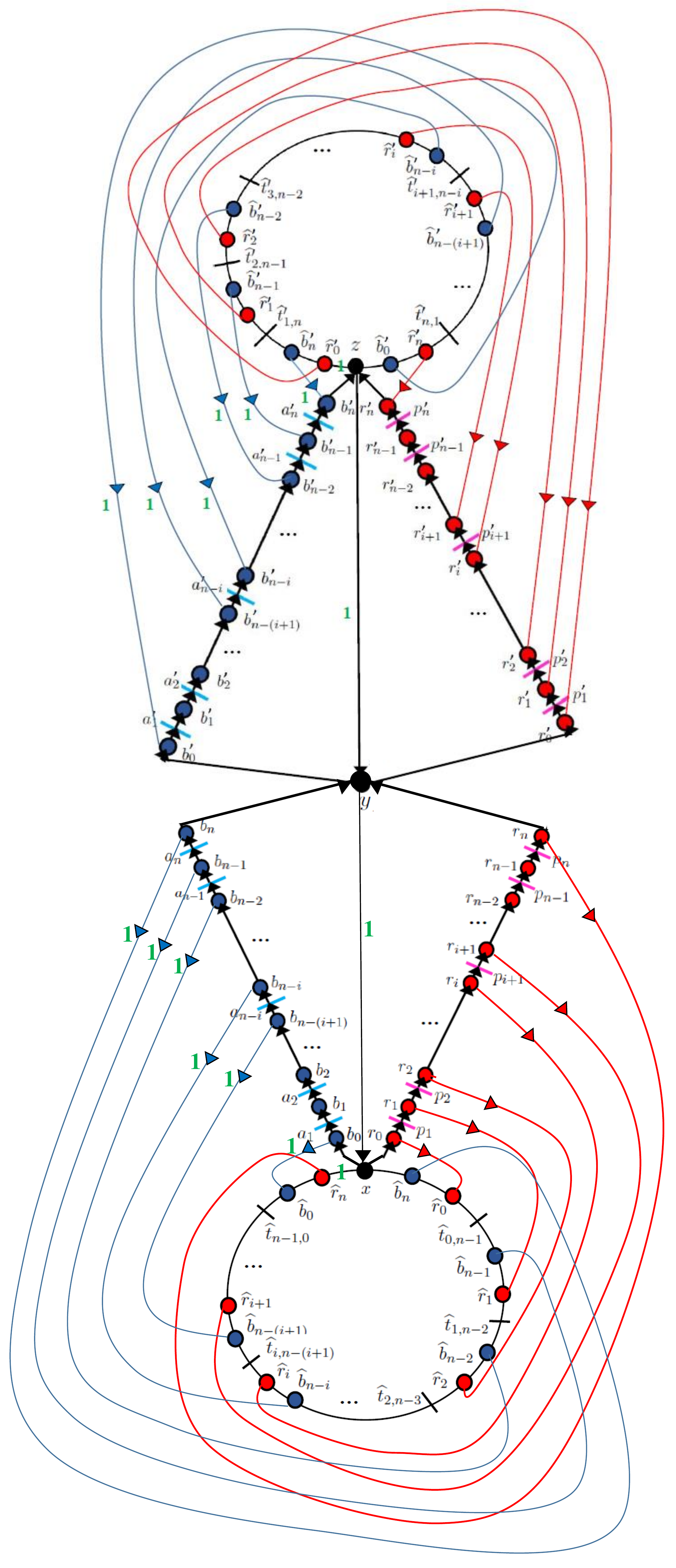}}
\caption{Gluing a forward clock and a reverse clock.}\label{fig:DoubleClock1}
\end{figure}

\subsection{The Double Clock Gadget}

\myparagraph{Structure.} Roughly speaking, an {\em $(n,k)$-double clock} is the result of gluing the tips of the hands of an $(n,k)$-forward clock and an $(n,k)$-reverse clock together as well as adding a 1-labeled arc from every vertex on the hands of the reverse clock to its ``twin'' on the forward clock. In what follows, since $n$ and $k$ would be clear from context, we omit explicit references to $(n,k)$. Formally, a double clock $\widetilde{C}$ is defined as the digraph obtained as follows. Let $C$ be a forward clock, and let $C'$ be a reverse clock. Identify the vertex $y$ of both of these clocks (see Fig.~\ref{fig:DoubleClock1}). All other vertices are distinct. Now, for all $i\in[n]_0$, add the arcs $(r'_i,r_i)$ and $(b'_i,b_i)$, and let the labels of both of these arcs be 1 (see Fig.~\ref{fig:DoubleClock2}).

\begin{figure}[t!]\centering
\fbox{\includegraphics[scale=0.425]{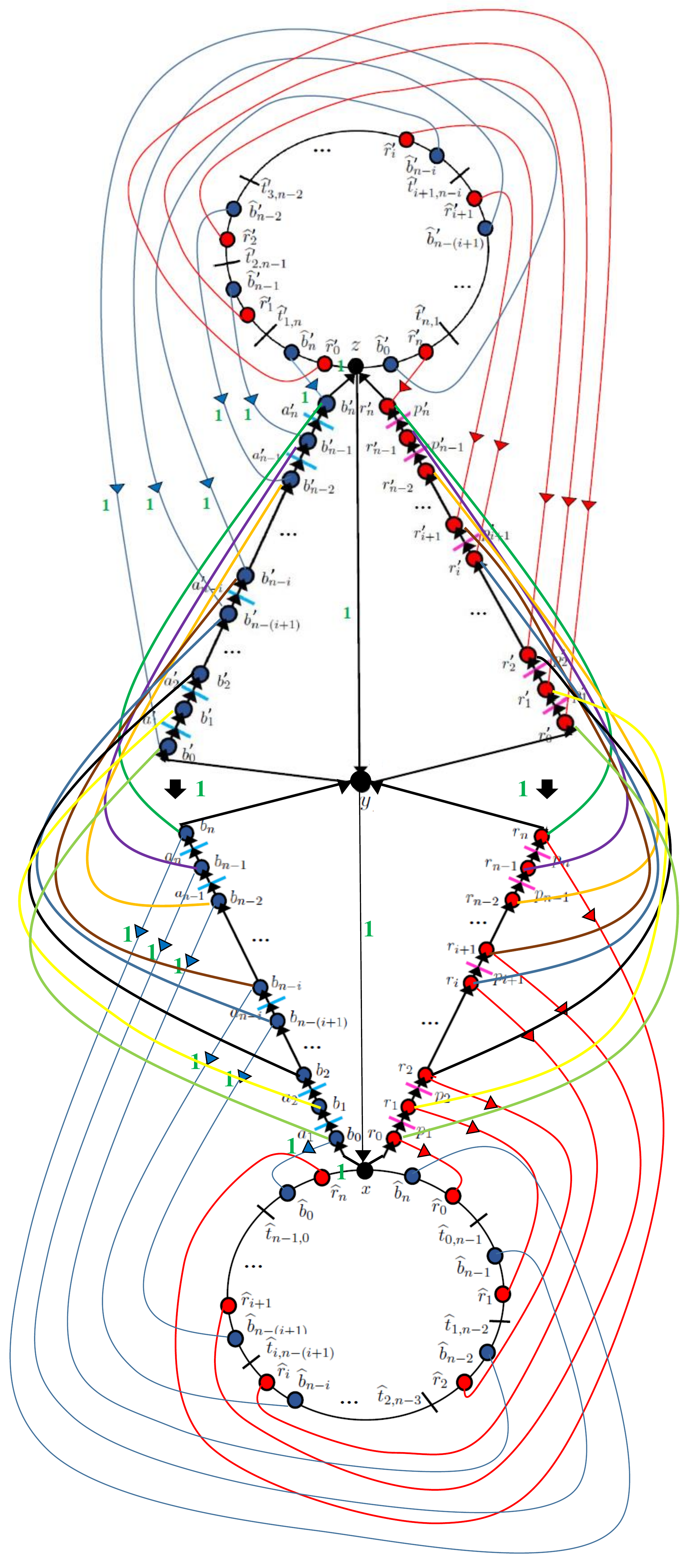}}
\caption{A double clock.}\label{fig:DoubleClock2}
\end{figure}

\smallskip
\myparagraph{Properties.} By the definition of a double clock, we first directly identify which directed odd cycles are present in such a clock.

\begin{observation}\label{obs:doubleCTypes}
Let $\widetilde{C}$ be a double clock. The set of directed odd cycles of $\widetilde{C}$ is the union of the following sets.
\begin{itemize}
\item {\bf Forward:} The set of directed odd cycles completely contained in the forward clock (see Observation \ref{obs:forwardCTypes}).
\item {\bf Reverse:} The set of directed odd cycles completely contained in the reverse clock (see Observation \ref{obs:reverseCTypes}).
\item {\bf Double Red:} For all $i\in[n]_0$, this set contains the direct odd cycle consisting of the arc $(r'_i,r_i)$, the directed path from $r_i$ to $y$ on the red hand of the forward clock, and the directed path from $y$ to $r'_i$ on the red hand of the reverse clock.
\item {\bf Double Blue:} For all $i\in[n]_0$, this set contains the direct odd cycle consisting of the arc $(b'_i,b_i)$, the directed path from $b_i$ to $y$ on the blue hand of the forward clock, and the directed path from $y$ to $b'_i$ on the blue hand of the reverse clock.
\end{itemize}
\end{observation}

We proceed to derive properties of ``cuts'' of a double clock. To this end, we again first need to define the kind of sets using which we would like to ``cut'' double clocks.

\begin{definition}\label{def:clockNiceCutDouble}
Let $\widetilde{C}$ be a double clock. We say that a set $X\subseteq V(\widetilde{C})$ {\em cuts $\widetilde{C}$ precisely} if there exists $i\in[n]$ such that $X=\{p_i,a_{n-i+1},p'_i,a'_{n-i+1},\widehat{t}_{i-1,n-i},\widehat{t}'_{i,n-i+1}\}$.
\end{definition}

We are now ready to present desired properties of ``cuts'' of a double clock.

\begin{lemma}\label{lem:doubleClockPrecise}
Let $\widetilde{C}$ be a double clock. A set $X\subseteq V(\widetilde{C})$ is a directed odd cycle transversal of $\widetilde{C}$ of weight exactly 60 if and only if $X$ cuts $\widetilde{C}$ precisely.
\end{lemma}

\begin{proof}
In the forward direction, let $X\subseteq V(\widetilde{C})$ be a directed odd cycle transversal of $\widetilde{C}$ of weight exactly 60. By Lemmata \ref{lem:forClockPrecise}, \ref{lem:forClockNotPrecise}, \ref{lem:revClockPrecise} and \ref{lem:revClockNotPrecise}, and since $w(X)=60$ while the only vertex that the forward and reverse clocks of $\widetilde{C}$ have in common is $y$, we deduce that $X$ is of the following form. The set $X$ is the union of two pairwise disjoint sets, $Y$ and $Y'$, such that $Y$ cuts the forward clock of $\widetilde{C}$ precisely, and $Y'$ cuts the reverse clock of $\widetilde{C}$ precisely. Hence, there exist $i,j,s\in[n]$ such that $Y=\{p_i,a_j,\widehat{t}_{s,n-s-1}\}$, $i-1\leq s$ and $j\leq n-s$. In particular, by Observation \ref{obs:forClockMaxN1}, it holds that $i+j\leq n+1$. Moreover, there exist $i',j',s'\in[n]$ such that $Y'=\{p'_{i'},a'_{j'},\widehat{t}'_{s',n-s+1}\}$, $s'\leq i'$ and $n-s'+1\leq j'$. In particular, by Observation \ref{obs:revClockMaxN1}, it holds that $i'+j'\geq n+1$. Thus, to prove the forward direction, it remains to show that $i=i'$, $j=j'$ and $j=n-i+1$. Indeed, if these equalities hold, then necessarily $s=i-1$ and $s'=i$.

We claim that $i\geq i'$. Indeed, if $i<i'$, then $X$ does not intersect the directed odd cycle of Type ``Double Red'' (see Observation \ref{obs:doubleCTypes}) that consists of the arc $(\post(p'_i),\post(p_i))=(r'_i,r_i)$, the directed path from $r_i$ to $y$ on the red hand of the forward clock, and the directed path from $y$ to $r'_i$ on the red hand of the reverse clock. Symmetrically, we claim that $j\geq j'$. Indeed, if $j<j'$, then $X$ does not intersect the directed odd cycle of Type ``Double Blue'' that consists of the arc $(\post(a'_j),\post(a_j))=(b'_j,b_j)$, the directed path from $b_j$ to $y$ on the blue hand of the forward clock, and the directed path from $y$ to $b'_j$ on the blue hand of the reverse clock.

Next, we observe that since $i\geq i'$, $j\geq j'$ and $i'+j'\geq n+1$, we have that $i+j\geq n+1$. However, since we have also argued that $i+j\leq n+1$, we have that $i+j=n+1$. Thus, $j=n-i+1$. Similarly, since $i\geq i'$, $j\geq j'$ and $i+j\leq n+1$, we have that $i'+j'\leq n+1$. However, since we have also argued that $i'+j'\geq n+1$, we have that $i'+j'=n+1$. Hence, we further deduce that $i+j=i'+j'$. However, since $i\geq i'$ and $j\geq j'$, this implies that $i=i'$ and $j=j'$. We thus conclude the correctness of the forward direction.

In the reverse direction, let $X\subseteq V(\widetilde{C})$ be a set that cuts $\widetilde{C}$ precisely. Then, there exist $i\in[n]$ such that $X=\{p_i,a_{n-i+1},p'_i,a'_{n-i+1},\widehat{t}_{i-1,n-i},\widehat{t}'_{i,n-i+1}\}$. Denote $Y=\{p_i,a_{n-i+1},\widehat{t}_{i-1,n-i}\}$ and $Y'=\{p'_i,a'_{n-i+1},\widehat{t}'_{i,n-i+1}\}$. Observe that $Y$ cuts the forward clock of $\widetilde{C}$ precisely, while $Y'$ cuts the reverse clock of $\widetilde{C}$ precisely. Hence, by Lemmata \ref{lem:forClockPrecise} and \ref{lem:revClockPrecise}, the set $X$ intersects all directed odd cycles of Types ``Forward'' and ``Reverse''. Thus, by Observation \ref{obs:doubleCTypes}, and since it is clear that $w(X)=60$, it remains to verify that $X$ intersects all directed odd cycles of Types ``Double Red'' and ``Double Blue''. Since $p_i\in X$, it holds that $X$ intersects every directed odd cycle of Type ``Double Red'' that consists of the arc $(r'_{i'},r_{i'})$ for some $i'<i$, the directed path from $r_{i'}$ to $y$ on the red hand of the forward clock, and the directed path from $y$ to $r'_{i'}$ on the red hand of the reverse clock. Moreover, since $p'_i\in X$, it holds that $X$ intersects every directed odd cycle of Type ``Double Red'' that consists of the arc $(r'_{i'},r_{i'})$ for some $i'\geq i$, the directed path from $r_{i'}$ to $y$ on the red hand of the forward clock, and the directed path from $y$ to $r'_{i'}$ on the red hand of the reverse clock. Thus, the set $X$ intersects every directed odd cycle of Type ``Double Red''. Symmetrically, since $a_i,a'_i\in X$, we deduce that $X$ also intersects every directed odd cycle of Type ``Double Blue''. This concludes the proof of the reverse direction.
\end{proof}

\begin{lemma}\label{lem:doubleClockNotPrecise}
Let $\widetilde{C}$ be a double clock. The weight of a set $X\subseteq V(\widetilde{C})$ that is a directed odd cycle transversal of $\widetilde{C}$ but does not cut $\widetilde{C}$ precisely is at least 70.
\end{lemma}

\begin{proof}
Since $X$ does not cut $\widetilde{C}$ precisely, Lemma~\ref{lem:doubleClockPrecise} implies that either $w(X)<60$ or $w(X)>60$. However, Lemmata \ref{lem:forClockPrecise} and \ref{lem:forClockNotPrecise} imply that the weight of the intersection of $X$ with the forward clock is either 30 or at least 40, and Lemmata \ref{lem:revClockPrecise} and \ref{lem:revClockNotPrecise} imply that the weight of the intersection of $X$ with the reverse clock is either 30 or at least 40. Furthermore, since $y$ is the only vertex that the forward and reverse clocks have in common and its weight is $2k+1>70$, we conclude that $w(X)\geq 70$.
\end{proof}

To analyze structures that combine several double clocks, we need to strengthen the reverse direction of Lemma \ref{lem:doubleClockPrecise}. More precisely, we need to derive additional properties of a double clock from which we remove a set that cuts it precisely (in addition to the claim that this graph excludes directed odd cycles).
For this purpose, we first introduce the following definition, which breaks a double clock into three ``pieces'' (see Fig.~\ref{fig:DoubleClock3}).

\begin{definition}\label{def:piecesDouble}
Let $\widetilde{C}$ be a double clock, and let $X$ be a set that cuts $\widetilde{C}$ precisely. Then, $\widetilde{C}[X,x]$ denotes the subgraph of $\widetilde{C}\setminus X$ induced by the set of vertices that both can reach $x$ and are reachable from $x$, and $\widetilde{C}[X,z]$ denotes the subgraph of $\widetilde{C}\setminus X$ induced by the set of vertices that both can reach $z$ and are reachable from $z$. Moreover, $\widetilde{C}[X,y]$ denotes the subgraph of $\widetilde{C}\setminus X$ induced by the set of vertices that belong to neither $\widetilde{C}[X,x]$ nor $\widetilde{C}[X,z]$.
\end{definition}

\begin{figure}[t!]\centering
\fbox{\includegraphics[scale=0.7]{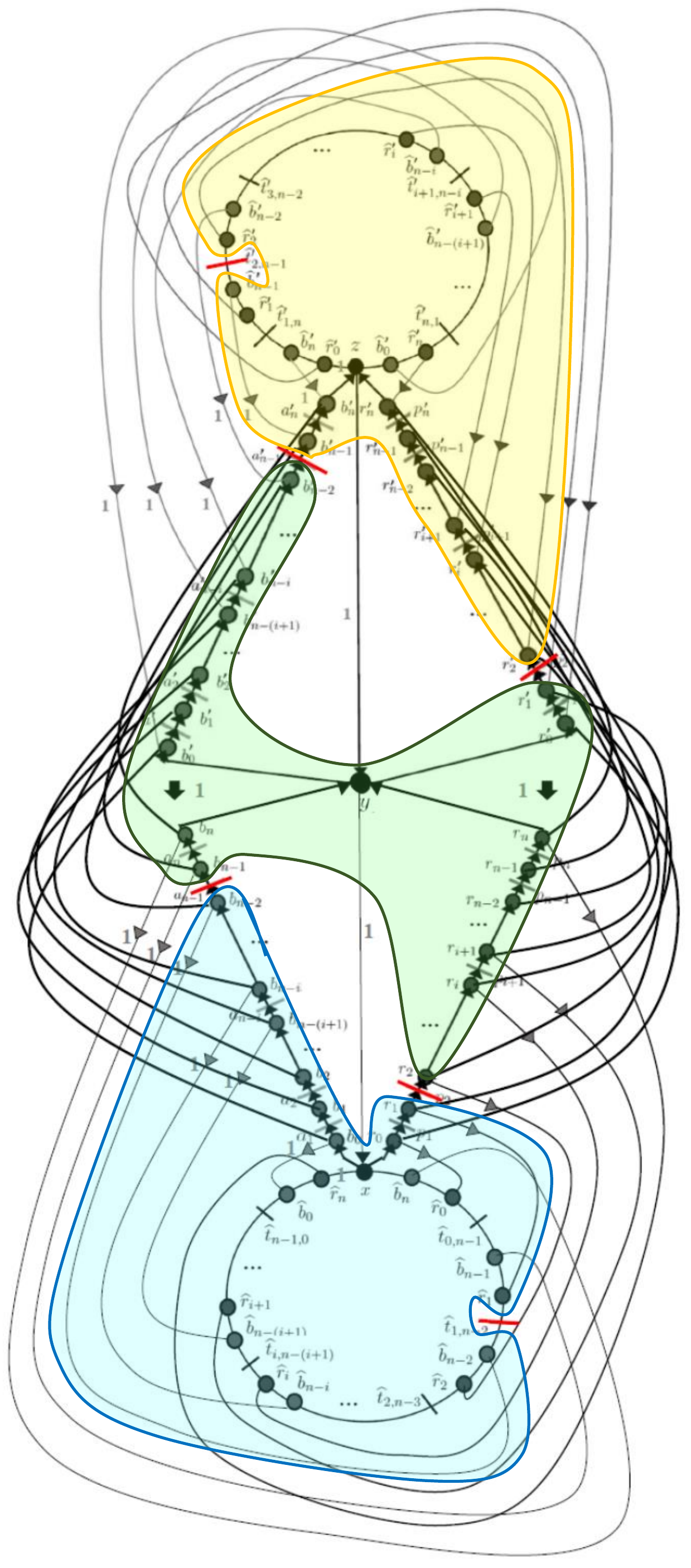}}
\caption{The components $\widetilde{C}[X,x]$ (blue), $\widetilde{C}[X,y]$ (green) and $\widetilde{C}[X,z]$ (yellow), where $X$ is the set of red vertices (see Definition \ref{def:piecesDouble}).}\label{fig:DoubleClock3}
\end{figure}

Notice that for a double clock $\widetilde{C}$, the only two directed paths from $x$ to $y$ are the red and blue hands of the forward clock of $\widetilde{C}$, the only two directed paths from $y$ to $z$ are the red and blue hands of the reverse clock of $\widetilde{C}$, and all of the directed paths from $x$ to $z$ contain the vertex $y$. Moreover, to every vertex $v$ on a hand of the forward clock, there exists exactly one directed path from $x$ that avoids $y$. On the other hand, from the vertex $v$, note that $x$ is reachable by using an arc to the face of the forward clock (in case of a vertex of weight $2k+1$) or an outgoing arc followed by an arc to the face of the forward clock (in case of a vertex of weight 10), after which we append either of the two paths from the vertex we have reached on the face of the forward clock to $x$. A symmetric claim holds in the context of $z$. In light of these observations, we identify each piece of Definition \ref{def:piecesDouble} as follows (see Fig.~\ref{fig:DoubleClock3}).

\begin{observation}\label{obs:vertexSetPieceDouble}
Let $\widetilde{C}$ be a double clock, and let $X=\{p_i,a_{n-i+1},p'_i,a'_{n-i+1},\widehat{t}_{i-1,n-i},\widehat{t}'_{i,n-i+1}\}$ be a set that cuts $\widetilde{C}$ precisely. Then, the following three conditions are satisfied.
\begin{enumerate}
\item $V(\widetilde{C}[X,x])=\widehat{R}\cup\widehat{B}\cup(\widehat{T}\setminus\{\widehat{t}_{i-1,n-i}\})\cup\{x\}\cup\{p_{i'}\in P: i'<i\}\cup\{r_{i'}\in R: i'<i\}\cup\{a_{i'}\in A: i'<n-i+1\}\cup\{b_{i'}\in B: i'<n-i+1\}$.
\item $V(\widetilde{C}[X,y])=\{y\}\cup\{p_{i'}\in P: i'>i\}\cup\{r_{i'}\in R: i'\geq i\}\cup\{a_{i'}\in A: i'>n-i+1\}\cup\{b_{i'}\in B: i'\geq n-i+1\}\cup\{p'_{i'}\in P': i'<i\}\cup\{r'_{i'}\in R': i'<i\}\cup\{a'_{i'}\in A': i'<n-i+1\}\cup\{b'_{i'}\in B': i'< n-i+1\}$.
\item$V(\widetilde{C}[X,z])=\widehat{R}'\cup\widehat{B}'\cup(\widehat{T}'\setminus\{\widehat{t}'_{i,n-i+1}\})\cup\{z\}\cup\{p'_{i'}\in P': i'>i\}\cup\{r'_{i'}\in R': i'\geq i\}\cup\{a'_{i'}\in A': i'>n-i+1\}\cup\{b'_{i'}\in B': i'\geq n-i+1\}$.
\end{enumerate}
\end{observation}

Furthermore, we notice that the three pieces are locally ``isolated'' in a double clock. Here, isolation means that there does not exist a vertex in one piece and a vertex in another piece such that the first vertex can reach the second one and vice versa. More precisely, since $\widetilde{C}[X,x]$ and $\widetilde{C}[X,z]$ are strongly connected directed graphs, while $(z,y),(y,x)\in A(\widetilde{C}\setminus X)$ and all of the directed paths of $\widetilde{C}$ from $x$ to $z$ contain $y$, we directly derive the following observation.

\begin{observation}\label{obs:doubleConComps}
Let $\widetilde{C}$ be a double clock, and let $X$ be a set that cuts $\widetilde{C}$ precisely. Then, the following two conditions are satisfied.
\begin{enumerate}
\item There do not exist vertices $u\in V(\widetilde{C}[X,x])$ and $v\in V(\widetilde{C}[X,y])\cup V(\widetilde{C}[X,z])$ such that there exists a directed path from $u$ to $v$ in $\widetilde{C}\setminus X$.
\item There do not exist vertices $u\in V(\widetilde{C}[X,y])$ and $v\in V(\widetilde{C}[X,z])$ such that there exists a directed path from $u$ to $v$ in $\widetilde{C}\setminus X$.
\end{enumerate}
\end{observation}

We also need to internally analyze each piece separately. To this end, we introduce one additional definition.

\begin{definition}\label{def:ellConsistent}
Let $D$ be a directed graph, and let $\ell: A(D)\rightarrow\{0,1\}$. We say that a function $f: V(D)\rightarrow \{{\bf b},{\bf w}\}$ is {\em $\ell$-consistent for $D$} if for all $(u,v)\in A(D)$, it holds that $\ell(u,v)=0$ if and only if $f(u)=f(v)$.
\end{definition}

When the graph $D$ is clear from context, we simply write $\ell$-consistent rather than $\ell$-consistent for $D$. First, we note the following simple observation, which hints at the relevance of Definition \ref{def:ellConsistent} to {\sc A-DOCT}.

\begin{observation}\label{obs:consistentImplyNoOdd}
Let $D$ be a directed graph, and let $\ell: A(D)\rightarrow\{0,1\}$. If there exists an $\ell$-consistent function for $D$, then $D$ does not contain a directed odd cycle. 
\end{observation}

Let us now derive another simple implication of Definition \ref{def:ellConsistent}.

\begin{lemma}\label{lem:swapEllConsistent}
Let $D$ be a directed graph, $\ell: A(D)\rightarrow\{0,1\}$, and $D'$ be some subgraph of $D$ whose underlying undirected graph is connected and which contains only $0$-labeled arcs. Then, if $D$ admits an $\ell$-consistent function, then $D$ also admits an $\ell$-consistent function $f$ such that for all $v\in V(D')$, it holds that $f(v)={\bf b}$.
\end{lemma}

\begin{proof}
Suppose that $D$ admits an $\ell$-consistent function $\widehat{f}$. Then, define $\widehat{f}': V(D)\rightarrow \{{\bf b},{\bf w}\}$ as follows. For all $v\in V(D)$, it holds that $\widehat{f}'(v)={\bf b}$ if and only if $\widehat{f}(v)={\bf w}$. Note that $\widehat{f}'$ is also $\ell$-consistent for $D$.
By Definition \ref{def:ellConsistent}, for all $u,v\in V(D')$, we have that $\widehat{f}(u)=\widehat{f}(v)$ and $\widehat{f}'(u)=\widehat{f}'(v)$. Thus, if for all $v\in V(D')$, it holds that $\widehat{f}(v)={\bf b}$, then $\widehat{f}$ is a function $f$ as stated in the lemma, and otherwise $\widehat{f}'$ is such a function $f$.
\end{proof}

We proceed by showing that for (arc-labeled) strongly connected directed graphs, we can easily find a consistent function.

\begin{lemma}\label{lem:colorStronglyConnected}
Let $D$ be a strongly connected directed graph, and let $\ell: A(D)\rightarrow\{0,1\}$. If $D$ does not contain a directed odd cycle, then $D$ admits a function $f$ that is $\ell$-consistent.
\end{lemma}

\begin{proof}
Suppose that $D$ does not contain a directed odd cycle. Let $\widehat{D}$ be the directed graph obtained from $D$ by subdividing every $0$-labeled arc once. That is, the graph $\widetilde{D}$ is obtained from $D$ by replacing every arc $a=(u,v)\in A(D)$ that is labeled $0$ by a new vertex $w_a$ and the arcs $(u,w_a)$ and $(w_a,v)$. Let $\widetilde{G}$ be the underlying undirected graph of $\widehat{D}$. Note that $V(D)\subseteq V(\widetilde{D}) = V(\widetilde{G})$. By Proposition \ref{prop:underlyingBipartite}, $\widetilde{G}$ is a bipartite graph. Then, there exists a bipartition $(X,Y)$ of the vertex set of $\widetilde{G}$. Define a function $f: V(D)\rightarrow \{{\bf b},{\bf w}\}$ as follows. For all $v\in V(D)$, it holds that $f(v)={\bf b}$ if and only if $v\in X$.

We claim that $f$ is $\ell$-consistent for $D$. First, note that for every arc $(u,v)\in A(D)$ that is labeled $1$, the edge $\{u,v\}$ belongs to $E(\widetilde{G})$. Thus, since $(X,Y)$ is a bipartition of $\widetilde{G}$, it holds that either both $u\in X$ and $v\in Y$ or both $v\in X$ and $u\in Y$. In either case, we have that $f(u)\neq f(v)$. Now, let $a=(u,v)$ be some arc of $D$ that is labeled 0. Then, $\{u,w_a\},\{w_a,v\}\in E(\widetilde{G})$. Thus, since $(X,Y)$ is a bipartition of $\widetilde{G}$, it holds that either $u,v\in X$ or $u,v\in Y$. In either case, we have that $f(u)=f(v)$. This concludes the proof of the lemma.
\end{proof}

Finally, we are ready to present the last property of a double clock relevant to our work.

\begin{lemma}\label{lem:colorDoubleClock}
Let $\widetilde{C}$ be a double clock, and let $X$ be a set that cuts $\widetilde{C}$ precisely. Then, the following three conditions are satisfied.
\begin{enumerate}
\item\label{lem:colorDoubleClockItem1} There exists an $\ell$-consistent function $f_x$ for $\widetilde{C}[X,x]$ such that $f_x(x)={\bf b}$ and for every vertex $v$ of $\widetilde{C}[X,x]$ that does not belong to the face of the forward clock, it holds that $f_x(v)={\bf b}$.
\item\label{lem:colorDoubleClockItem2} The function that assigns ${\bf b}$ to every vertex of  $\widetilde{C}[X,y]$ is $\ell$-consistent for $\widetilde{C}[X,y]$.
\item\label{lem:colorDoubleClockItem3} There exists an $\ell$-consistent function $f_z$ for $\widetilde{C}[X,z]$ such that $f_z(z)={\bf b}$ and for every vertex $v$ of $\widetilde{C}[X,z]$ that does not belong to the face of the reverse clock, it holds that $f_z(v)={\bf b}$.
\end{enumerate}
\end{lemma}

\begin{proof}
By the definitions of $\widetilde{C}[X,x]$ and $\widetilde{C}[X,z]$, we have that $\widetilde{C}[X,x]$ and $\widetilde{C}[X,z]$ are strongly connected directed graphs. Moreover, by Lemma \ref{lem:doubleClockPrecise}, $\widetilde{C}\setminus X$ does not contain a directed odd cycle, and therefore both $\widetilde{C}[X,x]$ and $\widetilde{C}[X,z]$ do not contain a directed odd cycle. Thus, by Lemma \ref{lem:colorStronglyConnected}, both $\widetilde{C}[X,x]$ and $\widetilde{C}[X,z]$ admit $\ell$-consistent functions. By Observation \ref{obs:vertexSetPieceDouble} and the definition of a double clock, we have that $\widetilde{C}[X,x]$ contains a subgraph whose vertex set consists of $x$ as well as every vertex of $\widetilde{C}[X,x]$ that does not belong to the face of the forward clock, whose underlying undirected graph is connected, and which consists only of $0$-labeled arcs. Symmetrically, we have that $\widetilde{C}[X,z]$ contains a subgraph whose vertex set consists of $z$ as well as every vertex of $\widetilde{C}[X,z]$ that does not belong to the face of the reverse clock, whose underlying undirected graph is connected, and which consists only of $0$-labeled arcs. Hence, by Lemma \ref{lem:swapEllConsistent}, we have that Conditions \ref{lem:colorDoubleClockItem1} and \ref{lem:colorDoubleClockItem3} are satisfied.

Finally, by Observation \ref{obs:vertexSetPieceDouble} and the definition of a double clock, we note that all of the arcs of $\widetilde{C}[X,y]$ are labeled $0$. Thus, it is clear that Condition \ref{lem:colorDoubleClockItem2} is satisfied as well.
\end{proof}

\subsection{The Synchronization Gadget}

Let $n,k\in\mathbb{N}$ such that $k\geq 100$, and let $I\subseteq [n]\times [n]$ be a set of pairs of indices. Here, we define an {\em $(n,k,I)$-synchronizer}. Since $n$ and $k$ would be clear from context, we simply write {\em $I$-synchronizer} rather than $(n,k,I)$-synchronizer. When $I$ is also clear from context (or immaterial), we omit it as well.

\begin{figure}[t!]\centering
\fbox{\includegraphics[scale=0.7]{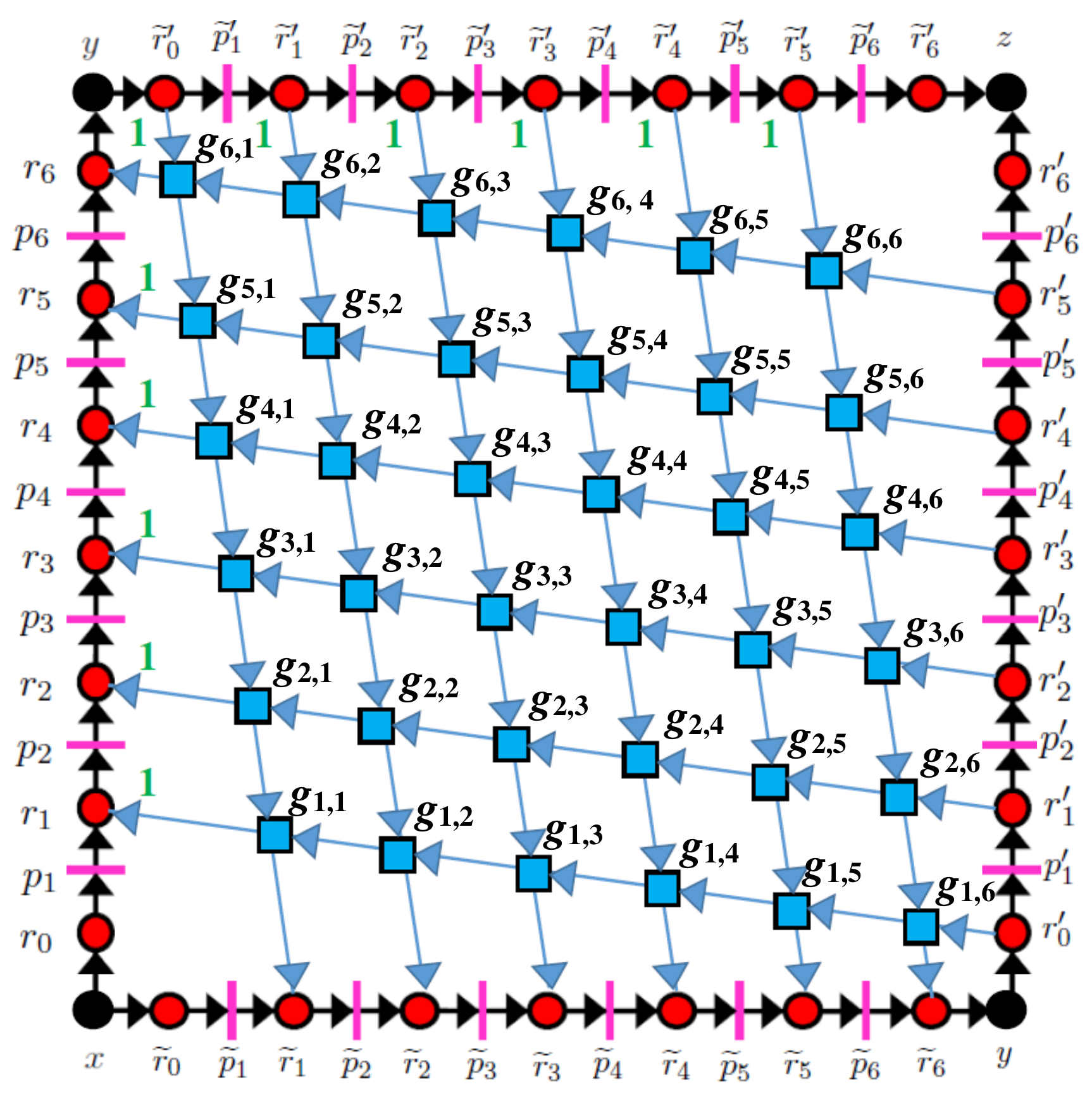}}
\caption{A synchronizer where $n=6$. The arcs labeled 1 are marked by a green `1'. The weight of vertices marked by circles is $2k+1$, the weight of vertices marked by lines is 10, and the weight of each vertex marked by a square is either $2k+1$ or $1$.}\label{fig:cropped_Sync1}
\end{figure}

\myparagraph{Structure.} The {\em hands} of a synchronizer $S$ are four red directed paths, $H$, $H'$, $\widetilde{H}$ and $\widetilde{H}'$. The vertex set of $H$ is the union of two pairwise disjoint sets, $R=\{r_i: i\in[n]_0\}$ (red) and $P=\{p_i: i\in[n]\}$ (pink). For all $i\in[n]$, we denote $\pre(p_i)=r_{i-1}$ and $\post(p_i)=r_i$. The arc set of $H$ is $\{(\pre(p_i),p_i): i\in[n]\}\cup\{(p_i,\post(p_i)): i\in [n]\}$ (see Fig.~\ref{fig:cropped_Sync1}). The path $\widetilde{H}$ is defined as the path $H$ where use tilde notation to specify vertices. Similarly, $H'$ and $\widetilde{H}'$ are defined as the path $H$ and $\widetilde{H}$, respectively, where we further use prime notation to specify vertices. The weight of each vertex on these paths is 10, and the label of each arc on these paths is 0.
Now, to obtain the {\em frame} of $S$, we add three vertices, $x$, $y$ and $z$, each of weight $2k+1$. Moreover, we add the arcs $(x,r_0)$, $(x,\widetilde{r}_0)$, $(r_n,y)$, $(\widetilde{r}_n,y)$, $(y,r'_0)$, $(y,\widetilde{r}'_0)$, $(r'_n,z)$ and $(\widetilde{r}'_n,z)$. The label of each of these arcs is 0. For the sake of clarity of illustrations, the vertex $y$ is drawn twice (see Fig.~\ref{fig:cropped_Sync1}).

Next, we define the {\em interior} of $S$ (see Fig.~\ref{fig:cropped_Sync1}). Roughly speaking, this part is a grid where each vertex has either a very high weight or a very low weight, depending on whether or not the pair of indices that the vertex represents belongs to $I$. Formally, the interior of $S$ is the graph $G$ on the vertex set $\{g_{i,j}: i,j\in [n]\}$ and the arc set $\{(g_{i+1,j},g_{i,j}): i\in[n-1], j\in[n]\}\cup\{(g_{i,j+1},g_{i,j}): i\in[n], j\in[n-1]\} $. The label of each of the arcs is 0. Moreover, for all $i,j\in[n]$, the weight of $g_{i,j}$ is $1$ if $(i,j)\in I$ and $2k+1$ otherwise.

Finally, we attach the frame of $S$ to the interior of $S$ (see Fig.~\ref{fig:cropped_Sync1}). To this end, for all $i\in[n]$, we add two arcs labeled 1: $(g_{i,1},\post(p_i))$ and $(\pre(\widetilde{p}_i'),g_{n,i})$. Moreover, for all $i\in[n]$, we add two arcs labeled 0: $(g_{1,i},\post(\widetilde{p}_i))$ and $(\pre(p_i'),g_{i,n})$. When the synchronizer $S$ is not clear from context, we add the notation $(S)$ to an element (vertex set or vertex) of the synchronizer.

\smallskip
\myparagraph{Properties.} By the definition of a synchronizer, we first directly identify which directed odd cycles are present in such a gadget.

\begin{observation}\label{obs:syncCTypes}
Let $S$ be a synchronizer. The set of directed odd cycles of $S$ is the union of the following sets.
\begin{itemize}
\item {\bf Horizontal Match:} For all $i\in[n]$, this set contains the direct odd cycle consisting of the directed path from $y$ to $\pre(p'_i)$ on $H'$, the (unique) directed path from $\pre(p'_i)$ to $\post(p_i)$ on the interior, and the directed path from $\post(p_i)$ to $y$ on $H$.

\item {\bf Horizontal Mismatch:} For all $i,j\in[n]$ such that $j<i$, this set contains every direct odd cycle consisting of the directed path from $y$ to $\pre(p'_i)$ on $H'$, some directed path from $\pre(p'_i)$ to $\post(p_j)$ on the interior, and the directed path from $\post(p_j)$ to $y$ on $H$.

\item {\bf Vertical Match:} For all $i\in[n]$, this set contains the direct odd cycle consisting of the directed path from $y$ to $\pre(\widetilde{p}'_i)$ on $\widetilde{H}'$, the (unique) directed path from $\pre(\widetilde{p}'_i)$ to $\post(\widetilde{p}_i)$ on the interior, and the directed path from $\post(\widetilde{p}_i)$ to $y$ on $\widetilde{H}$.

\item {\bf Vertical Mismatch:} For all $i,j\in[n]$ such that $j<i$, this set contains every direct odd cycle consisting of the directed path from $y$ to $\pre(\widetilde{p}'_i)$ on $\widetilde{H}'$, some directed path from $\pre(\widetilde{p}'_i)$ to $\post(\widetilde{p}_j)$ on the interior, and the directed path from $\post(\widetilde{p}_j)$ to $y$ on $\widetilde{H}$.
\end{itemize}
\end{observation}

We proceed to derive properties of ``cuts'' of a synchronizer. To this end, we again first need to define the kind of sets using which we would like to ``cut'' synchronizers.

\begin{definition}\label{def:syncCutsPrecise}
Let $S$ be an $I$-synchronizer. We say that a set $X\subseteq V(S)$ {\em cuts $S$ precisely} if there exist $i,j\in[n]$ such that $X=\{p_i,p'_i,\widetilde{p}_j,\widetilde{p}_j',g_{i,j}\}$ and $(i,j)\in I$.
\end{definition}

\begin{definition}\label{def:syncCutsRough}
Let $S$ be an $I$-synchronizer. We say that a set $X\subseteq V(S)$ {\em cuts $S$ roughly} if $X$ does not cut $S$ precisely and there exist $i,j\in[n]$ such that $\{p_i,p'_i,\widetilde{p}_j,\widetilde{p}_j'\}\subseteq X$.
\end{definition}

We are now ready to present desired properties of ``cuts'' of a synchronizer. Unlike the cases of clocks, here we only analyze cuts of the forms presented in Definitions \ref{def:syncCutsPrecise} and \ref{def:syncCutsRough}. The first property follows directly from Definition \ref{def:syncCutsPrecise}.

\begin{observation}\label{obs:syncWeight}
The weight of a set that cuts that a synchronizer precisely is exactly $41$. In particular, the weight of the intersection of this set with the interior is exactly 1.
\end{observation}

Let us now argue that all directed odd cycles are intersected.

\begin{lemma}\label{lem:syncNoOdd}
Let $S$ be a synchronizer, and let $X$ be a set that cuts $S$ precisely. Then, $S\setminus X$ does not contain a  directed odd cycle.
\end{lemma}

\begin{proof}
Since $X$ cuts $S$ precisely, there exist $i,j\in[n]$ such that $X=\{p_i,p'_i,\widetilde{p}_j,\widetilde{p}_j',g_{i,j}\}$. Since $g_{i,j}\in X$, it holds that $X$ intersects the directed odd cycle of Type ``Horizontal Match'' that consists of the directed path from $y$ to $\pre(p'_i)$ on $H'$, the directed path from $\pre(p'_i)$ to $\post(p_i)$ on the interior, and the directed path from $\post(p_i)$ to $y$ on $H$ (see Observation~\ref{obs:syncCTypes}). Moreover, since $p_i,p'_i\in X$, it holds that $X$ intersects all of the remaining directed odd cycles of Types ``Horizontal Match'' and ``Horizontal Mismatch''.
Symmetrically, since $g_{i,j}\in X$, it holds that $X$ intersects the directed odd cycle of Type ``Vertical Match'' that consists of the directed path from $y$ to $\pre(\widetilde{p}'_j)$ on $\widetilde{H}'$, the directed path from $\pre(\widetilde{p}'_j)$ to $\post(\widetilde{p}_j)$ on the interior, and the directed path from $\post(\widetilde{p}_i)$ to $y$ on $\widetilde{H}$. Finally, since $\widetilde{p}_i,\widetilde{p}'_i\in X$, it holds that $X$ intersects all of the remaining directed odd cycles of Types ``Vertical Match'' and ``Vertical Mismatch''. This concludes the proof of the lemma.
\end{proof}

Next, we analyze the weight of rough cuts.

\begin{lemma}\label{lem:syncRough}
Let $S$ be a synchronizer, and let $X$ be a set that cuts $S$ roughly. Then, $w(X)\geq 42$.
\end{lemma}

\begin{proof}
If $X$ contains at least five vertices of the frame, then $w(X)\geq 50$ since the weight of each vertex of the frame is at least 10. Thus, we next assume that $X$ contains exactly four vertices of the frame. Since $X$ cuts $S$ roughly, the set of these vertices can be denoted by $Y=\{p_i,p'_i,\widetilde{p}_j,\widetilde{p}_j'\}$ for some $i,j\in[n]$. Let $C$ be the directed odd cycle of Type ``Horizontal Match'' consisting of the directed path from $y$ to $\pre(p'_i)$ on $H'$, the directed path from $\pre(p'_i)$ to $\post(p_i)$ on the interior, and the directed path from $\post(p_i)$ to $y$ on $H$ (see Observation \ref{obs:syncCTypes}). Moreover, let $C'$ be the directed odd cycle of Type ``Vertical Match'' consisting of the directed path from $y$ to $\pre(\widetilde{p}'_j)$ on $\widetilde{H}'$, the (unique) directed path from $\pre(\widetilde{p}'_j)$ to $\post(\widetilde{p}_j)$ on the interior, and the directed path from $\post(\widetilde{p}_j)$ to $y$ on $\widetilde{H}$. Note that besides vertices of the frame, the only vertex that both $C$ and $C'$ have in common is $g_{i,j}$. Since $X$ must intersect both of these cycles but it cannot contain only $g_{i,j}$ in addition to $Y$ (since it does not cut $S$ precisely), we deduce that $X$ must contain at least two vertices of the interior. We thus conclude that $w(X)\geq 42$.
\end{proof}

As in the case of the double clock, we need to strengthen Lemma \ref{lem:syncNoOdd}. For this purpose, we introduce the following definition, which breaks a synchronizer into five ``pieces'' (see Fig.~\ref{fig:cropped_Sync2}).

\begin{definition}\label{def:syncPieces}
Let $S$ be a synchronizer, and let $X$ be a set that cuts $S$ precisely. Then, $S[X,y]$ denotes the subgraph of $S\setminus X$ induced by the set of vertices that both can reach $y$ and are reachable from $y$. Moreover, $S[X,x]$ denotes the subgraph of $S\setminus X$ induced by the set of vertices reachable from $x$, and $S[X,z]$ denotes the subgraph of $S$ induced by the set of vertices that can reach $z$. Finally, $S[X,l_x]$ denotes the subgraph of $S\setminus X$ induced by the set of vertices outside $S[X,x]$ that can reach a vertex of $S[X,x]$ without using any vertex of $S[X,y]$, and $S[X,l_z]$ denotes the subgraph of $S\setminus X$ induced by the set of vertices outside $S[X,z]$ that are reachable from a vertex of $S[X,z]$ without using any vertex of $S[X,y]$.
\end{definition}

\begin{figure}[t!]\centering
\fbox{\includegraphics[scale=0.7]{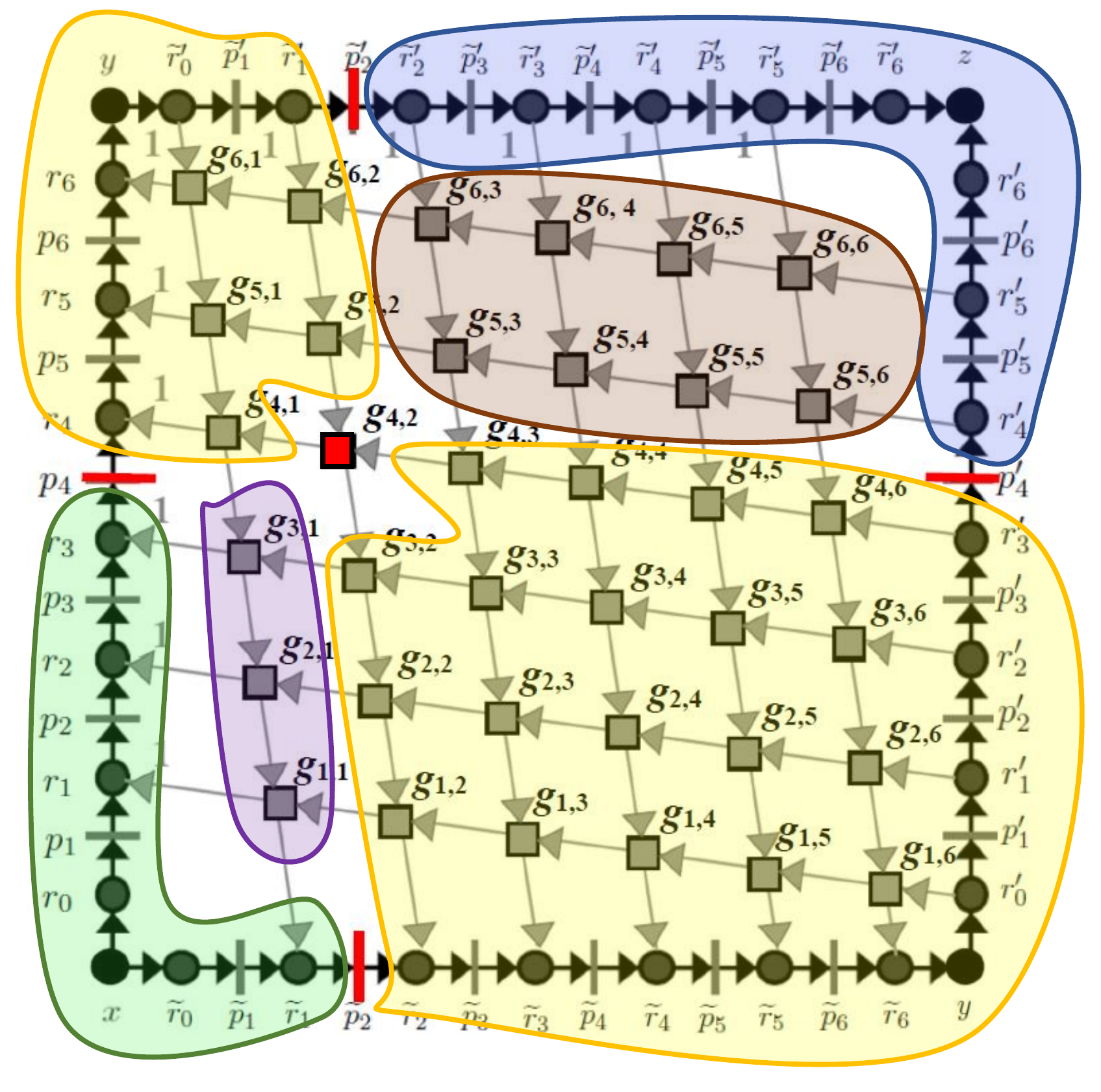}}
\caption{The components $S[X,x]$ (green), $S[X,y]$ (yellow), $S[X,z]$ (blue), $S[X,l_x]$ (purple) and $S[X,l_z]$ (brown), where $X$ is the set of red vertices (see Definition \ref{def:syncPieces}).}\label{fig:cropped_Sync2}
\end{figure}

Notice that for a synchronizer $S$, the only two directed paths from $x$ to $y$ are those internally consisting of $H$ and $\widetilde{H}$, the only two directed paths from $y$ to $z$ are those internally consisting of $H'$ and $\widetilde{H}'$, and all of the directed paths from $x$ to $z$ contain the vertex $y$.
Moreover, the vertex $y$ and every vertex $g_{i,j}$ of the interior are both contained in the two following (even) directed cycles, among other directed cycles, whose only common vertices are $y$ and $g_{i,j}$: (i) the directed cycle consisting of the path from $y$ to $\pre(p'_i)$ on $H'$, the (unique) directed path from $\pre(p'_i)$ to $g_{i,j}$ on the interior, the (unique) directed path from $g_{i,j}$ to $\post(\widetilde{p}_j)$ on the interior, and the directed path from $\post(\widetilde{p}_j)$ to $y$ on $\widetilde{H}$; (ii) the directed cycle consisting of the path from $y$ to $\pre(\widetilde{p}'_j)$ on $\widetilde{H}'$, the (unique) directed path from $\pre(\widetilde{p}'_j)$ to $g_{i,j}$ on the interior, the (unique) directed path from $g_{i,j}$ to $\post(p_i)$ on the interior, and the directed path from $\post(p_i)$ to $y$ on $H$. 
In light of these observations, we identify each piece of Definition \ref{def:syncPieces} as follows (see Fig.~\ref{fig:cropped_Sync2}).

\begin{observation}\label{obs:syncVertexSetPieces}
Let $S$ be a synchronizer, and let $X=\{p_i,p'_i,\widetilde{p}_j,\widetilde{p}_j',g_{i,j}\}$ be a set that cuts $S$ precisely. Then, the following five conditions are satisfied.
\begin{enumerate}
\item $V(S[X,x])=\{x\}\cup\{p_{i'}\in V(H): i'<i\}\cup\{r_{i'}\in V(H): i'<i\}\cup\{\widetilde{p}_{j'}\in V(\widetilde{H}): j'<j\}\cup\{\widetilde{r}_{j'}\in V(\widetilde{H}): j'<j\}$.
\item $V(S[X,y])=\{y\}\cup\{p_{i'}\in V(H): i'>i\}\cup\{r_{i'}\in V(H): i'\geq i\}\cup\{\widetilde{p}_{j'}\in V(\widetilde{H}): j'>j\}\cup\{\widetilde{r}_{j'}\in V(\widetilde{H}): j'\geq j\}\cup\{p'_{i'}\in V(H'): i'<i\}\cup\{r'_{i'}\in V(H'): i'<i\}\cup\{\widetilde{p}'_{j'}\in V(\widetilde{H}'): j'<j\}\cup\{\widetilde{r}'_{j'}\in V(\widetilde{H}'): j'<j\}\cup (\{g_{i',j'}\in V(G): i'\leq i, j'\geq j\}\cup\{g_{i',j'}\in V(G): i'\geq i, j'\leq j\})\setminus\{g_{i,j}\}$.
\item $V(S[X,z])=\{z\}\cup\{p'_{i'}\in V(H'): i'>i\}\cup\{r'_{i'}\in V(H'): i'\geq i\}\cup\{\widetilde{p}'_{j'}\in V(\widetilde{H}'): j'>j\}\cup\{\widetilde{r}'_{j'}\in V(\widetilde{H}'): j'\geq j\}$.
\item $V(S[X,l_x])=\{g_{i',j'}\in V(G): i'<i, j'<j\}$.
\item $V(S[X,l_z])=\{g_{i',j'}\in V(G): i'>i, j'>j\}$.
\end{enumerate}
\end{observation}

Furthermore, in light of Observation \ref{obs:syncVertexSetPieces}, we notice that the five pieces are locally ``isolated'' in a synchronizer as follows.

\begin{observation}\label{obs:syncConComps}
Let $S$ be a synchronizer, and let $X$ be a set that cuts $S$ precisely. Then, the following five conditions are satisfied.
\begin{enumerate}
\item There do not exist vertices $u\in V(\widetilde{C}[X,x])$ and $v\in V(S[X,l_x])\cup V(S[X,y])\cup V(S[X,l_z])\cup V(S[X,z])$ such that there exists a directed path from $u$ to $v$ in $S\setminus X$.
\item There do not exist vertices $u\in V(\widetilde{C}[X,l_x])$ and $v\in V(S[X,y])\cup V(S[X,l_z])\cup V(S[X,z])$ such that there exists a directed path from $u$ to $v$ in $S\setminus X$.
\item There do not exist vertices $u\in V(\widetilde{C}[X,y])$ and $v\in V(S[X,l_z])\cup V(S[X,z])$ such that there exists a directed path from $u$ to $v$ in $S\setminus X$.
\item There do not exist vertices $u\in V(\widetilde{C}[X,l_z])$ and $v\in V(S[X,z])$ such that there exists a directed path from $u$ to $v$ in $S\setminus X$.
\end{enumerate}
\end{observation}

Finally, we need to internally analyze each piece separately.

\begin{lemma}\label{lem:colorSync}
Let $S$ be a synchronizer, and let $X$ be a set that cuts $S$ precisely. Then, the following two conditions are satisfied.
\begin{enumerate}
\item\label{lem:colorSyncItem1} For each of the graphs $S[X,x], S[X,l_x],S[X,l_z]$ and $S[X,z]$, the function that assigns ${\bf b}$ to every vertex of the graph is $\ell$-consistent.

\item\label{lem:colorSyncItem2} There exists an $\ell$-consistent function $f_y$ for $S[X,y]$ such that for every vertex $v$ of the frame that belongs to $S[X,y]$, it holds that $f_y(v)={\bf b}$.
\end{enumerate}
\end{lemma}

\begin{proof}
By Observation \ref{obs:syncVertexSetPieces} and the definition of a synchronizer, we note that all of the arcs of each of the graphs $S[X,x], S[X,l_x],S[X,l_z]$ and $S[X,z]$ are labeled $0$. Thus, it is clear that Condition \ref{lem:colorSyncItem1} is satisfied

Next, by its definition, note that $S[X,y]$ is a strongly connected directed graph. Moreover, by Lemma \ref{lem:syncNoOdd}, $S\setminus X$ does not contain a directed odd cycle, and therefore $S[X,y]$ does not contain a directed odd cycle. Thus, by Lemma \ref{lem:colorStronglyConnected}, $S[X,y]$ admits an $\ell$-consistent function. By Observation \ref{obs:syncVertexSetPieces} and the definition of a synchronizer, we have that $S[X,y]$ contains a subgraph whose vertex set consists of all of the vertices of the frame that belong to $S[X,y]$, whose underlying undirected graph is connected, and which consists only of $0$-labeled arcs. Hence, by Lemma \ref{lem:swapEllConsistent}, we have that Condition \ref{lem:colorSyncItem2} is satisfied as well.
\end{proof}

\subsection{Reduction}\label{sec:reduction}

We are now ready to present the complete reduction from {\sc PSI} to {\sc A-DOCT}. For this purpose, let $(H,G,col)$ be an instance {\sc PSI}. We assume that $|V(G)|\geq 100$, else a solution can be found by brute force in polynomial time. If $G$ contains an isolated vertex to which no vertex in $H$ is mapped by $col$, then the input instance is a \noinstance; otherwise, by removing all of the isolated vertices of $G$ and the vertices of $H$ that are mapped to them, we obtain an instance of {\sc PSI} that is equivalent to $(H,G,col)$. Thus, we next assume that $G$ does not contain isolated vertices. For all $g\in V(G)$, denote $V^g=\{v\in V(H): col(v)=g\}$. We next assume that for all $g,g'\in V(G)$, it holds that $|V^g|=|V^{g'}|=n$ (for the appropriate $n$), else we can add isolated vertices to $H$ to ensure that this condition holds. Then, for all $g\in V^g$, denote $V^g=\{v^g_1,v^g_2,\ldots,v^g_n\}$. Let $<$ be some arbitrary order on $V(G)$.

We construct an instance ${\bf red}(H,G,col)=(D,k,\ell,w)$ of {\sc A-DOCT} as follows. First, we set $k=60|V(G)|+|E(G)|$. Next, we turn to construct $(D,k,\ell,w)$. For every $g\in V(G)$, we insert one $(n,k)$-double clock $\widetilde{C}^g$. For every edge $e=\{g,g'\}\in E(G)$ where $g<g'$, we insert one $(n,k,I^{e})$-synchronizer $S^e$ where $I^e=\{(i,j): \{v^g_i,v^{g'}_j\}\in E(H)\}$. We identify the vertices $x$, $y$ and $z$ of all double clocks and synchronizers. That is, we now have a single vertex called $x$, a single vertex called $y$ and a single vertex called $z$. Finally, for every edge $e=\{g,g'\}\in E(G)$ where $g<g'$, we identify the red hand of the forward clock of $\widetilde{C}^g$ with the hand $H$ of $S^e$, the red hand of the reverse clock of $\widetilde{C}^g$ with the hand $H'$ of $S^e$, the red hand of the forward clock of $\widetilde{C}^{g'}$ with the hand $\widetilde{H}$ of $S^e$, and the red hand of the reverse clock of $\widetilde{C}^{g'}$ with the hand $\widetilde{H}'$ of $S^e$. Here, by identifying two directed paths of the same number $2n+1$ of vertices, we mean that for all $i\in[2n+1]$, we identify the $i$th vertex on one path with the $i$th vertex on the other path. Consequently, for all $i\in[2n]$, we also identify the $i$th arc on one path with the $i$th arc on the other path. We remark that next, when we refer to an element of a specific double clock or a specific synchronizer, we would refer to the new unified vertex. For example, given an edge $e=\{g,g'\}\in E(G)$ where $g<g'$, we have that $r_3(C^g)=r_3(S^e)$ and $r'_5(C^{g'})=\widetilde{r}'_5(S^e)$. This completes the description of the reduction.

\subsection{Correctness}

It remains to derive the correctness of Theorem \ref{thm:doctHard1}. To this end, first note that since $|V(G)|\geq 100$ and $G$ does not contain isolated vertices, we have the following observation.

\begin{observation}\label{obs:budget}
Let $(H,G,col)$ be an instance of {\sc PSI}. Then, for $(D,k,\ell,w)={\bf red}(H,G,col)$, it holds that $100\leq k\leq 121|E(G)|$.
\end{observation}

Clearly, we also have the following observation.

\begin{observation}\label{obs:polyTime}
Let $(H,G,col)$ be an instance of {\sc PSI}. Then, the instance ${\bf red}(H,G,col)$ can be constructed in polynomial time.
\end{observation}

To verify the correctness of the reduction, we first prove the forward direction, summarized in the following lemmata.

\begin{lemma}\label{lem:direction1}
Let $(H,G,col)$ be a \yesinstance\ of {\sc PSI}. Then, $(D,k,\ell,w)={\bf red}(H,G,col)$ is a \yesinstance\ of {\sc A-DOCT}.
\end{lemma}

\begin{proof}
Since $(H,G,col)$ be a \yesinstance\ of {\sc PSI}, there exists a colorful mapping of $G$ into $H$. That is, there exists an injective function $\varphi:V(G)\rightarrow V(H)$ such that for every $g\in V(G)$, $col(\varphi(g))=g$, and for every $\{g,g'\}\in E(G)$, $\{\varphi(g),\varphi(g')\}\in E(H)$. For all $g\in V(G)$, let $i(g)$ denote the index $i\in[n]$ such that $\varphi(g)=v^g_i$. Now, we define the following sets.
\begin{itemize}
\item For all $g\in V(G)$: We define $X^g=\{p_i(\widetilde{C}^g),a_{n-i+1}(\widetilde{C}^g),p'_i(\widetilde{C}^g),a'_{n-i+1}(\widetilde{C}^g),\widehat{t}_{i-1,n-i}(\widetilde{C}^g),$ $\widehat{t}'_{i,n-i+1}(\widetilde{C}^g)\}$ where $i=i(g)$.  Note that $X^g$ cuts the double clock $\widetilde{C}^g$ precisely.
\item For all $e=\{g,g'\}\in E(G)$ where $g<g'$: We define $X^e=\{p_i(S^e),p'_i(S^e),\widetilde{p}_j(S^e),\widetilde{p}_j'(S^e),$ $g_{i,j}(S^e)\}$ where $i=i(g)$ and $j=i(g')$. Since $\{\varphi(g),\varphi(g')\}\in E(H)$, it holds that $(i,j)\in I$. Thus, we have that $X^e$ cuts the synchronizer $S^e$ precisely.
\end{itemize}

Accordingly, define $X=(\bigcup_{g\in V(G)}X^g)\cup(\bigcup_{e\in E(G)}X^e)$. By the definition of ${\bf red}(H,G,col)$ and Observation \ref{obs:syncWeight}, it holds that $w(X)=60|V(G)|+|E(G)|$. We claim that $X$ is a directed odd cycle transversal of $D$, which would imply that $(D,k,\ell,w)$ is a \yesinstance\ of {\sc A-DOCT}. To this end, we also define the following sets.
\begin{itemize}
\item $R_x = \bigcup_{g\in V(G)}V(\widetilde{C}^g[X^g,x])$. Note that $\bigcup_{e\in E(G)}V(S^e[X^e,x])\subseteq R_x$.
\item $R_z = \bigcup_{g\in V(G)}V(\widetilde{C}^g[X^g,z])$. Note that $\bigcup_{e\in E(G)}V(S^e[X^e,z])\subseteq R_z$.
\item $R_y = \bigcup_{e\in E(G)}V(S^e[X^e,y])$. Note that $\bigcup_{g\in V(G)}V(\widetilde{C}^g[X^g,y])\subseteq R_y$.
\item $R^l_{x} = \bigcup_{e\in E(G)}V(S^e[X^e,l_x])$.
\item $R^l_{z} = \bigcup_{e\in E(G)}V(S^e[X^e,l_z])$.
\end{itemize}

Denote ${\cal R}=\{R_x,R_z,R_y,R^l_x,R^l_z\}$. Note that the sets in $\cal R$ are pairwise disjoint. By Observations \ref{obs:doubleConComps} and \ref{obs:syncConComps}, there does not exist a directed odd cycle whose vertex set intersects more than one set from $\cal R$. Moreover, by Lemmata \ref{lem:colorDoubleClock} and \ref{lem:colorSync}, for every $R\in{\cal R}$, there exists a function that is $\ell$-consistent for $D[R\setminus X]$. By Observation \ref{obs:consistentImplyNoOdd}, we deduce that for every $R\in{\cal R}$, $D[R\setminus X]$ does not contain a  directed odd cycle. We thus derive that $X$ is a directed odd cycle transversal of $D$. This concludes the proof of the lemma.
\end{proof}

We next prove a strengthened version of the reverse direction, which would be necessary to derive our inapprixmability result.

\begin{lemma}\label{lem:direction2}
Fix $\epsilon\geq 0$. There exists $\delta=\delta(\epsilon)\geq 0$, where if $\epsilon>0$ then $\delta>0$, such that the following condition holds.
\begin{itemize}
\item  Given an instance $(H,G,col)$ of {\sc PSI}, if ${\bf red}(H,G,col)$ admits a $(1+\delta)$-approximate solution, then there exists a colorful mapping of a subgraph $G'$ of $G$ into $H$ such that $G'$ contains at least $(1-\epsilon)|E(G)|$ edges.
\end{itemize}
\end{lemma}

\begin{proof}
Let $\delta\leq 1$ be determined later. Suppose that $(H,G,col)$ is an instance of {\sc PSI} such that ${\bf red}(H,G,col)$ admits a $(1+\delta)$-approximate solution $X$. We need to show that there exists a subgraph $G'$ of $G$ with at least $(1-\epsilon)|E(G)|$ edges such that there exists a colorful mapping of $G'$ into $H$. Since $X$ is a $(1+\delta)$-approximate solution, it holds that $w(X)\leq (1+\delta)k=(1+\delta)(60|V(G)|+|E(G)|)$. Since $\delta\leq 1$, it holds that $X$ excludes every vertex of $D$ weight $2k+1$. In particular, $X$ does not contain any vertex that is present in more than one double clock. Let ${\cal C}^\star$ be the set of every double clock whose intersection with $X$  cuts it precisely. Moreover, let ${\cal S}$ be the set of every synchronizer such that each of its four hands has been identified with a hand of a double clock from ${\cal C}^\star$. Moreover, let ${\cal S}^\star$ be the set of every synchronizer in $\cal S$ whose intersection with $X$ cuts it precisely.
Denote $c^\star=|{\cal C}^\star|$, $s=|{\cal S}|$ and $s^\star=|{\cal S}^\star|$. Then, by Lemmata \ref{lem:doubleClockPrecise}, \ref{lem:doubleClockNotPrecise}, \ref{obs:syncWeight} and \ref{lem:syncRough}, it holds that
\[w(X)\geq 60c^\star + 70(|V(G)|-c^\star) + 2(s - s^\star) + s^\star = 70|V(G)|-10c^\star + 2s - s^\star.\]
Hence, we have that $70|V(G)|-10c^\star + 2s - s^\star\leq (1+\delta)(60|V(G)|+|E(G)|)$. Thus,
\[2s+70|V(G)|-(1+\delta)(60|V(G)|+|E(G)|)\leq 10c^\star + s^\star.\]

Since $G$ is a graph of maximum degree 3, it holds that $|E(G)|-s\leq 3(|V(G)|-c^\star)$. That is, $s\geq |E(G)|-3(|V(G)|-c^\star)$. Then, by the inequality above,
\[64|V(G)|+2|E(G)|-(1+\delta)(60|V(G)|+|E(G)|)\leq 4c^\star + s^\star.\]
That is,
\[(4-60\delta)|V(G)| + (1-\delta)|E(G)|\leq 4c^\star + s^\star.\]

Since $c^\star\leq|V(G)|$, it holds that $s^\star\geq (1-\delta)|E(G)|-60\delta|V(G)|$. Then, since $G$ does not contain isolated vertices, $|V(G)|\leq 2|E(G)|$. Thus, $s^\star\geq (1-121\delta)|E(G)|$. Fix $\delta=\min\{\epsilon/121,1\}$. Note that if $\epsilon>0$ then $\delta>0$. Overall, it holds that
\[s^\star\geq (1-\epsilon)|E(G)|.\]

Define $E^\star=\{e\in E(G): S^e\in{\cal S}^\star\}$, and let $V^\star$ be the set of every vertex in $V(G)$ that is incident in $G$ to at least one edge in $E^\star$. Note that $|E^\star|=s^\star$. Hence, to conclude that the lemma is correct, it is sufficient to show that $(H,G^\star,col)$ is a \yesinstance. That is, we need to show that there exists an injection $\varphi:V(G^\star)\rightarrow V(H)$ such that for every $g\in V(G^\star)$, $col(\varphi(g))=g$, and for every $\{g,g'\}\in E^\star$, $\{\varphi(g),\varphi(g')\}\in E(H)$. To this end, we define an injection $\psi:V(G^\star)\rightarrow V(H)$ as follows. For all $g\in V(G^\star)$, since $\widetilde{C}^g\in{\cal C}^\star$, we can define $i(g)\in[n]$ as the index such that the intersection of the approximate solution $X$ with $\widetilde{C}^g$ is equal to $\{p_{i(g)},a_{n-i(g)+1},p'_{i(g)},a'_{n-i(g)+1},\widehat{t}_{i(g)-1,n-i(g)},\widehat{t}'_{i(g),n-i(g)+1}\}$. Then, for all $g\in V(G^\star)$, set $\psi(g)=v^g_{i(g)}$, and note that since $v^g_{i(g)}\in V^g$, it holds that $col(\psi(g))=g$. Now, consider some edge $e=\{g,g'\}\in E^\star$ where $g<g'$, and denote $i=i(g)$ and $j=i(g')$. Since $\{g,g'\}\in E^\star$, it holds that the intersection of $X$ with $S^e$ cuts $S^e$ precisely, and therefore $(i,j)\in I^e$. By the definition of $I^e$, we deduce that $\{v^g_i,v^{g'}_j\}\in E(G)$, and thus $\{\psi(g),\psi(g')\}\in E(H)$. This concludes the proof of the lemma.
\end{proof}

As a direct corollary to Lemma \ref{lem:direction2} with $\epsilon=0$, we have the following result.

\begin{corollary}\label{cor:direction2}
If $(H,G,col)$ is a \noinstance\ of {\sc PSI}, then $(D,k,\ell,w)$ is a \noinstance\ of {\sc A-DOCT}.
\end{corollary}

We are now ready to derive the correctness of Theorem~\ref{thm:doctHard1}.

\begin{proof}[Proof of Theorem \ref{thm:doctHard1}]
By Lemma \ref{lem:direction1}, Corollary \ref{cor:direction2} and Observations \ref{obs:budget} and \ref{obs:polyTime}, given any instance $(H,G,col)$ of {\sc PSI}, we can construct in polynomial time an equivalent instance $(D,k,\ell,w)$ of {\sc A-DOCT} such that $k\leq 121|E(G)|$. 
Thus, by Proposition \ref{prop:psiHard}, {\sc A-DOCT} is \WOH. Moreover, by Proposition \ref{prop:psiHard}, unless {\sf ETH} fails, {\sc A-DOCT} cannot be solved in time $f(k)\cdot n^{o(\frac{k} {\log k})}$ for any function $f$. Hence, by Corollary \ref{cor:annotate}, we conclude that Theorem~\ref{thm:doctHard1} is correct.
\end{proof}

\section{Parameterized Inapproximability}\label{sec:paramInapprox}

In this section, we prove Theorem \ref{thm:doctHard2}. For convenience, let us restate the theorem below.

\inapproxTheorem*

We first recall basic concepts concerning constraint satisfaction, after which we reduce $\epsilon$-gap-BCSP to the special case of this problem where every variable occurs in at most three constraints. Then, we turn to conclude the proof of Theorem \ref{thm:doctHard2}.

\subsection{Constraint Satisfaction}

Given a set of variables $X=\{x_1,x_2,\ldots,x_k\}$ and a family of pairwise-disjoint domains ${\cal D}=\{D_1,D_2,\ldots,D_k\}$, a {\em binary constraint} is a pair $c=((x_i,x_j),R)$ where $x_i,x_j\in X$, $i\neq j$, and $R$ is a binary relation over $D_i\times D_j$. An {\em evaluation} is a function $\psi: X\rightarrow \bigcup{\cal D}$ such that for all $x_i\in X$, $\psi(x_i)\in D_i$. An evaluation $\psi$ is said to {\em satisfy} $((x_i,x_j),R)$ if $(\psi(x_i),\psi(x_j))\in R$. Moreover, given a set $C$ of binary constraints, an evaluation $\psi$ is said to {\em satisfy} $C$ if it satisfies every constraint $c\in C$. For all $i\in[k]$, let $C_i\subseteq C$ denote the subset of constraints where $x_i$ occurs, and let $s_i=|C_i|$.
We assume w.l.o.g.~that for all distinct $i,j\in[k]$, $|C_i\cap C_j|\leq 1$, and that for all $i\in[k]$, $|C_i|\geq 1$. In other words, for every pair of variables in $X$, there exists at most one binary constraint in $C$ where both of these variables occur, and for every variable in $X$, there exists at least one binary constraint where it occurs. 

The \csplong\ (\cspshort) is defined as follows.

\begin{center}
\begin{boxedminipage}{.8\textwidth}
\decnamedefn{\csplong\ (\cspshort)}{A set $X=\{x_1,x_2,\ldots,x_k\}$ of $k$ variables, a family of pairwise-disjoint domains ${\cal D}=\{D_1,D_2,\ldots,D_k\}$, and a set $C$ of binary constraints.}
{Does there exist an evaluation that satisfies $C$?}
\end{boxedminipage}
\end{center}

Recall that the promise problem $\epsilon$-{\sc gap-BCSP} is defined as BCSP where the input instance is promised to either be satisfiable, or have the property that every evaluation satisfies less than $(1-\epsilon)$ fraction of the constraints. For a fixed integer $d$, the $\epsilon$-{\sc gap-BCSP$_d$} is defined as the special case of $\epsilon$-{\sc gap-BCSP} every variable is present in at most $d$ constraints. Towards the proof of the hardness of $\epsilon$-{\sc gap-BCSP$_3$}, we first consider $\epsilon$-{\sc gap-BCSP$_4$}. For this proof, we need to recall  the notion of an expander.

\begin{definition}\label{def:expander}
Given $n,d\in\mathbb{N}$ and $0\leq\gamma\leq 1$, an {\em $(n,d,\gamma)$-expander} is an undirected $d$-regular graph $G$ on $n$ vertices such that for every set $S\subseteq V(G)$ of size at most $\frac{1}{2}|V(G)|$, the number of edges with one endpoint in $S$ and the other endpoint in $V(G)\setminus S$ is at least $\gamma\cdot d\cdot|S|$.
\end{definition}

For the sake of brevity, given $n_1,n_2\in\mathbb{N}$, we refer to any  $(n,d,\gamma)$-expander where $n_1\leq n\leq n_2$ as an {\em $([n_1,n_2],d,\gamma)$-expander}. We would also need to rely on the following result.

\begin{proposition}[\cite{DBLP:journals/cpc/AlonSS08}]\label{prop:expander}
There exist $\gamma>0$ and $\ell\in\mathbb{N}$ such that for all $s\in\mathbb{N}$, an $([s,\ell s],3,\gamma)$-expander can be constructed in polynomial time.
\end{proposition}

\begin{lemma}\label{lem:toBCSP4}
Assuming the \PIH, there exists $\epsilon>0$ such that $\epsilon$-{\sc gap-BCSP$_4$} is \W[1]-hard.
\end{lemma}

\begin{proof}
To prove that the lemma is correct, we present a reduction from $\epsilon$-{\sc gap-BCSP} to $\delta$-{\sc gap-BCSP$_4$} where $\delta=\displaystyle{\frac{\epsilon}{4\ell(1+\frac{1}{3\gamma})}}$. Here, $\gamma$ and $\ell$ are the fixed constants stated in Proposition \ref{prop:expander}. For this purpose, let $I=(X,{\cal D},C)$ be an instance of $\epsilon$-{\sc gap-BCSP}. Then, we construct an instance $\widehat{I}=(\widehat{X},\widehat{\cal D},\widehat{C})$ of $\delta$-{\sc gap-BCSP$_4$} as follows. First, for all $i\in[k]$, apply Proposition \ref{prop:expander} to construct an $([s_i,\ell s_i],3,\gamma)$-expander $G_i$ (recall that $s_i=|C_i|$, the number of constraints where $x_i$ occurs), and denote $n_i=|V(G_i)|$. Now, for all $i\in[k]$, define $\widehat{X}^i=\{\widehat{x}^i_1,\widehat{x}^1_2,\ldots,\widehat{x}^i_{n_i}\}$, and let $f^i: \widehat{X}^i\rightarrow V(G_i)$ be an arbitrarily chosen bijective function. Accordingly, set $\widehat{X}=\bigcup_{i=1}^k\widehat{X}^i$. Notice that since $|C_i\cap C_j|\leq 1$ for all distinct $i,j\in[k]$, we have that $|C|\leq k^2$. Moreover, since the constraints are binary, we have that $\displaystyle{\sum_{i=1}^ks_i}= 2|C|$. Therefore, $|\widehat{X}|=\displaystyle{\sum_{i=1}^k}n_i\leq \displaystyle{\ell\sum_{i=1}^ks_i}=2\ell|C|\leq 2\ell k^2$. That is, the number of variables in the new instance is bounded by a function of $k$.

We proceed to define $\widehat{D}$ by letting the domain of every $\widehat{x}^i_j$, where $i\in[k]$ and $j\in[n_i]$ be $\widehat{D}^i_j=\{\widehat{d}^i_j: d\in D_i\}$. Finally, let us define $\widehat{C}$ as follows. For all $i\in[k]$, let $g^i: C_i\rightarrow X^i$ be an arbitrarily chosen injective function. Then, for all $c=((x_i,x_j),R)\in C$, denote $\widehat{c}=((x^i_p,x^j_q),\widehat{R}=\{(\widehat{d}^i_p,\widehat{d'}^j_q): (d,d')\in R\})$ where $x^i_p=g^i(x_i)$ and $x^j_q=g^j(x_j)$. Define $C^\star=\{\widehat{c}: c\in C\}$. We how define a set of equality constraints on the set of `copies' of each variable. Define  $C_{=}=\{((x^i_p,x^i_q),\{(d^i_p,d^i_q): d\in D_i\}): i\in[k], p,q\in [n_i], (p,q)\in E(G_i)\}$. Finally, set $\widehat{C}=C^\star\cup C_{=}$. Since for all $i\in[k]$, the graph $G_i$ is 3-regular, we have that every variable in $\widehat{X}^i$ occurs in at most three constraints in $C_{=}$. Moreover, since for all $i\in[k]$ the function $g^i$ is injective, we have that every variable in $\widehat{X}^i$ occurs in at most one constraint in $C^\star$. Thus, every variable in $\widehat{X}$ occurs in at most four constraints in total. Thus, since every constraint is binary, we also have that $|\widehat{C}|\leq 2|\widehat{X}|$. Since we have already shown that $|\widehat{X}|\leq 2\ell|C|$, we derive that $|\widehat{C}|\leq 4\ell|C|$.

To prove that the reduction is correct, we argue that if $I$ is satisfiable (admits an evaluation satisfying all constraints) then so is $\widehat{I}$ and conversely, if $\widehat{I}$ admits an evaluation  that satisfies at least $(1-\delta)|\widehat{C}|$ constraints, then $I$ admits an evaluation that satisfies at least $(1-\epsilon)|C|$ constraints.

Suppose that $I$ admits an evaluation $\psi$ that satisfies all of the constraints. Then, we have that $\widehat{I}$ also admits an evaluation $\widehat{\psi}$ that satisfies all of the constraints. Indeed, we simply define $\widehat{\psi}$ by setting $\widehat{\psi}(x^i_j)=d^i_j$, where $d=\psi(x_i)$, for all $i\in[k]$ and $j\in [n_i]$.

Next, suppose that $\widehat{I}$ admits an evaluation $\widehat{\psi}$ that satisfies at least $(1-\delta)|\widehat{C}|$ constraints. We define an evaluation $\psi$ for $I$ as follows. For all $i\in[k]$ and $d\in D_i$, define $X^i(d)=\{x^i_j\in X^i: \widehat{\psi}(x^i_j)=d^i_j\}$. Now, for all $i\in[k]$, let $\widetilde{d}^i$ be a value in $D_i$ that among all values in $D_i$, maximizes $|X^i(d)|$ (if there is more than one choice, choose one arbitrarily). Moreover, for all $i\in[k]$, denote $Y^i=X^i\setminus X^i(\widetilde{d}^i)$. Then, for all $i\in[k]$, we set $\psi(x_i)=\widetilde{d}^i$. We claim that $\psi$ satisfies at least $(1-\epsilon)|C|$ of the constraints in $C$. To show this, we first note that all $\widehat{c}=((x^i_p,x^i_q),R)$ such that either both $\widehat{\psi}(x^i_p)=\widetilde{d}^i$ and $\widehat{\psi}(x^i_q)\neq\widetilde{d}^i$ or both $\widehat{\psi}(x^i_p)\neq\widetilde{d}^i$ and $\widehat{\psi}(x^i_q)=\widetilde{d}^i$, a unique constraint in $C_{=}$ is violated by $\widehat{\psi}$. Since for all $i\in[k]$, $G_i$ is an $([s_i,\ell s_i],3,\gamma)$-expander, we derive that at least $\displaystyle{\sum_{i=1}^k\gamma\cdot 3\cdot|Y^i|}=\displaystyle{3\gamma\sum_{i=1}^k|Y^i|}$ constraints in $C_{=}$ are violated by $\widehat{\psi}$. Thus, $\displaystyle{3\gamma\sum_{i=1}^k|Y_i|<\delta|\widehat{C}|}$, which implies that $\displaystyle{\sum_{i=1}^k|Y_i|\leq \frac{\delta}{3\gamma}|\widehat{C}|}$. Let us now denote by $C_Y$ the set of all constraints $c=((x_i,x_j),R)\in C$ such that $g^i(c)\in Y^i$ or $g^j(c)\in Y_j$. Then, since for all $i\in[k]$, $g^i$ is an injective function, we have that $|C_Y|\leq \displaystyle{\sum_{i=1}^k|Y_i|}$, and thus we derive that $|C_Y|\leq \displaystyle{\frac{\delta}{3\gamma}|\widehat{C}|}$. Notice that for all $c\in C\setminus C_Y$ that is violated by $\psi$, there is a unique constraint in $C^\star$ that is also violated by $\widehat{\psi}$. Thus, the number of constraints in $C\setminus C_Y$ that are violated by $\psi$ is smaller than $\delta|\widehat{C}|$. Overall, we conclude that $\psi$ violates less than $\delta|\widehat{C}|+|C_Y|\leq \displaystyle{(1+\frac{1}{3\gamma})\delta|\widehat{C}|}\leq 4\ell\displaystyle{(1+\frac{1}{3\gamma})\delta|C|}$ constraints in $C$. Since $\delta=\displaystyle{\frac{\epsilon}{4\ell(1+\frac{1}{3\gamma})}}$, we conclude that $\psi$ violates less than $\epsilon|C|$ constraints in $C$. This concludes the proof of the lemma.
\end{proof}

We now turn to prove the hardness of $\epsilon$-{\sc gap-BCSP$_3$}.

\begin{lemma}\label{lem:toBCSP3}
Assuming the \PIH, there exists $\epsilon>0$ such that $\epsilon$-{\sc gap-BCSP$_3$} is \W[1]-hard.
\end{lemma}

\begin{proof}
By Lemma \ref{lem:toBCSP4}, we have that there exists $\epsilon>0$ such that $\epsilon$-{\sc gap-BCSP$_4$} is \W[1]-hard. To prove that the lemma is correct, we present a reduction from $\epsilon$-{\sc gap-BCSP$_4$} to $\delta$-{\sc gap-BCSP$_3$} where $\delta=\epsilon/15$. For this purpose, let $I=(X,{\cal D},C)$ be an instance of $\epsilon$-{\sc gap-BCSP$_4$}. Then, we construct an instance $\widehat{I}=(\widehat{X},\widehat{\cal D},\widehat{C})$ of $\delta$-{\sc gap-BCSP$_3$} as follows. First, define $\widehat{X}=X\cup X'$ where $X'=\{x_1',x_2',\ldots,x_k'\}$. Hence, $|\widehat{X}|\leq 2k$. Now, let us define $\widehat{\cal D}$. For all $i\in[k]$, the domain of $x_i$ is defined as $D_i\in{\cal D}$, and the domain of $x'_i$ is defined as $D_i'=\{d': d\in D_i\}$. Now, for all $i\in[k]$, let $(A_i,B_i)$ be a partition of $C_i$ such that $|A_i|,|B_i|\leq 2$. Note that the existence of such a partition follows from the fact that $|C_i|\leq 4$. Then, for all $c=((x_i,x_j),R)\in C$, denote $\widehat{c}=((x_i,x_j),R)$ if $((x_i,x_j),R)\in A_i\cap A_j$, $\widehat{c}=((x_i,x_j'),R'=\{(p,q'): (p,q)\in R\})$ if $((x_i,x_j),R)\in A_i\cap B_j$, $\widehat{c}=((x_i',x_j),R'=\{(p',q): (p,q)\in R\})$ if $((x_i,x_j),R)\in B_i\cap A_j$ and $\widehat{c}=((x_i',x_j'),R'=\{(p',q'): (p,q)\in R\})$ otherwise.
Define $C^\star=\{\widehat{c}: c\in C\}$. We now define a set of equality constraints on each pair of copies of the variables. Define  $C_{=}=\{((x_i,x'_i),\{(d,d'): d\in D_i\}): i\in[k]\}$. Finally, set $\widehat{C}=C^\star\cup C_{=}$. Clearly, every variable in $\widehat{X}$ occurs in at most three constraints in $\widehat{C}$, and the construction can be performed in polynomial time.

To prove that the reduction is correct, we argue that if $I$ is satisfiable then so is $\widehat{I}$ and conversely, if $\widehat{I}$ admits an evaluation  that satisfies at least $(1-\delta)|\widehat{C}|$ constraints, then $I$ admits an evaluation that satisfies at least $(1-\epsilon)|C|$ constraints.

Suppose that $I$ admits an evaluation $\psi$ that satisfies all of the constraints. Then, we have that $\widehat{I}$ also admits an evaluation $\widehat{\psi}$ that satisfies all of the constraints. Indeed, we simply define $\widehat{\psi}$ by setting $\widehat{\psi}(x_i)=\psi(x_i)$ and $\widehat{\psi}(x_i')=\psi(x_i)'$ for all $i\in[k]$.

Next, suppose that $\widehat{I}$ admits an evaluation $\widehat{\psi}$ that satisfies at least $(1-\delta)|\widehat{C}|$ constraints. We define an evaluation $\psi$ for $I$ as follows. For all $i\in[k]$, we set $\psi(x_i)=\widehat{\psi}(x_i)$. We claim that $\psi$ satisfies at least $(1-\epsilon)|C|$ of the constraints in $C$. To show this, we denote $Y=\{x_i\in X: \widehat{\psi}(x_i)\neq \widehat{\psi}(x_i')\}$. Since $\widehat{\psi}$ violates less than $\delta|\widehat{C}|$ constraints, we have that $|Y|<\delta|\widehat{C}|$. Let $C_Y$ denote the subset of constraints of $C$ where at least one variable of $Y$ occurs.
Note that since every variable in $X$ occurs in at most four constraints in $C$, we have that $|C_Y|\leq 4|Y|\leq 4\delta|\widehat{C}|$. Moreover, note that for every constraint $\widehat{c}\in C^\star\setminus C_Y$ that is satisfied by $\widehat{\psi}$ is also satisfied by $\psi$. Since $\widehat{\psi}$ violates less than $\delta|\widehat{C}|$ constraints in total, we have that $\widehat{\psi}$ also violates less than $\delta|\widehat{C}|$ constraints from $C^\star\setminus C_Y$. Thus, we have that $\psi$ violates less than $5\delta|\widehat{C}|$ constraints from $C$. Since every variable occurs in at least one constraint, we have that $|\widehat{C}|=|C^\star|+|C_{=}|=|C|+|X|\leq 3|C|$. We thus conclude that violates less than $15\delta|C|=\epsilon|C|$ constraints from $C$. This concludes the proof of the lemma.
\end{proof}

\subsection{Proof of Theorem~\ref{thm:doctHard2}}

We now translate Hypothesis \ref{conjecture} in terms of {\sc PSI}. For this purpose, we define the promise problem $\epsilon$-{\sc gap-PSI} as {\sc PSI} where the input instance is promised to either be a \yesinstance, or have the property that for every subgraph $G'$ of $G$ with at least $(1-\epsilon)|E(G)|$ edges, there does not exist a colorful mapping of $G'$ into $H$. It is straightforward to see that if $\epsilon$-{\sc gap-BCSP} is \W[1]-hard, then $\epsilon$-{\sc gap-PSI} is \W[1]-hard as well. For the sake of completeness, we present the reduction.

\begin{lemma}\label{lem:equivConjecture}
Assuming the \PIH and \FPT $\neq$ \W[1], there exists $\epsilon>0$ such that $\epsilon$-{\sc gap-PSI} is \W[1]-hard.
\end{lemma}

\begin{proof}[Sketch.]
Let $(X,{\cal D},C)$ be an instance of $\epsilon$-{\sc gap-BCSP$_3$}. Then, we construct (in polynomial time) an instance $(H,G,col)$ of {\sc PSI} as follows. First, we set $V(G)=X$ and $E(G)=\{\{x,x'\}:$ there exists $R$ such that $((x,x'),R)\in C\}$. Second, we set $V(H)=\bigcup {\cal D}$, and we let $E(H)$ contain every edge $\{d,d'\}$ for which there exist $i\neq j$ and $R$ such that $d\in D_i$, $d'\in D_j$, $((x_i,x_j),R)\in C$ and $(d,d')\in R$. Finally, for all $i\in[k]$ and $d\in D_i$, we set $col(d)=x_i$. Notice that for all $0\leq \alpha\leq 1$, there exists $C'\subseteq C$ of size at least $\alpha|C|$ such that there exists an evaluation satisfying $C'$ if and only if there exists a subgraph $G'$ of $G$ with at least $\alpha|E(G)|$ edges such that there exists a colorful mapping of $G'$ into $H$. Hence, by Lemma \ref{lem:toBCSP3}, we conclude that the lemma is correct.
\end{proof}

We remark that it is also straightforward to see that if $\epsilon$-{\sc gap-PSI} is \W[1]-hard, then $\epsilon$-{\sc gap-BCSP$_3$} is \W[1]-hard, and hence $\epsilon$-{\sc gap-BCSP} is \W[1]-hard as well.

Finally, we ready to prove the correctness of Theorem~\ref{thm:doctHard2}. 

\begin{proof}[Proof of Theorem \ref{thm:doctHard2}]
Suppose that there exists $\epsilon>0$ for which there does not exist an \FPT algorithm for $\epsilon$-{\sc gap-PSI}. Let $\delta=\delta(\epsilon)$ be defined according to Lemma \ref{lem:direction2}. We claim that {\sc A-DOCT} does not admit a $(1+\delta)$-approximation algorithm that runs in time $f(k)\cdot n^{\bigoh(1)}$ for any function $f$. By Lemma \ref{lem:equivConjecture} and Corollary \ref{cor:annotate}, we would thus conclude the correctness of Theorem~\ref{thm:doctHard2}. Suppose, by way of contradiction, that our claim is false. Then, let $\cal B$ be a $(1+\delta)$-approximation algorithm for {\sc A-DOCT} that runs in time $f(k)\cdot n^{\bigoh(1)}$ for some function $f$. We define an algorithm, Algorithm $\cal A$ as follows. Given an instance $(H,G,col)$ of {\sc PSI}, it constructs the instance $(D,k,\ell,w)$ of {\sc A-DOCT} as described in Section \ref{sec:reduction}, and calls Algorithm $\cal B$ with $(D,k,\ell,w)$ as input. Then, if $\cal B$ outputs \No, then $\cal A$ outputs \No, and otherwise $\cal A$ outputs \Yes. By Observations \ref{obs:budget} and \ref{obs:polyTime}, Algorithm $\cal A$ runs in time $f(|E(H)|)\cdot |I|^{\bigoh(1)}$ where $|I|$ is the size of the input instance. On the one hand, by Lemma \ref{lem:direction1}, if Algorithm $\cal A$ is given as input a \yesinstance\ of {\sc PSI}, then it constructs a \yesinstance\ of {\sc A-DOCT}. Next, since $\cal B$ is a $(1+\delta)$-approximation algorithm for {\sc A-DOCT}, Algorithm $\cal A$ outputs \Yes. On the other hand, suppose that Algorithm $\cal A$ is given as input an instance $(H,G,col)$ of {\sc PSI} for which there does not exist a subgraph $G'$ of $G$ with at least $(1-\epsilon)|E(G)|$ edges such that there exists a colorful mapping of $G'$ into $H$. By Lemma \ref{lem:direction2}, Algorithm $\cal A$ constructs an instance $(D,k,\ell,w)$ of {\sc A-DOCT} which does not admit a $(1+\delta)$-approximate solution. Then, since $\cal B$ is a $(1+\delta)$-approximation algorithm for {\sc A-DOCT}, Algorithm $\cal A$ outputs \No. We have thus reached a contradiction to the choice of $\epsilon$. This concludes the proof of Theorem~\ref{thm:doctHard2}.
\end{proof}

\section{Parameterized Approximation}

\newcommand{\restrictedcwc}{{\sc Restricted Colorful Walk Cover}}
\newcommand{\cwcomp}{{\sc Colorful Walk Cover Compression}}

\newcommand{\cwc}{{\sc Colorful Walk Cover}}
\newcommand{\bundle}{{\sf Bundle}}
\newcommand{\cZ}{{\cal Z}}

  We are now ready to begin the section on the parameterized approximation algorithm for {\sc DOCT}. We will in fact prove a more general result by giving a parameterized approximation algorithm for a covering problem on labeled digraphs which generalizes both {\sc DOCT} and the {\sc Node Unique Label Cover} problem~\cite{ChitnisCHPP16}. Before we formally define labeled digraphs, we need the following notation.

For $\ell\in {\mathbb N}$, $S_\ell$ denotes the symmetric group, which is the set of all permutations of the set $[\ell]$. For $\sigma_1,\sigma_2\in S_\ell$, we denote by $\sigma_1\circ \sigma_2$ the permutation obtained by composing $\sigma_1$ and $\sigma_2$ as follows. For every $i\in [\ell]$, $\sigma_1\circ \sigma_2 (i)=\sigma_2( \sigma_2(i))$. Finally, unless otherwise specified, all paths and walks we refer to in this section are directed.

Let $D$ be a digraph with possible self-loops. A {\em cycle cover} of $D$ is a set $\{C_1,\dots, C_r\}$ of vertex-disjoint cycles such that every vertex is part of some cycle. That is, $V(D)=\bigcup_{i\in [r]} V(C_i)$. We point out that cycles of length 1 which correspond to self-loops are allowed to be part of a cycle cover. It is easy to see that the cycle covers of $V(D)$ are in one-to-one correspondence with the permutations of $V(D)$.

\begin{definition}
	Let $D$ be a digraph and $\sigma:A(D)\to S_\ell$ be an assignment of permutations of $[\ell]$ to the arcs of $D$. We call the pair $(D,\sigma)$ a {\em labeled digraph}. When $\sigma$ is clear from the context, we just use $D$ to denote the labeled digraph $(D,\sigma)$. For a set $Z\subseteq V(D)$, we denote by $\sigma|_{Z}$, the restriction of $\sigma$ to arcs with both endpoints in $Z$.
	For a directed walk  $P=v_1,\dots,v_r$ in $D$, we denote by $\sigma(P)$ the composed permutation $\sigma((v_1,v_2))\circ \dots \circ \sigma((v_{r-1},v_r))$.
		\end{definition}

	In all labeled digraphs we work with in this section, there are no {\em duplicate} arcs. That is, there is no pair $a,a'=(u,v)\in A(D)$ such that $\sigma(a)=\sigma(a')$. However, we may have multiple arcs from $u$ to $v$ labeled with distinct elements of $S_\ell$. We are now ready to define our main combinatorial structure which, as we will show, generalizes odd cycles in directed graphs.

	\begin{definition}
		Let $(D,\sigma)$ be a labeled digraph. Let $H$ be a strongly connected subgraph of $D$ and let $v\in V(H)$. We say that $H$ is a $v$-{\em colorful walk} if for every $i\in [\ell]$,  
		$H$ contains a closed $v$-walk $W$ such that $\sigma(W)(i)\neq i$.
	\end{definition}

	\begin{definition}
		Let $(D,\sigma)$ be a labeled digraph and $H$ be a strongly connected subgraph of $D$. We say that $H$ is a {\em colorful walk} if it is a $v$-colorful walk for every $v\in V(H)$. A set $S\subseteq V(D)$ intersecting every colorful walk  contained in $D$ is called a {\em colorful walk cover} of $D$. The explicit reference to $D$ is ignored when $D$ is clear from the context.
	\end{definition}

	Before we proceed, we provide a short proof of the fact that colorful walks generalize odd cycles and odd closed walks in digraphs. 
Let $\ell=2$ and let $\pi$ denote the permutation $(1$ $2)$. That is, $\pi(1)=2$ and $\pi(2)=1$. Furthermore, 
suppose that for a digraph $D$,  we assign $\sigma(a)=\pi$ for every $a\in A(D)$. Then, observe that $D$ has a colorful walk if and only if it has an odd  closed walk. Therefore, colorful walks generalize odd closed walks and since a directed graph has an odd cycle if and only if it has an odd closed walk,  the following problem is a  generalization of {\sc DOCT}.

	\medskip

\defparproblem{{\sc Colorful Walk Cover}}{A digraph $D$, function $\sigma:A(D)\to S_\ell$, and integer $k$.}{$k,\ell$.}{Does $D$ have a colorful walk cover of size at most $k$? 
}

\medskip

Due to Theorem \ref{thm:doctHard1} and the fact that {\sc DOCT} is a special case of {\sc Colorful Walk Cover} when $\ell=2$, we get the following corollary. 

\begin{corollary}
	The {\sc Colorful Walk Cover} problem is {\WOH}  even for $\ell=2$. 
	Furthermore, assuming the ETH there is no algorithm for {\sc Colorful Walk Cover} with running time $f(k,\ell)n^{o(k/\log k)\cdot g(\ell)}$ for any functions $f$ and $g$.
	\end{corollary}

We now complement this negative result with a positive approximation result. Recall that due to \cite{KhotV15}, we know that there is no constant factor approximation for the \emph{optimization} version of {\sc DOCT} and hence for {\sc Colorful Walk Cover} assuming the Unique Games Conjecture.  
However, we show that if allowed {\FPT} time, then the optimization version of even {\sc Colorful Walk Cover} can be  approximated up to a constant factor.  Recall that an algorithm for the optimization version of {\cwc} is an $\alpha$-approximation algorithm if it always outputs either a colorful walk cover of size at most $\alpha\cdot k$ or {\sc No}, where if the input instance is a yes-instance, then the algorithm cannot output {\sc No}. We now state this result formally.

\begin{restatable}{theorem}{fptapprox}\label{thm:fpt_approx}
{\sc Colorful Walk Cover} admits an $\ell^{\Oh(k+\ell)} 2^{\Oh(k^2)}n^{\Oh(1)}$ time \FPT-approximation algorithm with approximation ratio $2$.
\end{restatable}

We obtain our 2-approximation algorithm (Theorem~\ref{thm:doctappx}) for the optimization version of {\sc DOCT} as a consequence of this result.
%
%
%
%
The rest of this section is therefore dedicated to proving Theorem \ref{thm:fpt_approx}. The crux of our approximation algorithm is an {\FPT}  algorithm (see Lemma~\ref{lem:restricted_fpt}) for the following variant of the {\cwc} problem.
\medskip

\defparproblem{{\sc Restricted Colorful Walk Cover (Restricted CWC)}}{A digraph $D$, function $\sigma:A(D)\to S_\ell$,  integer $k$ and a colorful walk cover $W\subseteq V(D)$ of size at most $2k+1$ such that $D[W]$ is strongly connected.}{$k,\ell$.}{Does there exist a colorful walk cover  of size at most $k$ which is disjoint from $W$? }

\begin{restatable}{lemma}{restrictedfpt}\label{lem:restricted_fpt}
{\sc Restricted CWC} has an algorithm running in time $\ell^{\bigoh(k+\ell)}   2^{\bigoh(k^2)}n^{\bigoh(1)}$.
\end{restatable}

Note that for an instance $(D,\sigma,k,W)$ of {\restrictedcwc}, we refer to a set $S$ which is a colorful walk cover of size at most $k$ and disjoint from $W$, as a  \emph{solution} for this instance.
The rest of the section is organized as follows. In the first subsection, we recall the notions of \emph{shadows} and {\em shadow covering} and describe the relation between the {\sc Colored Walk Cover} problem and the well-studied {\sc Node Unique Label Cover} ({\sc Node ULC}) problem. This relation will play an important role in our algorithm for {\restrictedcwc}.
 The second subsection is dedicated to a problem-specific {\em shadow removal} subroutine and the full  proof of Lemma~\ref{lem:restricted_fpt}. In the final subsection, we prove Theorem~\ref{thm:fpt_approx} using Lemma~\ref{lem:restricted_fpt}.

\newcommand{\doubling}{{\sf Doubling}}

\subsection{Shadow Covering and Connections to {\sc Node Unique Label Cover}}

We begin with the following lemma which  shows that any subgraph which is a colorful walk with respect to some vertex in it is in fact a colorful walk with respect to every vertex in it.  
	
	\begin{lemma}
		Let $(D,\sigma)$ be a labeled digraph and $H$ be a strongly connected subgraph of $D$. Then, $H$ is a colorful walk if and only if it is a $v$-colorful walk for some $v\in V(H)$.
	\end{lemma}

	\begin{proof}
		The forward direction is trivial and hence we now argue the converse direction. That is, suppose that $H$ is a $v$-colorful walk for some $v\in V(H)$. Let $u\in V(H)$ be a vertex distinct from $v$ and $i\in [\ell]$. In order to prove the lemma, it suffices to demonstrate the presence of a closed $u$-walk $Q$ in $H$ such that $\sigma(Q)(i)\neq i$.
		
		Let $P_1$ and $P_2$ be arbitrary $u$-$v$ and $v$-$u$ paths in $D$ respectively.  Let $i'=\sigma(P_1)(i)$ and let $W$ be an arbitrary $v$-walk  in $H$ such that $\sigma(W)(i')\neq i'$. Since $H$ is a $v$-colorful walk, $W$ exists. Now, consider the two closed $u$-walks $W_1=P_1+P_2$ and $W_2=P_1+W+P_2$ which are both contained in $H$. Observe that $\sigma(W_1)(i)$ and $\sigma(W_2)(i)$ are distinct and at least one of these is distinct from $i$. This completes the proof of the lemma.   
		\end{proof}

	We now define an auxiliary digraph which is the natural directed version of a graph first defined by Loksthanov et al.~\cite{LokshtanovRSULC16}. We then argue that it is possible to test in polynomial time whether for a given labeled digraph $(D,\sigma)$, vertices $u,v\in V(D)$ and $\alpha,\beta\in [\ell]$,  $D$ has a $u$-$v$ walk which `maps' $\alpha$ to $\beta$. This subroutine will be used in several places in our algorithm, including in the step where we test whether the input has a colorful walk.

	\begin{definition}
	With every labeled digraph $(D,\sigma)$, we associate an auxiliary digraph $H_{D,\sigma}$ (we ignore the reference to $\sigma$ when clear from the context) which is  defined as follows. The vertex set of $H_{D}$ is $\{v_i|v\in V(D), i\in [\ell]\}$.  
	 The arc set of $H_D$ is defined as follows. For every arc $a=(u,v)$ and for every $i\in [\ell]$, we have an arc $(u_i,v_{\sigma(a)(i)})$. That is, we add an arc from $u_i$ to $v_j$ where $j$ is the image of $i$ under the permutation $\sigma(a)$.
	\end{definition}

	\begin{lemma}\label{lem:auxiliary_equivalence} Let $(D,\sigma)$ be a labeled digraph and let $u$ and $v$ be (not necessarily distinct) vertices in  $V(D)$.
	For every $i,j\in [\ell]$, there is a $u$-$v$ walk $W$ in $D$ such that $\sigma(W)(i)=j$ if and only if there is a $u_i$-$v_j$ path $P$  in the digraph $H_D$.
	\end{lemma}

	\begin{proof} The lemma follows by a straightforward induction on the length of the walk $W$ in the forward direction and that of the path $P$ in the converse direction. 
%
%
	\end{proof}
	
	Due to the above lemma, it is easy to verify whether a digraph $D$ has a colorful walk.

	\begin{observation}
		There is a polynomial time algorithm that, given a labeled digraph $(D,\sigma)$, decides whether $D$ has a colorful walk.
	\end{observation}

\begin{definition}
	Let $(D,\sigma)$ be a labeled digraph. We define by $\doubling(D)$ the labeled digraph $(D',\sigma')$ obtained as follows. Initially, $D'=D$ and $\sigma'=\sigma$. Then, for every $a= (u,v)\in A(D)$, we add an arc $a'\in (v,u)\in A(D')$ and set $\sigma'(a')=(\sigma(a))^{-1}$. 
	Finally, we remove duplicate arcs as follows. As long as there are arcs $a,a'=(u,v)\in A(D')$  such that $a\in A(D)$, $a'\in A(D')\setminus A(D)$ and $\sigma'(a)=\sigma'(a')$, we remove the arc $a'$. 
\end{definition}

	
	Observe that it is possible that $\doubling(D)=(D,\sigma)$. In fact such graphs will be of special interest to us.   
	The following lemma generalizes Proposition \ref{prop:underlyingBipartite}.

	\begin{lemma}\label{lem:iff_colorful_undirected} Let $(D,\sigma)$ be a labeled digraph, where $D$ is strongly connected. Then, $D$ has a colorful walk if and only if $\doubling(D)$ has a colorful walk.
		\end{lemma}

\begin{proof}
Since $\doubling(D)$ is a supergraph of $D$, it follows that if $D$ has a colorful walk then so does $\doubling(D)$. We now argue the converse.
Let $(D',\sigma')=\doubling(D)$ and suppose that $D'$ has a colorful walk. Let $v$ be an arbitrary vertex in this subgraph. We now argue that $D$ has a colorful walk as well. More specifically, we argue that $D$ has a $v$-colorful walk. 
By Lemma \ref{lem:auxiliary_equivalence}, it is sufficient to show that
for every $i\in [\ell]$, there is a $j\neq i$ such that $H_D$ has a $v_i$-$v_j$ path. On the other hand, Lemma \ref{lem:auxiliary_equivalence} also implies that for every $\alpha\in [\ell]$, there is a $\beta\neq \alpha$ such that $H_{D'}$ has a $v_\alpha$-$v_\beta$ path. Therefore, we fix $\alpha\in [\ell]$ and let $P$ be a $v_\alpha$-$v_\beta$ path in $H_{D'}$ where $\beta\neq \alpha$. Our objective now is to demonstrate the existence of a $\lambda\neq \alpha$ and a $v_\alpha$-$v_\lambda$ path in $H_D$.

 Let $x\in V(D')$ and $\gamma\in [\ell]$ be such that $x_\gamma$ is the \emph{last} vertex in the traversal of $P$ from $v_\alpha$ towards $v_\beta$ with the property that $H_D$ also has a  $v_\alpha$-$x_\gamma$ path $P_1$. If $x=v$ and $\gamma\neq \alpha$ then we are already done. Hence, we assume that this is not the case. More specifically, $x_\gamma\neq v_\beta$. However, observe that it could be the case that $x_\gamma$ is $v_\alpha$ itself. 

Let $y_\delta$ be the vertex that appears immediately after $x_\gamma$ in the traversal of $P$ starting from $v_\alpha$. Due to our choice of $x_\gamma$ it must be the case that $(y_\delta,x_\gamma)\in A(H_D)$ and there is no $x_\gamma$-$y_\delta$ path in $H_D$. However, since $D$ is strongly connected, it follows that there is a $\delta'\in [\ell]$ such that $H_D$ has a $v_\alpha$-$y_{\delta'}$ path, where $\delta'\neq \delta$. Finally, from the existence of the arc $(y_\delta,x_\gamma)$ in $H_D$, we infer the existence of an arc $(y_{\delta'},x_\rho)$, where $\rho\neq \gamma$. We have thus obtained a $v_\alpha$-$x_\gamma$ path and a $v_\alpha$-$x_\rho$ path in $H_D$, where $\gamma\neq \rho$.

Since $D$ is strongly connected, there is an $x$-$v$ path $P_2$ in $D$. Clearly, either $\sigma(P_2)(\gamma)\neq \alpha$ or $\sigma(P_2)(\rho)\neq \alpha$. We assume without loss of generality that it is the former and denote $\sigma(P_2)(\gamma)$ by $\lambda$. As a result, we obtain an $x_\gamma$-$v_\lambda$ path in $H_D$. Since we already have a $v_\alpha$-$x_\gamma$ path in $H_D$, we conclude that there is a $v_\alpha$-$v_\lambda$ path in $H_D$. This completes the proof of the lemma.	
\end{proof}

The {\sc Node Unique Label Cover} problem ({\sc Node ULC}) which was introduced in the parameterized complexity setting by Chitnis et al.~\cite{ChitnisCHPP16}, is a special case of the {\sc Colorful Walk Cover} problem where the input digraph $D$ has the property that $\doubling(D)=D$. There are several \FPT algorithms for  {\sc Node Unique Label Cover} parameterized by $\ell$ and $k$, with varying dependencies on the two parameters~\cite{ChitnisCHPP16,IwataWY16,LokshtanovRSULC16}. While any of these algorithms serve our purpose, we chose to use the algorithm of Iwata et al. which has the best dependence on the parameters albeit at the cost of a larger (but still polynomial) dependence on input size when compared to the other two.

\begin{proposition}{\rm \cite{IwataWY16}}
	\label{prop:nodeulc_is_easy}{\sc Node Unique Label Cover} can be solved in time $\bigoh(\ell^{2k}n^{\bigoh(1)})$.
\end{proposition}


%

It is straightforward to see that in an instance $(D,\sigma,k)$ of {\sc Node ULC}, we can \emph{forbid} any set of vertices to be part of the solution by simply making $k+1$ copies of it. That is, for a vertex $v$, we add new vertices $v_1,\dots, v_k$ and for every $i\in [k]$ and $(v,u)\in A(D)$ we add an arc $(v_i,u)$ and set $\sigma((v_i,u))=\sigma((v,u))$. This ensures that any inclusion-wise minimal set of size at most $k$ that covers all colorful walks in the resulting digraph must be disjoint from $v,v_1,\dots, v_k$. We will require this operation in the proof of Lemma \ref{lem:restricted_fpt} where the algorithm of Proposition \ref{prop:nodeulc_is_easy} will be used as a subroutine. Hence, we reformulate this proposition as follows.

\begin{lemma}\label{lem:nodeulc_is_easy}
There is an algorithm that, given an instance $(D,\sigma,k)$ of 
	{\sc Node Unique Label Cover} and a  set $W\subseteq V(D)$, runs in time $\bigoh(\ell^{2k}n^{\bigoh(1)})$ and either returns a solution disjoint from $W$ or correctly concludes that no such solution exists.
\end{lemma}

We now recall the notion of shadows  from~\cite{Chitnis:2012DSFVS}.


\begin{definition}{\rm \cite{Chitnis:2012DSFVS}} Let $D$ be a digraph and $T$ be a set of terminals. Let $X\subseteq V(G)$ be a subset of vertices. 

\begin{itemize}
	\item The forward shadow $f_{D,T}(X)$ of $X$ (with respect to $T$) is the set of vertices $v$ such
that $X$ is a $T$-$\{v\}$ separator in $D$.

\item The reverse shadow $r_{G,T}(X)$ of $X$ (with respect to $T$) is the set of vertices $v$ such
that $X$ is a $\{v\}$-$T$ separator in $D$.
\end{itemize}

The shadow of $X$ (with respect to $T$) is the union of $f_{G,T}(X)$ and $r_{G,T}(X)$.

\end{definition}

The notions of shadows and `shadowless solutions' have proved to be a key component of several \FPT algorithms for cut-problems \cite{Chitnis:2012DSFVS,ChitnisHM13,LokshtanovR12,LokshtanovM13}.

\begin{definition}
	Let $I=(D,\sigma,k,W)$ be an instance of {\restrictedcwc} and let $S$ be a solution for this instance. If the shadow of $S$ with respect to $W$ is empty, then we say that $S$ is a {\em shadowless solution}.
\end{definition}

Combining the notion of shadowless solutions with Lemma~\ref{lem:iff_colorful_undirected}, we make the following crucial observation which implies that if the shadow of a solution $S$ with respect to $W$ is empty, then the set $S$ is a colorful walk cover even for the labeled digraph $\doubling(D)$.

\begin{observation}\label{obs:shadowless_same_doubling}
	Let $(D,\sigma,k,W)$ be an instance of {\restrictedcwc} 
	and $S$ be a solution for this instance. Let $C$ be the unique strongly connected component of $D-S$ containing $W$. Then, $\doubling(C)$ does not have a colorful walk. 
\end{observation}

Because of Lemma \ref{lem:iff_colorful_undirected}, Proposition \ref{prop:nodeulc_is_easy} and Observation \ref{obs:shadowless_same_doubling}, our objective from now on is to transform a given instance $I$ of {\restrictedcwc} into an instance $I'$ such that if $I$ is a yes-instance then $I'$ has a shadowless solution and is a no-instance otherwise. This transformation has two steps -- (a) the shadow covering step and (b) the shadow removal step. For the shadow covering step, we will use a result of Chitnis et al.~\cite{Chitnis:2012DSFVS}. Building on the work of Marx and Razgon~\cite{MarxR14} and Chitnis et al.~\cite{ChitnisHM13}, they gave a generic method to compute a set which `covers' the shadow of a solution when dealing with cut-problems satisfying certain properties. In order to make this statement precise and describe how it applies in our setting, we begin with the following definition.


\begin{definition}{\rm \cite{Chitnis:2012DSFVS}} Let $D$ be a digraph and let $\cF=\{F_1,\dots, F_q\}$  be a set of subgraphs of $D$. An $\cF$-{\em transversal} is a set
of vertices of $D$ that intersects every $F\in \cF$. For a set $T\subseteq V(G)$, we say that $\cF$ is $T$-{\em  connected} if, for every $i\in [q]$, each vertex
of $F_i$ can reach some vertex of $T$ by a walk completely contained in $F_i$ and is reachable
from some vertex of $T$ by a walk completely contained in $F_i$.
\end{definition}

The notion of $T$-connectivity of a family $\cF$ is relevant to us because in an instance $(D,\sigma,k,W)$ of {\restrictedcwc} the set of all colorful walks is clearly $W$-connected. This is a consequence of the fact that $W$ intersects all colorful walks in $D$ and each colorful walk itself is a strongly connected subgraph of $D$. As a result, we can  utilise the following {\em shadow covering} lemma of Chitnis et al.

\begin{proposition}{\rm (Theorem 3.6 \cite{Chitnis:2012DSFVS})}\label{prop:original_shadow_removal} Let $D$ be a digraph and let $T\subseteq V(D)$. Given $T$ and $D$, we can construct a set $\{Z_1,\dots, Z_t\}$ with $t=2^{\bigoh(k^2)}\log ^2n$ in time $2^{\bigoh(k^2)}n^{\bigoh(1)}$ such that, for any set $F$ of $T$-connected subgraphs, if there exists an $\cF$-transversal of size at most $k$, then there is an $\cF$-transversal $X$ of size at most $k$ such that for at least one $i\in [t]$, we have 

\begin{itemize}
	\item $X\cap Z_i=\emptyset$, and 
	\item $Z_i$ covers the shadow of $X$ with respect to $T$.
\end{itemize}
	
\end{proposition}

We use the fact that the set of  colorful walks in an instance of {\restrictedcwc} is $W$-connected to reformulate Proposition~\ref{prop:original_shadow_removal} as follows, so that it is easier to invoke in our context.  

\begin{lemma}\label{lem:shadow_removal_cwc}
	Let $I=(D,\sigma,k,W)$ be an instance of {\restrictedcwc}. Given $I$, we can construct a set $\{Z_1,\dots, Z_t\}$ with $t=2^{\bigoh(k^2)}\log ^2n$ in time $2^{\bigoh(k^2)}n^{\bigoh(1)}$ such that if $I$ is a yes-instance, then there is a colorful walk cover $S$ of size at most $k$ such that for at least one $i\in [t]$, we have 

\begin{itemize}
	\item $S\cap Z_i=\emptyset$, and 
	\item $Z_i$ covers the shadow of $S$ with respect to $W$.
\end{itemize}

A set $Z_i$ with these two properties is called a {\em shadow cover} for $S$ with respect to $W$.
\end{lemma}

\subsection{Shadow-removal and \FPT Algorithm for {\restrictedcwc}}

We now proceed to the \emph{shadow-removal step}. This is the more problem-specific part of the approach from \cite{Chitnis:2012DSFVS}.   In this step, our objective is to \emph{remove} the vertices in a shadow cover for a solution $S$ in such a way that $S$ becomes a shadowless solution for the resulting instance and we do not create new, smaller solutions for this instance. In order to achieve this, we need to define an appropriate analogue of the \emph{torso} operation~\cite{Chitnis:2012DSFVS}. However, before we define this operation, we prove the following lemma which gives a  subroutine required in the definition of this operation.

\begin{lemma}\label{lem:find_permutations}
	Let $(D,\sigma)$ be a labeled digraph and let $Z\subseteq V(D)$. There is an algorithm that, given $(D,\sigma)$, $Z$ and a pair of (not necessarily distinct) vertices $u,v\notin Z$, runs in time $\bigoh(\ell^\ell)+n^{\bigoh(1)}$ and computes a set $\cX=\{\pi_1,\dots, \pi_r\}$ of permutations of $[\ell]$ such that for every    $\alpha,\beta\in [\ell]$, there is a $u$-$v$ walk $P$ with all internal vertices in $Z$ such that $\sigma(P)(\alpha)=\beta$ if and only if there is a permutation $\pi\in \cX$ such that $\pi(\alpha)=\beta$.	
\end{lemma}

\begin{proof}

We first construct the subgraph $D'=D[\{u,v\}\cup Z]$ with $\sigma'$ being the restriction of $\sigma$ to the arcs in $D'$. We then construct the digraph $H_{D',\sigma'}$.  In order to define $\cX$, we will now construct another auxiliary digraph $Q$ with vertex set $[\ell]$ from which we will then extract the set $\cX$.


For every $i,j\in [\ell]$ we test whether there is a $u_i$-$v_j$ path in $H_{D'}$ with all internal vertices in $\cZ=\bigcup_{v\in Z} \{v_1,\dots, v_\ell\}$, and if such a path exists, then 
%
%
%
	we add the arc $(i,j)$. This completes the construction of $Q$. Observe that $Q$ may contain self-loops.
	Before we define $\cX$, we prove the following property of $Q$. Recall that a cycle cover of a digraph is a set of vertex-disjoint cycles such that every vertex is part of some cycle. 	
	
	\begin{claim}\label{clm:cycle_cover}
		Every arc $a\in A(Q)$ is part of a cycle cover in $Q$.
			\end{claim}
			
			\begin{proof}
			Let $a=(\alpha,\beta)\in A(Q)$. 
			By the definition of $Q$, it must be the case that there is a $u_\alpha$-$v_\beta$ path in $H_{D',\sigma'}$ with all internal vertices in $\cZ$, which in turn implies the presence of a $u$-$v$ walk $P$ in $D$ such that all internal vertices of $P$ lie in $Z$ and $\sigma(P)(\alpha)=\beta$. Therefore the permutation $\sigma(P)$ contains a cycle of the form $(\dots \alpha$ $\beta \dots)$. As a result, the cycle cover of $Q$ which corresponds to $\sigma(P)$ contains the arc $a$. 
			This completes the proof of the claim.  
			\end{proof}

	We are now ready to define $\cX$. For every arc $a=(\alpha,\beta)\in A(Q)$, we pick an arbitrary cycle cover in $Q$ containing $a$ and call it ${\cal C}_a$. By the claim above, we know that such a cycle cover exists.	For every such  cycle cover ${\cal C}_a$, we define the permutation $\pi_a$ as the corresponding permutation of $[\ell]$. 
	Finally, we define $\cX=\{\pi_a|a\in A(Q)\}$.
%
%
%
	 This completes the construction of $\cX$ and we now argue that it satisfies the required properties.

	By Lemma~\ref{lem:auxiliary_equivalence}, we know that for every $i,j\in [\ell]$, there is a $u$-$v$ walk $P$  such that $\sigma(P)(i)=j$ if and only if there is a $u_i$-$v_j$ path $P'$ in $H_{D'}$. Furthermore, it is easy to see that $P$ has all internal vertices in $Z$ if and only if $P'$ has all internal vertices in $\cZ$. Now, suppose that for some $\alpha,\beta\in [\ell]$, there is a $u$-$v$ walk $P$ with all internal vertices in $Z$ such that $\sigma(P)(\alpha)=\beta$. Then, there is a $u_\alpha$-$v_\beta$ path in $H_{D'}$  with all internal vertices in $\bigcup_{v\in Z} \{v_1,\dots, v_\ell\}$, which, by the construction of the digraph $Q$ implies that $(\alpha,\beta)\in A(Q)$ and hence there is a cycle in $Q$ which contains the arc $(\alpha,\beta)$. Due to  Claim~\ref{clm:cycle_cover}, we conclude that there is a permutation $\pi\in \cX$ such that $\pi(\alpha)=\beta$. This completes the argument in the forward direction. The converse direction follows by retracing the above argument. Note that the required time is the time required to compute $Q$ plus the time required to compute the cycle covers for the arcs in $Q$. The first part takes polynomial time and the second can be achieved by simply enumerating all cycle covers of $Q$. Since $|Q|=\ell$, the second step only requires time $\bigoh(\ell^\ell)$, hence completing the proof of the lemma.
%
%
%
%
%
%
%
\end{proof}

We are now ready to define the {\sf labeled-torso} operation.

\begin{definition}\label{def:labeled_torso_define}
	Let $(D,\sigma)$ be a labeled digraph, $Z\subseteq V(D)$  and for every ordered pair $(u,v)\in (V(D)\setminus Z)^2$, let $\cX_{uv}$ be the set of permutations returned by the algorithm of Lemma~\ref{lem:find_permutations} on input $(D,\sigma)$, $Z$, and the pair $u,v$. We let $\cH$ denote the set $\{\cX_{uv}|(u,v)\in (V(D)\setminus Z)^2\}$. We denote by {\sf labeled-torso}$(D,Z,\cH)$ the labeled digraph $(D',\sigma')$ obtained from $D$ as follows. 
	
	\begin{itemize}
	\item Set $D'=D, \sigma'=\sigma$. 
		\item Delete the set $Z$.
		\item For every ordered pair $(u,v)$ in $(V(D)\setminus Z)^2$, add $q=|\cX_{uv}|$ arcs $a^{uv}_1,\dots, a^{uv}_q=(u,v)$ and for each $j\in [q]$, set $\sigma'(a^{uv}_j)=\pi_q$ where $\cX_{uv}=\{\pi_1,\dots, \pi_q\}$.
	\end{itemize}
\end{definition}

Having defined the {\sf labeled-torso} operation we proceed to show that it preserves all colorful walks.

\begin{lemma}\label{lem:torso_preserves_paths}
	Let $(D,\sigma),Z,\cH$ be as in Definition~\ref{def:labeled_torso_define} and let $(D',\sigma')$ = {\sf labeled-torso}$(D,Z,\cH)$. Then, the following statements hold. 
	\begin{itemize} 
	\item For every (not necessarily distinct) $u,v\in V(D)\setminus Z$ and $\alpha,\beta\in [\ell]$, if there is a $u$-$v$ walk $P$ in $D$ such that $\sigma(P)(\alpha)=\beta$ then there is a  $u$-$v$ walk $P'$ in $D'$ such that $\sigma'(P')(\alpha)=\beta$ and $V(P')=V(P)\setminus Z$.
	\item 	For every (not necessarily distinct) $u,v\in V(D')$ and $\alpha,\beta\in [\ell]$, if there is a $u$-$v$ walk $P'$ in $D'$ such that $\sigma'(P')(\alpha)=\beta$ then there is a  $u$-$v$ walk $P$ in $D$ such that $\sigma(P)(\alpha)=\beta$ and $V(P)\subseteq V(P')\cup Z$.
	\end{itemize}

\end{lemma}

\begin{proof} For the first statement, let $u,v\in V(D)\setminus Z$ and $\alpha,\beta\in [\ell]$ be such that there is a directed $u$-$v$ walk $P=z_1,\dots, z_t$ in $D$ with $\sigma(P)(\alpha)=\beta$. If $P$ is also present in $D'$, then we are done. Suppose that this is not the case and let $z_{i_1}$ and $z_{i_2}$ be a pair of consecutive vertices in $P$ which are not in $Z$ such that $i_1<i_2$ and if $i_1+1<i_2$, then the vertices $z_{i_1+1},\dots, z_{i_2-1}$ are all in $Z$. Let $Q_1=P[u,z_{i_1}]$, $Q_2=P[z_{i_1},z_{i_2}]$ and $Q_3=P[z_{i_2},v]$ be three subwalks of $P$. Let $x=z_{i_1}$,  $y=z_{i_2}$ and furthermore, suppose that $\gamma,\delta\in [\ell]$ are such that $\sigma(Q_1)(\alpha)=\gamma$, $\sigma(Q_2)(\gamma)=\delta$ and $\sigma(Q_3)(\delta)=\beta$.

 Observe that $Q_2$ is a walk with all internal vertices in $Z$. Then, by Lemma~\ref{lem:find_permutations}, there is a permutation $\pi\in \cX_{xy}$ such that $\pi(\gamma)=\delta$ and by the definition of {\sf labeled-torso}, there is an arc $a=(x,y)$ in $D'$ such that $\sigma'(a)(\gamma)=\delta$. Therefore, we replace the subwalk $Q_2$ with this arc $a$ and we do this for every such consecutive pair of vertices in $P$ which are not in $Z$ but  all vertices in between them are in $Z$. The  walk resulting from performing this replacement for every such pair  is  a $u$-$v$ walk $P'$ in $D'$ such that $\sigma'(P')(\alpha)=\beta$ and $V(P')=V(P)\setminus Z$. This completes the argument for the first statement.

 For the second statement, let $P'$ be a directed $u$-$v$ walk in $D'$ such that $\sigma'(P')(\alpha)=\beta$. If $P'$ is also in $D$, then we are done. Suppose that this is not the case and let $(x,y)\in A(D')\setminus A(D)$ be an arc in $P'$. Let $Q_1=P'[u,x], Q_2=P'[x,y],Q_3=[y,v]$ be three subwalks of $P'$, where $Q_2$ is in fact the arc $(x,y)$ which by our assumption is not in $D$. Furthermore, as earlier, let $\gamma,\delta\in [\ell]$ be such that $\sigma(Q_1)(\alpha)=\gamma$, $\sigma(Q_2)(\gamma)=\delta$ and $\sigma(Q_3)(\delta)=\beta$.
  
 By the definition of {\sf labeled-torso} it must be the case that there is a permutation $\pi\in \cX_{xy}$ such that $\pi(\gamma)=\delta$. Furthermore, by Lemma~\ref{lem:find_permutations}, we know that this can happen only when there is an $x$-$y$ walk $Q_2'$ in $D$ such that $\sigma(Q_2')(\gamma)=\delta$ and all internal vertices of $Q_2'$ lie in $Z$. Therefore, we replace  the arc $(x,y)$ with the  $x$-$y$ walk $Q_2'$ which is contained in $D$ and we do this for every arc in $P'$ which is not in $A(D)$. The result is clearly a  $u$-$v$ walk $P$ in $D$ such that $\sigma(P)(\alpha)=\beta$ and $V(P)\subseteq V(P')\cup Z$. This completes the proof of the lemma.
\end{proof}

\begin{lemma}\label{lem:shadow_removal_cwc_II}
	Let $I=(D,\sigma,k,W)$ be an instance of {\restrictedcwc} and let $Z\subseteq V(D)\setminus W$ be such that if $I$ is a yes-instance, then it has a solution $S$ for which $Z$ is a shadow-cover. There is an algorithm that, given $I$ and $Z$ runs in   time $\bigoh(\ell^\ell n^{\bigoh(1)})$ and returns an instance $I'=(D',\sigma',W,k)$ such that if $I$ is a no-instance, then $I'$ is a no-instance and if $I$ is a yes-instance then $I'$ is a yes-instance with  a shadowless solution.
\end{lemma}

\begin{proof}
Let $\cH$ be the family of sets from Definition~\ref{def:labeled_torso_define}. Since $\cH$ can be computed by invoking the algorithm of Lemma~\ref{lem:find_permutations} for every pair of vertices in $V(D)\setminus Z$ , it follows that the time required to compute $\cH$ is $\bigoh(\ell^\ell n^{\bigoh(1)})$. Let $(D',\sigma')=$ {\sf labeled-torso}$(D,Z,\cH)$ and $I'=(D',\sigma',k,W)$. We now argue that $I'$ satisfies the required properties.

We first argue that if $I$ is a no-instance then $I'$ is a no-instance. In order to do so, we argue that any colorful walk cover of $D'$ is also a colorful walk cover of $D$. Suppose that this is not the case and let $S$ be a colorful walk cover of $D'$ which is not a colorful walk cover of $D$. Then, $D-S$ has a colorful walk $H$. If $H$ is contained entirely in $Z$, then it contradicts our assumption that $Z$ is a shadow-cover for $S$. Hence, we may assume that $H$ has at least one vertex outside $Z$, call it $h$. Now, Lemma~\ref{lem:torso_preserves_paths} implies that for every $h$-walk $P$ in $D$ and $\alpha,\beta$ such that $\sigma(P)(\alpha)=\beta$, there is an $h$-walk $P'$  in $D'$ such that $\sigma'(P')(\alpha)=\beta$ and $V(P')=V(P)\setminus Z$. Since $P$ is disjoint from $S$, we conclude that $P'$ is also disjoint from $S$. But this implies that there is an $h$-colorful walk in $D'$ which is disjoint from $S$, a contradiction.

%

In the converse direction, we argue that if $I'$ is a no-instance then $I$ is a no-instance. In order to do so, we argue that any colorful walk cover of $D$ disjoint from $Z$ is a colorful walk cover of $D'$ disjoint from $S$. Suppose that this is not the case and let $S$ be a colorful walk cover of $D$ disjoint from $Z$ which is not a colorful walk cover of $D'$. Then, there is a vertex $h\in V(D')$ and a subgraph $H$ which is an $h$-colorful walk in $D'-S$.

By Lemma~\ref{lem:torso_preserves_paths}, we know that for every $u$,$v$,$\alpha,\beta$ and a $u$-$v$ walk $P'$ in $D'$ such that $\sigma'(P')(\alpha)=\beta$, there is a $u$-$v$ walk $P$ in $D$ such that $\sigma(P)(\alpha)=\beta$ and $V(P)\subseteq V(P')\cup Z$. Since $P'$ is disjoint from $S$ and $S$ is disjoint from $Z$, it follows that $P$ is also disjoint from $S$. As a result, we infer the presence of an $h$-colorful walk in $D-S$ as well, a contradiction.

Finally, observe that due to Lemma~\ref{lem:torso_preserves_paths},  every pair of vertices  in the same strongly connected component as $W$ in $D-S$ remain in the same strongly connected component as $W$ in $D'-S$. Furthermore, $Z$ covers all vertices in the shadow of $S$ with respect to $W$ in $D$ and $V(D')=V(D)\setminus Z$. Therefore, the shadow of $S$ with respect to $W$ in $D'$ is empty. This completes the proof of the lemma.	
	\end{proof}

We are now ready to complete the proof of Lemma~\ref{lem:restricted_fpt}.

\restrictedfpt*

\begin{proof} Let $I=(D,\sigma,k,W)$ be the given instance of {\restrictedcwc}. We first execute the algorithm of Lemma~\ref{lem:shadow_removal_cwc} on this instance and obtain sets $\{Z_1,\dots, Z_t\}$, where $t=2^{\bigoh(k^2)\log^2 n}$. For every $i\in [t]$, we execute the algorithm of Lemma~\ref{lem:shadow_removal_cwc_II} on input $I$ and $Z_i$ to obtain the instance $I_i'=(D_i',\sigma_i',k,W)$. Finally, for each $i\in [t]$, we execute the {\sc Node Unique Label Cover} algorithm of Lemma~\ref{lem:nodeulc_is_easy} on input $I''_i=(\doubling(D_i',\sigma_i'),k)$ and the set $W$ to either compute a colorful walk cover of size at most $k$ disjoint from $W$ or correctly conclude that no such set exists. 
Finally, if for any $i\in [t]$, the solution to $I''_i$ is not {\sc NO}, then we return the computed set as the solution for the given instance of {\restrictedcwc}. 

The correctness and the claimed bound on the running time of this algorithm follow from those of Lemma~\ref{lem:shadow_removal_cwc}, Lemma~\ref{lem:shadow_removal_cwc_II}, Lemma~\ref{lem:nodeulc_is_easy} and Observation~\ref{obs:shadowless_same_doubling}. This completes the proof of the lemma.
\end{proof}

\subsection{The Parameterized Approximation for {\cwc}}

We are finally ready to complete the proof of Theorem~\ref{thm:fpt_approx}. For the sake of completeness, we restate it here.

\fptapprox*

At the highest level, our algorithm utilises the \emph{iterative compression} technique introduced by Reed, Smith and Vetta \cite{ReedSV04} in order to prove the fixed-parameter tractability of the {\sc Odd Cycle Transversal} problem on undirected graphs. It has subsequently gone on to become a fundamental tool in the fpr-algorithmist's toolbox. We now proceed to provide a description of the highest level of our algorithm which uses this technique.

Given an instance $(D,\sigma,k)$ of {\cwc}, where $V(D)=\{v_1,\dots,v_n\}$, we define a labeled graph $(D_i,\sigma_i$ where $V_i=\{v_1,\dots,v_i\}$, $D_i=D[V_i]$ and $\sigma_i$ is the restriction of $\sigma$ to $V_i$. We iterate through the instances $(D_i,\sigma_i,k)$ starting from $i=2k+1$ and for the $i^{th}$ instance, with the help of a \emph{known} solution $\hat S_i$ of size at most $2k+1$ we either correctly conclude that the $i^{th}$ instance has no colorful walk cover of size at most $k$ or try to find a colorful walk cover $S_i$ of size at most $2k$, i.e, a 2-approximation. This problem, which is known as the \emph{compression} problem is formally defined as follows.\\

\defparproblem{\cwcomp}{$(D,\sigma,k,\hat S)$ where $ \hat S$, a colorful walk cover of size at most $2k+1$.}{$k,\ell$}
{Does there exist a colorful walk cover of size at most $k$ for this instance?}

\vspace{10 pt}
\noindent
Our algorithm for the {\cwc} problem comprises of `solving' at most $n$ instances of the {\cwcomp} problem. Henceforth, in this context, we use `solving' to also mean obtaining a 2-approximate solution.  Let  $I_i=(D_i,\sigma,k,\hat S_i)$ be the $i^{th}$ instance of {\cwcomp}. Clearly, the set $V_{2k+1}$ is a solution of size at most $2k+1$ for the instance $I_{2k+1}$. It is also easy to see that if $S_{i-1}$ is a colorful walk cover of size at most $2k$ for instance $I_{i-1}$, then the set $S_{i-1}\cup \{v_i\}$ is a colorful walk cover of size at most $2k+1$ for the instance $I_i$. We use these two observations to initiate the iteration with the instance $(D_{2k+1},\sigma,k,\hat S_{2k+1}=V_{2k+1})$ and either compute a colorful walk cover of size at most $2k$ for this instance or correctly conclude that there is no colorful walk cover of size at most $k$. If there is such a solution $S_{2k+1}$, then we set $\hat S_{2k+2}=S_{2k+1}\cup \{v_{2k+2}\}$ and try to compute a colorful walk cover of size at most $2k$ for the instance $I_{k+2}$ and so on. If, on the other hand during any iteration, the corresponding instance is found to {\em not} have a colorful walk cover of size at most $k$, then it implies that the original instance is  a {\No} instance. Since the only way we proceed in the iteration is by computing a 2-approximate colorful walk cover $S_i$ for the instance $I_i$ of {\cwcomp}, the required 2-approximate colorful walk cover for the original input instance will be $S_{n}$. Since there can be at most $n$ iterations, the total time taken is bounded by $n$ times the time required to solve the {\cwcomp} problem.
We now discuss how to solve the {\cwcomp} problem by reducing it to a bounded number of instances of the {\restrictedcwc} problem. However, before we proceed, we need the following definition and proposition (see \cite{CyganFKLMPPS15}).

\begin{definition}
	Let $D$ be a digraph and $\cX=\{X_1,\dots, X_r\}$ be a set of disjoint vertex sets of $D$. A set $S\subseteq V(D)\setminus \bigcup_{i\in [r]}X_i$ is called an $\cX$-{\em skew separator} if  for every $1\leq i<j\leq r$, there is no directed $X_i$-$X_j$ path in $D-S$.
\end{definition}

\begin{proposition}{\rm \cite{CyganFKLMPPS15}}\label{prop:skew-sep} There is an algorithm that, given a digraph $D$, a set $\cX=\{X_1,\dots, X_r\}$ of disjoint vertex sets and an integer $k$, runs in time $\bigoh(4^k n^{\bigoh(1)})$ and either returns an $\cX$-skew separator of size at most $k$ or correctly concludes that one does not exist.	
\end{proposition}

\begin{lemma}\label{lem:skew_sep_exists}
	Let $(D,\sigma,k,\hat S)$ be an instance of {\cwcomp} and let $S$ be a solution for this instance. There exists an ordered partition $\cW=\{W_1,\dots, W_r\}$ of $\hat S\setminus S$ such that $S$ is a $\cW$-skew separator in $D$.
\end{lemma}

\begin{proof}
Let ${\cal C}=\{C_1,\dots, C_t\}$ be the set of strongly connected components of $D-S$ such that for every $1\leq i<j\leq t$, there is no $C_i$-$C_j$ path in $D$. Let ${\cal C}'=\{C_{i_1},\dots, C_{i_r}\}$ be the subset of $\cal C$ comprising strongly connected components intersecting $\hat S\setminus S$, where $i_{j}<i_{j'}$ for every $1\leq j<j'\leq r$. For each $j\in [r]$, let $W_j=(\hat S\setminus S)\cap C_{j}$. From the definitions of ${\cal C}$ and ${\cal C'}$, it follows that there is no $W_i$-$W_j$ path in $D-S$ for any $1\leq i<j\leq r$, implying that $S$ is a $\cW$-skew separator, where $\cW=\{W_1,\dots, W_r\}$. This completes the proof of the lemma.	
\end{proof}

We refer to the unique partition $\cW$ in the proof of the above lemma, as the partition of $\hat S\setminus S$ which \emph{respects} $S$.

\begin{definition}\label{def:skew_sep_empty}
	Let $I=(D,\sigma,k,\hat S)$ be an instance of {\cwcomp} and let $\cW=\{W_1,\dots, W_r\}$ be a partition of $\hat S$ such that  $D$ has a $\cW$-skew separator of size 0. Suppose that if $I$ is a yes-instance then $\cW$ is the partition respecting a solution for $I$  disjoint from $\hat S$.   Then $I$ is called a $\cW$-{\em nice instance}. 
	\end{definition}

\begin{lemma}
	\label{lem:skew_sep_empty}
	There is an algorithm that, given an instance $I=(D,\sigma,k,\hat S)$ of {\cwcomp} and a partition $\cW=\{W_1,\dots, W_r\}$ of $\hat S$ such that $I$ is a $\cW$-nice instance, runs in time $\ell^{\bigoh(k+\ell)} 2^{\bigoh(k^2)} n^{\bigoh(1)}$ and either returns a colorful walk cover of size at most $k$ or correctly concludes that one does not exist.
\end{lemma}

\begin{proof} 
Let $S$ be a solution disjoint from $\hat S$ such that $\cW$ is the partition which respects $S$.
For each $i\in [r]$, let $C_i$ denote the unique strongly connected component of $D$ which contains the set $W_i$. Since $I$ is a $\cW$-nice instance, it has a $\cW$-skew separator of size 0. Therefore, for $1\leq i<j\leq r$, $C_i$ and $C_j$ are distinct. For every $i\in [r]$, we now define the labeled digraph $(D_i,\sigma_i)=(D[C_i],\sigma|_{C_i})$ and the instance $I_i=(D_i,\sigma_i,k,W_i)$ of {\restrictedcwc}.

We execute the algorithm of Lemma~\ref{lem:restricted_fpt} for each $i\in [r]$ and let $L_i$ denote the result of the execution on the instance $I_i$, where $L_i$ can either denote {\sc No} or a smallest solution for the instance $I_i$.
 If for any $i\in [r]$,  $L_i$ is {\sc No}, then we return that $I$ is a no-instance. This is correct because $I_i$ is a sub-instance of $I$. On the other hand, suppose that for each $i\in [r]$, $L_i$ denotes a vertex set which we know is a \emph{smallest} colorful walk cover of $D_i$ of the required kind. Since the digraphs $D_1,\dots, D_r$ are vertex-disjoint and every colorful walk is contained in one of these digraphs, we conclude that $S'=\bigcup_{i\in [r]}S_i$ is a smallest colorful walk cover for the instance $I$. Therefore, if $S'$ is larger than $k$ then we return {\sc No} and otherwise we return $S'$. The running time of this algorithm is dominated by the time required for at most $2k+1$ invocations of the algorithm of Lemma~\ref{lem:restricted_fpt}, proving the stated bound on the running time. This completes the proof of the lemma.
\end{proof}


\begin{definition}
Let $(D,\sigma)$ be a labeled digraph.	A {\em consistent labeling} of $D$ is a function $\Gamma:V(D)\to[\ell]$ such that for every arc $a=(u,v)\in A(D)$, $\sigma(a)(\Gamma(u))=\Gamma(v)$. 
For a set $X\subseteq V(D)$ and function $\chi:X\to [\ell]$, we say that $\chi$ is an {\em extendible consistent labeling} of $D$ if $D$ has a consistent labeling $\Gamma$ such that $\Gamma|_X=\chi$.
\end{definition}

%
%

\begin{proposition}{\rm \cite{ChitnisCHPP16}}\label{prop:undirected_consistent_labeling}
	Let $(D,\sigma)$ be a strongly connected labeled digraph such that $\doubling(D)$ $=D$. Then, $D$ has a consistent labeling if and only if it has no colorful walks. 
\end{proposition}

Combining Proposition~\ref{prop:undirected_consistent_labeling} and Lemma~\ref{lem:iff_colorful_undirected}, we make the following observation. 

\begin{observation}\label{obs:directed_consistent_labeling}
	A strongly connected labeled digraph has a consistent labeling if and only if it has no colorful walks. 
\end{observation}

The above observation implies that every strongly connected component of $D-S$ has a consistent labeling. We now define the operation of {\em bundling} a set of vertices as follows.

\begin{definition}
Let $(D,\sigma)$ be a labeled digraph. Let $X\subseteq V(D)$ and $\Gamma:X\to [\ell]$. We denote by $\bundle(D,\sigma,X,\Gamma)$ the labeled digraph $(D',\sigma')$ obtained from $D$ as follows. Initially, $D'=D$, $\sigma'=\sigma$. For every pair $x_1,x_2\in X$, we pick an arbitrary permutation $\pi\in S_\ell$ such that $\pi(\Gamma(x_1))=\Gamma(x_2)$ and we add an arc $a=(x_1,x_2)$ with $\sigma'(a)=\pi$.  
\end{definition}

Note that this operation is essentially the same as identifying the vertices of $X$ to get a single new vertex and then updating the labels on the arcs adjacent to the resulting new vertex in a certain way specified by the function $\Gamma$. However, we define it in this way because it simiplies the presentation in the rest of the section.

\begin{lemma}\label{lem:bundling_is_fine}
	Let $(D,\sigma)$ be a strongly connected labeled digraph with no colorful walks and let  $\Gamma:V(D)\to [\ell]$ be a consistent labeling. Then, for any $X\subseteq V(D)$, 
	\begin{itemize}\item $\Gamma$ is a consistent labeling for the labeled digraph $D'=\bundle(D,\sigma,X,\Gamma|_{X})$ and \item $D'$ does not contain a colorful walk. 
		
	\end{itemize}

\end{lemma}

\begin{proof} 
The first statement is a simple consequence of the fact that $\Gamma$ is already a consistent labeling of $D$ and the newly added arcs  clearly do not violate the condition required for $\Gamma$ to be consistent. 
The second statement of the lemma follows from the first statement and 
 Observation~\ref{obs:directed_consistent_labeling}.
\end{proof}

We are now ready to present our algorithm that `solves' the {\cwcomp} problem. That is, an algorithm that, if the given instance is not a {\sc No} instance, returns a colorful walk cover whose size is at most twice the given budget.

\begin{lemma}
There is an algorithm that, given an instance $(D,\sigma,k,\hat S)$ of 
	{\cwcomp}, runs in time $\ell^{\bigoh(k)} 2^{\bigoh(k^2)}n^{\bigoh(1)}$ and either computes a colorful walk cover of size at most $2k$ or correctly concludes that there is no colorful walk cover of size at most $k$.
%
\end{lemma}

\begin{proof}
Let $S$ be a solution for the given instance of {\cwcomp} and let $Y=S\cap \hat S$. Let 
$\cW_Y=\{W_1,\dots, W_r\}$ be an ordered partition of $\hat S\setminus Y$ which respects $S$.

We first guess the set $Y$ and the partition $\cW_Y$. Furthermore, we guess a function $\Gamma:\hat S\setminus S\to [\ell]$ such that for every $i\in [r]$, the restriction $\Gamma|_{W_i}$ is an extendible consistent labeling of the strongly connected component of $D-S$ containing $W_i$.  Due to Observation~\ref{obs:directed_consistent_labeling}, such a $\Gamma$ must exist. Furthermore, since $|\hat S|$ is bounded by $2k+1$, there are $\ell^{\bigoh(k)}$ choices for $\Gamma$. We now construct a new digraph by `bundling' each set in $\cW$.
%
This is done as follows.
For each $i\in [r]$, we define the graph 
$(D_i,\sigma_i)=\bundle(D_{i-1},\sigma_{i-1},\Gamma_{W_i},W_i)$, where $(D_0,\sigma_0)=(D,\sigma)$. Clearly, the strongly connected components of $D_r-S$ are the same as those of $D-S$ and by Lemma~\ref{lem:bundling_is_fine}, it follows that $\Gamma|_{W_i}$ is still an extendible consistent labeling for the strongly connected component of $D_r-S$ containing $W_i$. Furthermore, for {\em any} set $X$ disjoint from $\hat S\setminus S$, for every $i\in [r]$, the vertices in $W_i$ remain in the same strongly connected component of $D_r-X$.

We now execute the algorithm of Lemma~\ref{lem:skew_sep_empty} to compute a $\cW$-skew separator of size at most $k$. If no such separator exists, then by Lemma~\ref{lem:skew_sep_exists}, we may correctly conclude that the instance $I$ is a no-instance and hence we return the same. On the other hand, let $X$ be a $\cW$-skew separator of size at most $k$ and let $D'=D_r-X$ with $\sigma'$ being the associated labeling function. Observe that there is a $\cW$-skew separator of size 0 in $D'$ and $D'$ now has a colorful walk cover $S\setminus X$ of size at most $k$ such that the partition $\cW$ respects $S\setminus X$.

 We now construct the instance $(D',\sigma',k,\hat S\setminus (X\cup Y))$ of {\cwcomp} which as we have already argued, is a $\cW$-nice instance. We then execute the algorithm of Lemma~\ref{lem:skew_sep_empty} to compute a colorful walk cover $Z$ of size at most $k$ for $D'$. If no such set exists, then $I$ is a no-instance and we return the same. Otherwise, the set $X\cup Z$ is a colorful walk cover for $D$ of size at most $2k$. Hence, we return $X\cup Z$. This completes the description of the algorithm. The bound on the running time follows from that of Lemma~\ref{lem:skew_sep_empty} and the fact the number of invocations of the algorithm of this lemma is bounded by the product of the number of choices for $Y$, $\Gamma$ and $\cW$. Since this is bounded by 
 $\ell^{\bigoh(k)}2^{\bigoh(k\log k)}$, the running time follows and this completes the proof of the lemma.   
	\end{proof}

\section{Conclusion}

Our results on {\sc Directed Odd Cycle Transversal} raise a few natural questions. 

\begin{itemize}\item  
The first question is whether one can improve on the approximation factor of 2 in Theorem~\ref{thm:doctappx}  or strengthen the inapproximability result in Theorem~\ref{thm:doctinappx}   to show that even such an improvement is unlikely.

\item Secondly, although  Theorem~\ref{thm:doctHard1} implies that {\doct} is unlikely to have a kernel of \emph{any size}, our  \FPT-approximation algorithm implies that {\doct} does have a 2-approximate kernel of exponential size (see Proposition 3.2,~\cite{LokshtanovPRS16}). Therefore, an exciting new challenge related to {\doct} is to determine whether it has a $c$-approximate kernel of \emph{polynomial size} for some constant $c$ and if so, to find the  smallest such constant. Note that Theorem~\ref{thm:doctinappx} also rules out a $(1+\epsilon)$-approximate kernel (for some $\epsilon>0$) of \emph{any} size for {\doct}. 
\end{itemize}

We conclude by pointing out that the parameterized complexity of the {\sc Directed Multicut} problem where the number of terminal pairs is  3, remains open. As was the case for {\sc DOCT}, it is quite likely that an  {\FPT} algorithm or a {\sf W}-hardness proof  for this problem would require new insights into the structure of directed cuts.


\myparagraph{Acknowledgements.} The authors would like to thank Micha{\l} W{\l}odarczyk for enlightening discussions on the DOCT problem.

\bibliographystyle{siam}
\bibliography{references}

\end{document}